\documentclass[
    pra,
    amsmath,
    amssymb,
    amsthm,
    reprint,
    twocolumn,
    superscriptaddress,
    nofootinbib, 
    longbibliography
    ]{revtex4-2}

\usepackage{graphicx}
\usepackage[colorlinks = true, linkcolor=teal, citecolor=teal, urlcolor  = teal]{hyperref}

\usepackage{color}
\usepackage[dvipsnames]{xcolor}
\usepackage{mathtools}
\usepackage{titletoc} % for partial table of contents
\usepackage{easy-todo} % toc of appendix breaks if removed

\usepackage{float}
\usepackage{array}
\usepackage{verbatim}

\usepackage[utf8]{inputenc}
\usepackage[T1]{fontenc}
\usepackage{etoolbox}

\usepackage{anyfontsize}        % Can help to remove certain font size warnings
\usepackage{microtype}          % Helps eliminate protrusion into the margins
\usepackage[export]{adjustbox}  % for alignement of graphics in equations

\usepackage{newtxtext}
\usepackage{newtxmath}

\usepackage{dsfont}         % used for instance for the identity
\usepackage{braket}         % for braket notation
\usepackage{soul}           % smarter underline \ul{}
\usepackage{MnSymbol}       % for \rrangle
\usepackage{pifont}         % for check marks and cross marks like \cmark and \xmark
\newcommand{\cmark}{{\color{ForestGreen}\ding{51}}}%
\newcommand{\xmark}{{\color{BrickRed}\ding{55}}}

\usepackage{makecell}
\usepackage{multirow}                   % for cells covering multiple rows or columns in tables
\usepackage{tabularx}                   % for fixed widths
\usepackage{algorithm}          % for algorithms
\usepackage{algpseudocodex}   

\usepackage{ragged2e}

\makeatletter
\g@addto@macro\bfseries{\boldmath}
\makeatother

\usepackage{cleveref}
\usepackage{orcidlink}
\newcommand{\MYhref}[3][blue]{\href{#2}{\color{#1}{#3}}}%

\usepackage{thmtools}
\usepackage{thm-restate}        % for restatable 

\makeatletter 
    
\renewcommand\onecolumngrid{
\do@columngrid{one}{\@ne}
\def\set@footnotewidth{\onecolumngrid}
\def\footnoterule{\kern-6pt\hrule width 1.5in\kern6pt}%
}

\renewcommand\twocolumngrid{
        \def\footnoterule{
        \dimen@\skip\footins\divide\dimen@\thr@@
        \kern-\dimen@\hrule width.5in\kern\dimen@}
        \do@columngrid{mlt}{\tw@}
}

\makeatother    
\theoremstyle{definition}
\newtheorem{proposition}{Proposition}%[section]
\newtheorem{definition}{Definition}

\newtheorem{theorem}{Theorem}%[section]
\newtheorem{corollary}{Corollary}

\newtheorem*{remark}{Remark}

\DeclareMathOperator{\poly}{poly}

\DeclareMathOperator*{\argmin}{arg\,min}

\newcommand{\II}{\mathds{1}}  
  
\newcommand{\CC}{\mathds{C}}  

\newcommand{\inner}[2]{\langle #1 \vert #2 \rangle}
\newcommand{\proj}[1]{\ket{#1}\!\!\bra{#1}}
\newcommand{\abs}[1]{\lvert #1 \rvert}
\newcommand{\norm}[1]{\lVert #1 \rVert}

\newcommand{\Abs}[1]{\left\lvert #1 \right\rvert}
\newcommand{\Norm}[1]{\left\lVert #1 \right\rVert}

\newcommand{\ketbra}[2]{\ket{#1}\!\!\bra{#2}}
\newcommand{\Tr}[2]{\operatorname{Tr}_{#1} \left[ #2 \right]}

\newcommand{\superket}[1]{\left| #1 \right\rrangle}
\newcommand{\superbra}[1]{\left\llangle #1 \right|}
\newcommand{\superketbra}[2]{\superket{#1}\!\!\superbra{#2}}

\newcommand{\choi}{\mathrm{C}}

\newcommand{\approxrep}[1]{\vphantom{#1}^\filledtriangledown \mspace{-5mu} #1}
\newcommand{\preapproxrep}[1]{\vphantom{#1}^\smalltriangledown \mspace{-5mu} #1}
\DeclareMathOperator{\ii}{i}
\DeclareMathOperator{\ee}{e}

\usepackage{twemojis}

\definecolor{checkgreen}{rgb}{0.3, 0.8, 0.3}

\definecolor{newgreen}{rgb}{0.1, 0.6, 0.05}

\definecolor{warnred}{rgb}{0.8, 0.15, 0.1}

\definecolor{sorange}{rgb}{0.8, 0.3, 0}

\definecolor{forestgreen}{rgb}{0.0, 0.7, 0.0}

\definecolor{yellows}{RGB}{255, 190, 0}
\definecolor{oranges}{RGB}{248, 98, 0}
\definecolor{reds}{RGB}{233, 0, 34}
\definecolor{purple1s}{RGB}{90, 43, 144}
\definecolor{purple2s}{RGB}{171, 0, 139}
\definecolor{pink1s}{RGB}{205, 0, 101}
\definecolor{pink2s}{RGB}{255, 0, 102}
\definecolor{blue1s}{RGB}{0, 78, 139}
\definecolor{blue2s}{RGB}{0, 160, 216}
\definecolor{blue3s}{RGB}{52, 205, 255}
\definecolor{green1s}{RGB}{0, 140, 56}
\definecolor{green2s}{RGB}{101, 205, 0}

\usepackage{tikz-network}
\newcommand{\tikzineq}[1]{\vcenter{\hbox{\begin{tikzpicture} #1 \end{tikzpicture}}}}

\begin{document}
\title{Classical simulation of noisy quantum circuits via locally entanglement-optimal unravelings}

\author{\MYhref[black]{https://orcid.org/0000-0002-9409-193X}{Simon~Cichy}}
\email{simon.cichy@fu-berlin.de}
\affiliation{Dahlem Center for Complex Quantum Systems, Freie Universit\"{a}t Berlin, 14195 Berlin, Germany}
	
\author{\MYhref[black]{https://orcid.org/0000-0002-8706-1732}{Paul~K.~Faehrmann}}
\affiliation{Dahlem Center 
for Complex Quantum Systems, Freie Universit\"{a}t Berlin, 14195 Berlin, Germany}

\author{\MYhref[black]{https://orcid.org/0000-0003-1626-2761}{Lennart~Bittel}}
\affiliation{Dahlem Center for Complex Quantum Systems, Freie Universit\"{a}t Berlin, 14195 Berlin, Germany}

\author{\MYhref[black]{https://orcid.org/0000-0003-3033-1292}{Jens~Eisert}}
\email{jense@zedat.fu-berlin.de}
\affiliation{Dahlem Center for Complex Quantum Systems, Freie Universit\"{a}t Berlin, 14195 Berlin, Germany}
\affiliation{Helmholtz-Zentrum Berlin f\"{u}r Materialien und Energie, Hahn-Meitner-Platz 1, 14109 Berlin, Germany}
\affiliation{Fraunhofer Heinrich Hertz Institute, 10587 Berlin, Germany}

\author{\MYhref[black]{https://orcid.org/0000-0003-3782-7602}{Hakop~Pashayan}}
\email{hakop.pashayan@foxconn.com}
\affiliation{Dahlem Center for Complex Quantum Systems, 
Freie Universit\"{a}t Berlin, 14195 
Berlin, Germany}
\affiliation{Hon Hai (Foxconn) Research Institute, Taipei
114699, Taiwan}

\date{\today}

\begin{abstract}
Classical simulations of noisy quantum circuits are instrumental to our understanding of the behavior of real-world quantum systems and the identification of regimes where one expects quantum advantage. In this work, 
we present a highly parallelizable tensor-network-based classical algorithm -- equipped with rigorous accuracy guarantees -- for simulating $n$-qubit quantum circuits with arbitrary single-qubit noise. Our algorithm represents the state of a noisy quantum system by a particular ensemble of matrix product states from which we stochastically sample. 
Each pure state evolved under a single qubit noise process is then represented by the ensemble of states that achieves the minimal average entanglement (the entanglement of formation) between the noisy qubit and the remainder.
This approach lets us use a more compact representation of the quantum state for a given accuracy requirement and noise level. For a given maximum bond dimension $\chi$ and circuit, our algorithm comes with an upper bound on the simulation error, runs in $\poly(n,\chi)$-time and improves upon related prior work (1) in scope: by extending analytic methods from the three commonly considered noise models to general single qubit noise (2) in performance: 
by deriving an exact, closed-form solution to the local entanglement minimization problem -- previously approached via variational heuristics -- thereby guaranteeing local optimality without numerical optimization overhead,
and (3) in conceptual contribution: by showing that the fixed unraveling used in prior work becomes equivalent to our choice of unraveling in the special case of depolarizing, dephasing and amplitude damping noise acting on a maximally entangled state.
\end{abstract}
\maketitle

\section{Introduction}

Quantum computers have the potential to outperform classical computers in certain structured computational problems. Significant effort has been directed towards delineating the boundary between those that admit an efficient classical 
solution and those that can be used to demonstrate quantum advantage. Some important results in this direction have focused on the specific structure of a problem necessary for achieving a quantum advantage~\cite{AaronsonAmbainis2014, bendavid2020symmetric,childs2020graph,bendavid_childs2020}.
Others have shown the hardness of efficient classical simulation~\cite{aaronsonComputationalComplexityLinear2013, Bremner2016IQP_hardness, bouland2018ConjugatedCliffordsHardness, pashayanEstimationQuantumProbabilities2020, Yoganathan2019, Oszmaniec2022FermionSamplingHardness, SupremacyReview} by relying on plausible complexity-theoretic assumptions and an anti-concentration property of certain classes of quantum circuits
\cite{SupremacyReview}.
However, the bulk of scientific effort has focused on 
the \emph{classical simulation of quantum computation}. 

Classical simulation has served as an invaluable tool for the prediction and study of quantum phenomena as well as the precise exploration of the quantum advantage boundary.
Efficient classical simulation algorithms for stabilizer circuits~\cite{gottesmanHeisenbergRepresentationQuantum1998, aaronsonImprovedSimulationStabilizer2004}, matchgate circuits~\cite{Valiant2001,Terhal2002Fermion} 
(also known as non-interacting fermions) and 
non-entangling circuits identify necessary resources for quantum computational speedup. Extensions of these classical simulation techniques beyond the polynomial runtime regime have established operationally meaningful quantifications of these resources. Prominent examples include the stabilizer rank~\cite{BravyiGosset2016Improved,Qassim2021improvedupperbounds}, the stabilizer extent~\cite{bravyiSimulationQuantumCircuits2019,Heimendahl2021stabilizerextentis} which has been generalized to the mixed state setting~\cite{seddonQuantifyingQuantumSpeedups2021}, extended to measures of non-matchgateness~\cite{ReardonSmith2024improvedsimulation,Dias2024classicalsimulation,reardonsmith2024fermionic} and measures related to entanglement such as a \emph{matrix 
product state's} (MPS) \cite{MPSReps,vidalEfficientClassicalSimulation2003, vidalEfficientSimulationOneDimensional2004} bond dimension
in a \emph{tensor network description} \cite{Orus_2019,RevModPhys.93.045003,AreaReview}.
As the quantum computational task moves away from the islands of classically efficiently simulable quantum circuits, these resources quantifiably bound the exponential runtime growth of classical simulators, breaking the direct exponential dependence on the number of qubits. 

Particularly relevant to the recent thrust of quantum technological development is the deep and fascinating nuance introduced by the presence 
of \emph{noise} in quantum computation. 
At sufficiently low noise levels, in principle, quantum error correction addresses the challenges posed by the inevitable presence of quantum noise, but it introduces significant overheads~\cite{Roads}. 
Without error correction, the stochastic nature of noise can diminish quantum computational resources such as non-stabilizerness (also known as magic)~\cite{Virmani2005, seddonQuantifyingQuantumSpeedups2021, garca2024pauli} and entanglement~\cite{Virmani2005}.

Growing evidence suggests that, at 
logarithmic depth or beyond, typical quantum circuits in certain families that are hard to efficiently classically simulate become efficiently classically simulable in the presence of local depolarizing noise 
\cite{francaLimitationsOptimizationAlgorithms2021, deshpandeTightBoundsConvergence2022, aharonovPolynomialtimeClassicalAlgorithm2022, schusterPolynomialtimeClassicalAlgorithm2024, meleNoiseinducedShallowCircuits2024a, nohEfficientClassicalSimulation2020}, so that an eventual quantum advantage would disappear. However, these results are limited in scope, both in the types of circuits and noise that they can accommodate. Indeed, they
apply to typical random quantum circuits, limiting the insight one can draw about what quantum circuit properties computationally challenge a classical approach. Also, beyond the standard depolarizing noise model, even less is known about how noise affects the hardness of classical simulation.

Generalization to the rich landscape of local noise models is crucial for transforming a fragmented patchwork of techniques into a unified and principled theory that delineates the classical-quantum boundary.
First, this landscape exhibits stark phenomenological differences.
Indeed, while local depolarizing noise appears to aid classical simulability, local coherent noise can hinder it\footnote{For example, Ref.~\cite{bennink2017} has tackled the non-trivial task of classically simulating stabilizer circuits in the presence of \emph{low levels} of single qubit coherent noise.}. Thus, in the absence of this generalization, a quantitative or even qualitative characterization of classical simulability in the presence of noise is incomplete. Second,  
such generalization can be invaluable in connecting the well-understood pockets of parameter space and facilitating the transfer of tools and observations between previously disconnected settings. 

In this work, we focus on and build on matrix product state-based unraveling simulators from an entirely new perspective. This is a classical simulation technique that is highly parallelizable and amenable to large-scale simulation~\cite{isakovSimulationsQuantumCircuits2021, Hintermueller}. Importantly, this approach makes explicit use of the fact that noise can decay entanglement, making it particularly well suited to simulation via tensor network techniques.

This work, and MPS-based unraveling simulators more generally, focus on utilizing the presence of noise to find a low-entanglement representation of the quantum state. This approach is underpinned by the observation that, in contrast to pure quantum states, a mixed quantum state, $\rho$, can be \emph{non-uniquely} decomposed into a convex combination of pure quantum states $\rho=\sum_i p_i \ketbra{\phi_i}{\phi_i}$ also known as an unraveling. In some sense, this degree of freedom in the decomposition of $\rho$ `grows' as the purity of $\rho$ decreases, and this can be utilized to reduce entanglement 
\footnote{An extreme but illustrative example is to note that the maximally mixed state on two qubits can be decomposed as a uniform mixture of four Bell states (each of which has maximal entanglement); however, it can also be decomposed as a uniform mixture of computational basis states (each of which has zero entanglement).}. However, this beautiful idea brings with it significant implementation challenges.  
In an $n$-qubit system, it is not obvious how to choose a suitable entanglement-based cost function that both lends itself to optimization and connects to meaningful properties of a classical simulator, such as runtime or simulation accuracy. This is exacerbated by the vast degree of freedom in the mixed state decompositions. Prior works have predominantly dealt with this complexity either by simplifying the system~\cite{vovkEntanglementOptimalTrajectoriesManyBody2022, kolodrubetzOptimalityLindbladUnfolding2023, chengEfficientSamplingNoisy2023, chenOptimizedTrajectoryUnraveling2024} or resorting to heuristic optimization algorithms~\cite{chenOptimizedTrajectoryUnraveling2024, vovkQuantumTrajectoryEntanglement2024, darabanNonunitarityMaximizingUnraveling2025}. These simplifications involve restricting consideration to specific noise models combined with simpler systems, such as translationally invariant systems or random circuits.
The alternative approach also suffers from significant drawbacks, as heuristic optimization algorithms such as gradient descent can be slow and lack a proof showing that local minima can be avoided. 

Our work substantially builds upon many of the significant benefits that MPS-based unraveling simulators already possess.
This leads to a local optimization problem: minimizing the entanglement between the target qubit for the noise channel and the rest of the system. While Ref.~\cite{chenOptimizedTrajectoryUnraveling2024} identified this objective and approached it using variational heuristics, we tackle the problem analytically. 
We make the crucial observation that an $n$-qubit mixed state generated in this way can be mapped, via an isometry, to a $2$-qubit mixed state where, by utilizing Ref.~\cite{woottersEntanglementFormationArbitrary1998}, we can compute closed-form solutions that easily map back to the provably entanglement-optimal unraveling of the $n$-qubit system. 
With this approach, we achieve two highly desirable and competing goals. 
First, and most crucially, we attain a highly flexible simulation algorithm where we do not need to compromise the generality of the noise model or impose restrictions on the class of circuits (except one-dimensional spatial locality). 
Second, we attain closed-form analytical solutions that not only make performance improvements upon the prior art but also lend themselves to analytical investigation that can aid future progress on this topic. 
We demonstrate this by using the conceptual insights gained from our analysis to connect to and generalize previous results for random circuits.

Our work makes a conceptually enlightening connection to other important works on this topic. Our entanglement-optimal mixed state decompositions depend not only on the single qubit noise channel and its target location but also on the state of the $n$-qubit system immediately prior to the action of the noise channel. We show that when this state is fixed to be Haar random, then our unraveling reduces to that of prior works ~\cite{chenOptimizedTrajectoryUnraveling2024, chengEfficientSamplingNoisy2023}. These works focus on random circuits and specific single-qubit noise models and claim optimality (with respect to the R\'enyi 2-quasientropy) in this setting. Additionally, we present another conceptually meaningful cost function that is minimized by our unraveling applied to Haar random states.
In a precise sense, this new cost function minimizes the unitarity of the Kraus decomposition, favoring Kraus operators that act like rank-1 operators that disentangle the target qubit from the rest of the system. This generalizes fixed unraveling strategies for random circuits to arbitrary single-qubit noise while unifying two conceptually very distinct approaches.

Our final significant contribution is to provide circuit and bond dimension-specific upper bounds on the trace distance error between the final state of the true system and that of any given MPS-based unraveling simulator. Our error bounds are not generic. Rather, they are sensitive to the MPS-based unraveling simulators and apply to the previous body of work on the topic~\cite{vovkEntanglementOptimalTrajectoriesManyBody2022, kolodrubetzOptimalityLindbladUnfolding2023, chengEfficientSamplingNoisy2023, vovkQuantumTrajectoryEntanglement2024, chenOptimizedTrajectoryUnraveling2024, darabanNonunitarityMaximizingUnraveling2025}. 
This contribution addresses a shortcoming of all MPS-based unraveling simulators and may be of potential interest beyond this setting.

This work is structured as follows.
In Section~\ref{sec:preliminaries}, we introduce key theoretical background covering unravelings, entanglement measures and matrix product states.
In Section~\ref{sec:unraveling}, we focus on performance guarantees. There, we characterize a broad class of simulators, state our theorem bounding the errors committed by such algorithms and connect this to bounds on errors in sampling and expectation estimate type classical simulation tasks.
In Section~\ref{sec:our_algorithm}, we present our locally entanglement-optimal unraveling algorithm and compare it to relevant prior work.
In Section~\ref{sec:restricted_settings}, we consider our unraveling in restricted settings, such as noisy random circuits and Lindbladian systems. There, we also discuss our ``least unitary'' unraveling.
In Section~\ref{sec:performance}, we present numerical results comparing our MPS-based unraveling simulator to other similar work. We also touch on conditions for efficient classical simulability and discuss other simulation methods (including other tensor network methods and Pauli propagation methods) in comparison to MPS-based unraveling simulators.
In Section~\ref{sec:limitations}, we discuss some limitations of our work and conclude in Section~\ref{sec:conclusion}.

\section{Preliminaries}
\label{sec:preliminaries}

In this section, we recap a few notions that help understand the principles of quantum trajectory unraveling.
We introduce the tool set that is necessary to understand how to unravel open system dynamics into trajectories.

\subsubsection*{Mixed state decompositions}

    A quantum state which 
    is not pure can always be written as a convex combination of normalized pure states 
    \begin{equation}
\label{eq:def_ensemble_decompositions}
        \rho = \sum_{i=1}^{r} p_i \ketbra{\phi_i}{\phi_i} \, .
    \end{equation}
    We call this an ensemble decomposition of the state $\rho$. 
    A familiar example is the eigendecomposition, which also satisfies the orthonormality condition $\braket{\phi_i | \phi_j}=\delta_{i,j}$. 
    However, a richer family of decompositions can be found when an orthogonality condition is not imposed.
    These decompositions are related by a unitary. 
    Given the eigendecomposition 
    $\{\lambda_i, \ket{v_i}\}_{i=1}^{r}$ of $\rho$, by choosing any $r'\in \{r, r+1,\ldots\}$ and unitary $U\in \mathbb{U}(r')$, one can construct another decomposition$\{p_{i'}, \ket{\phi_{i'}}\}_{i'=1}^{r'}$ of $\rho$ defined by the below equation 
    \begin{equation}
    \label{eq:freedom_ensembles_from_eigen}
        \sqrt{p_{i'}} \ket{\phi_{i'}} = \sum_{i=1}^{r} U_{i,i'} \sqrt{\lambda_{i}} \ket{v_{i}} \quad \forall i' \in \{ 1, \dots, r' \} \, .
    \end{equation}
    This also implies that any two decompositions of $\rho$ are also related by a unitary in a similar way.

\subsubsection*{Quantum channel descriptions}

Any completely positive map (also called a quantum channel) has a Kraus representation. 
That is a set of operators $\{K_i\}_{i=1}^{r}$ such that
\begin{equation}
\label{eq:def_kraus}
    \mathcal{N} : \rho \mapsto \sum_{i=1}^{r} K_i \rho K_i^\dagger.
\end{equation}
This is trace-preserving if and only if $\sum_{i=1}^{r} K_i^\dagger K_i = \II$.
Such a Kraus decomposition induces an ensemble decomposition of the output state.
Say we apply a quantum channel $\mathcal{N}$ to some state vector $\ket{\psi}$. 
Then the outcome is given by 
\begin{equation}
\label{eq:Kraus_to_ensemble}
     \mathcal{N}[\ketbra{\psi}{\psi}] 
        = \sum_{i=1}^{r} K_i \ketbra{\psi}{\psi} K_i^\dagger
        = \sum_{i=1}^{r} p_i \ketbra{\phi_i}{\phi_i},
\end{equation}
    where $\{p_i\}_i$ forms a valid probability distribution with $p_i = \| K_i \ket{\psi} \|_2^2 = \bra{\psi} K_i^\dagger K_i \ket{\psi}$ and the normalized state vectors $\ket{\phi_i}$ are $\ket{\phi_i} = 
    p_i^{-1/2}
    K_i \ket{\psi}$.
    Those are the \emph{pure-state trajectories} when unraveling the noisy evolution, and they are not unique. 
    Like ensemble decompositions, Kraus decompositions 
    have a unitary (or rather, isometric) degree of freedom. 
    Given two decompositions 
    $\{K_i\}_{i=1}^{r}$ and $\{K'_{i'}\}_{i'=1}^{r'}$ of the same channel, there exists a unitary such that 
    \begin{equation}
    \label{eq:unitarity_kraus}
        K'_{i'} = \sum_{i=1}^{r} U_{i,i'} K_{i} \quad \forall i' \in \{ 1, \dots, r' \} \, .
    \end{equation}
    This freedom, which we will also refer to as the choice of unraveling, is essentially synonymous with the steering of quantum channels.
    
    \begin{definition}[Unraveling]
    Given a channel $\mathcal{N}$ and an initial state vector $\ket{\psi}$, we call $\{p_i,\ket{\phi_i}\}_i$ an unraveling of the evolution of $\ket{\psi}$ under $\mathcal{N}$ when Eq.~\eqref{eq:Kraus_to_ensemble} is satisfied. 
    \end{definition}
    We will employ this freedom in the choice of unraveling to aid classical simulation by reducing the entanglement. The choice of Kraus decomposition in Eq.~\eqref{eq:Kraus_to_ensemble} that minimizes entanglement will generally depend on the state $\ket{\psi}$.
    By slight abuse of nomenclature, we sometimes refer to any specific Kraus decomposition of a channel $\mathcal{N}$ as an unraveling of this channel.

    \begin{figure}
        \centering
        \includegraphics[scale=1]{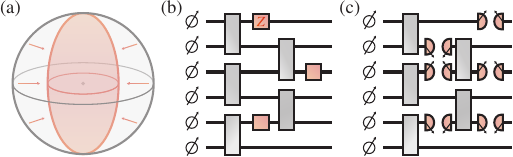}
        \caption{Effect of single-qubit dephasing noise on the Bloch sphere (a) and two interpretations of the unraveling of single-qubit dephasing noise into pure-state trajectories. Circuit corresponding to the unraveling of one trajectory in the (b) mixed-unitary picture and (c) projective picture. The trajectories (and their entanglement) have very 
        different behaviors in both.}
        \label{fig:dephasing_unraveling}
    \end{figure}
    
    Consider the example of the dephasing channel.
    It is often described by its {mixed-unitary} representation 
    \begin{equation}
        \mathcal{N}[\rho] = \left(1-\frac{p}{2}\right) \rho + \frac{p}{2}Z \rho Z 
    \end{equation}
    with Kraus operators $\{K_i\}_{i=1}^{2} =  \{ \sqrt{1-{p}/{2}} \II, \sqrt{{p}/{2}} Z \}$
    and noise rate $p \in [0,1]$.
    This is a probabilistic application of a unitary gate.
    For a given input $\ket{\psi}$, the mixed-unitary unraveling results in the two 
    state vectors $\ket{\psi}$ and $Z \ket{\psi}$ with probabilities $1-\frac{p}{2}$ and $\frac{p}{2}$ respectively.
    The same channel can also be written in terms of projectors on the computational basis
    \begin{equation}
        \mathcal{N}[\rho] = \left(1-p\right) \rho + p \ketbra{0}{0} \rho \ketbra{0}{0} + p \ketbra{1}{1} \rho \ketbra{1}{1}
    \end{equation}
    with $\{K'_i\}_{i=1}^{3} = \left\{ \sqrt{1-p} \II, \sqrt{p} \ketbra{0}{0}, \sqrt{p} \ketbra{1}{1} \right\}$.
    In this case, the outcome of the noise process is unraveled into $\{ \ket{\psi}, \ket{0}, \ket{1} \}$ with probabilities $\{ 1-p, p |\inner{0}{\psi}|^2, p |\inner{1}{\psi}|^2 \}$.
    This second decomposition can be obtained from the first through the unitary
    \begin{equation}
        \begin{pmatrix}
            K_1' \\ K_2' \\ K_3'
        \end{pmatrix} 
        = 
        \begin{pmatrix}
            \sqrt{\frac{2-2p}{2-p}} & 0 & \sqrt{\frac{p}{2-p}} \\ 
            \frac{\sqrt{p}}{\sqrt{2(2-p)}} & \frac{1}{\sqrt{2}} & \sqrt{\frac{1-p}{2-p}} \\ 
            \frac{\sqrt{p}}{\sqrt{2(2-p)}} & -\frac{1}{\sqrt{2}} & \sqrt{\frac{1-p}{2-p}} 
        \end{pmatrix}
        \begin{pmatrix}
            K_1 \\ K_2 \\ 0
        \end{pmatrix} \, .
    \end{equation}
    Note that the study of the specific unraveling of the dephasing channel into the $\ketbra{0}{0}$ and $\ketbra{1}{1}$ projectors, when applied to random circuits, is known as the field of measurement-induced phase transitions (see, e.g., the references within Refs.~\cite{kolodrubetzOptimalityLindbladUnfolding2023, chengEfficientSamplingNoisy2023, chenOptimizedTrajectoryUnraveling2024, vovkQuantumTrajectoryEntanglement2024}).
    A representation of the loss of phase information in the computational basis, as well as examples of circuits resulting from the unraveling of single trajectories, are represented in Fig.~\ref{fig:dephasing_unraveling}.
    Any other unraveling (i.e., any Kraus decomposition) can also be interpreted as the coupling of the target qubit to an auxiliary system followed by a measurement thereof. 
    
    In Appendix~\ref{sec:unitary_freedom}, we discuss more details and a simple example where the unitary degree of freedom of unravelings affects entanglement. We also demonstrate how the relations above are derived and discuss some implications.
    
\subsubsection*{Entanglement measures}

    Here we define notions of entanglement that are central to our work. 
    For pure states, given a partition $A:B$ of the system, it can be quantified by the 
    von Neumann entanglement entropy
    \begin{equation}
        \label{def:entanglement_entropy}
        E(\ket{\psi}_{A:B}) = -\Tr{}{\rho_A \log \rho_A}
        \quad \text{where} \quad
        \rho_A = \Tr{B}{\ketbra{\psi}{\psi}} \, .
    \end{equation}
    In the case of mixed states, for each decomposition of the form~\eqref{eq:def_ensemble_decompositions}, one can define an ensemble-averaged entanglement entropy
    \begin{equation}
    \label{def:ensemble-averaged_entanglement}
        E_{\mathrm{av}}(\{ p_i, \ket{\phi_i} \}) = \sum_{i} p_i E(\ket{\phi_i}_{A:B}) 
        \!\quad \text{with} \!\quad 
        \sum_i p_i \ketbra{\phi_i}{\phi_i} = \rho \, .
    \end{equation}
    These values can vary vastly, and the minimal reachable one over all decompositions is the \emph{entanglement of formation} of the state, defined as
    \begin{align}
    \label{def:entanglement_of_formation}
        E_{\mathrm{oF}}(\rho_{A:B}) & 
        = \inf_{\substack{\{p_i, \ket{\phi_i}\} : \\ \sum_i p_i \ketbra{\phi_i}{\phi_i}=\rho}} \sum_{i} p_i E(\ket{\phi_i}_{A:B}) \, 
    \end{align}
    as the convex hull of the reduced entropy function,
    also known as the \emph{convex roof extension} 
    of the entanglement entropy~\cite{uhlmannRoofsConvexity2010, bennettMixedstateEntanglementQuantum1996}.
    Note that although one can construct ensemble decompositions with an arbitrary number of elements, it is sufficient to consider decompositions with $r^2$ elements where $r$ is the rank of $\rho$~\cite{uhlmannEntropyOptimalDecompositions1998, uhlmannRoofsConvexity2010, audenaertVariationalCharacterizationsSeparability2001}.
    
\subsubsection*{Matrix product states}

    A \emph{matrix product state} (MPS)~\cite{vidalEfficientClassicalSimulation2003, MPSReps,vidalEfficientSimulationOneDimensional2004, schollwoeckDensitymatrixRenormalizationGroup2011} is a representation of 
    a quantum state  where the coefficient of each element (usually expressed in the computational basis) of the state vector is given by a sequence of matrix products as
    \begin{equation} \label{eq:def_MPS}
        \Ket{\psi} = 
         \sum_{\boldsymbol{\sigma}\in\{0,1\}^n} A^{[1]}_{\sigma_1} A^{[2]}_{\sigma_2} \dots A^{[n]}_{\sigma_n} \ket{\boldsymbol{\sigma}} \, .
    \end{equation}
    When viewing the basis label as an open index, an MPS is thus a linear tensor network where the contraction of the tensor train gives the full state vector.
    Note that all $A^{[k]}$ are rank-3 tensors 
    $A^{[k]} \in \CC^{2 \times d_k \times d_{k+1}}$
    except for $A^{[1]}$ and $A^{[n]}$ which are rank-2 tensors
    \footnote{Here the word rank is used for the number of indices of a tensor, so a scalar is a rank-$0$ tensor, a vector is rank $1$, and a matrix is rank $2$. It does not refer to concepts related to the rank of a matrix (maximal number of linearly independent columns).}. 
    In the worst case, the bond dimensions (dimensions of the contracted indices, i.e., dimensions of the matrices in Eq.~\eqref{eq:def_MPS}) can scale exponentially with $d_k \leq 2^{\min (k, n-k)}$. 
    Still, these constructions are commonly used when dealing with states with low entanglement~\cite{verstraeteMatrixProductStates2006, schuchEntropyScalingSimulability2008}.
    As one is most often given a fixed computational budget, the matrices are truncated. 
    The error committed in this procedure can be bounded using the singular values of the tensors.
    Indeed, the Schmidt decomposition of the state for any cut $1:k \ | \ k+1:n$ of the line can be obtained via an adequate singular value decomposition of one tensor.
    Then the truncation error when retaining the $d_k$ largest singular values (which are the Schmidt coefficients) $s_i$ (from $d_{k,{\max}} = 2^{\min (k, n-k)}$) is given by
    \begin{equation}
        \epsilon_k(d_k) = \sum_{i=d_k+1}^{d_{k,{\max}}} s_i^2 \, .
    \end{equation}
    Truncating all bonds for all bipartitions along the line, the $2$-norm difference between the true state and the approximation is upper-bounded by 
    \begin{equation}
    \label{eq:MPS_truncation_error}
        \| \ket{\psi} - \ket{\psi'} \|_2^2 \leq 2 \sum_{k=1}^{n-1} \epsilon_k(d_k) 
    \end{equation}
    where $\ket{\psi'}$ denotes the (sub-normalized) approximate version of $\ket{\psi}$ where each bond has been truncated (represented graphically in Fig.~\ref{fig:overview}, top right).

    For simplicity, let us consider that all bonds are truncated to a maximum bond dimension $\chi$.
    MPS then allow for the efficient computation in $\poly(n, \chi)$ time of many quantities of interest.
    These include inner products with other MPS of bond dimension $\leq \chi$, expectation values of MPO observables, index marginalization, and, importantly, computing the Schmidt decomposition across a cut in the MPS.
    The latter implies the computation of the entanglement entropy across any cut of the linear structure, as the squared singular values over the cut (after adequately shifting the orthogonality center) are the eigenvalues of the reduced density matrix obtained from tracing out the left or right part of the cut.  

    \begin{figure*}
        \centering
        \includegraphics[scale=1]{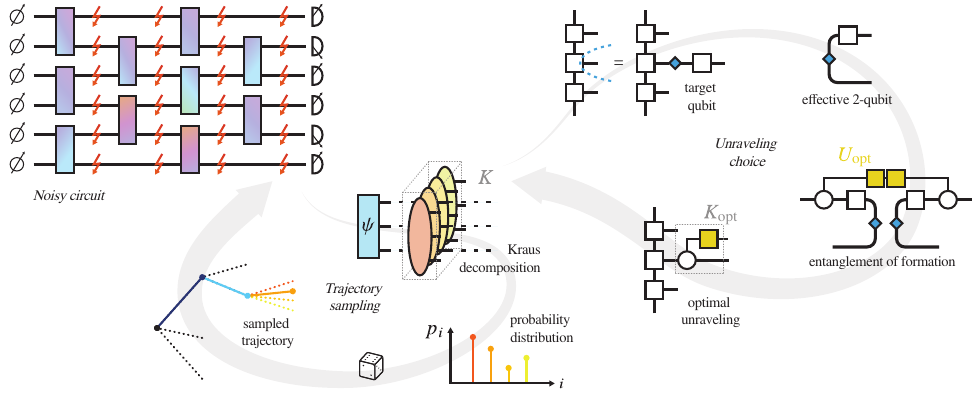}
        \caption{
        Graphical representation of the unraveling of quantum trajectories based on matrix product states, for noisy circuits
        Consider the setting of one-dimensional noisy circuits, composed of (two-)local unitary gates, followed by single-qubit noise channels (top left). 
        The unitary gates are dealt with by the standard MPS contraction. 
        The noise results in a statistical mixture of pure states. 
        This mixture is then unraveled to sample one trajectory, and the resulting state is truncated to the allowed bond dimension.
        The trajectory sampling results from computing the probabilities of each path and sampling one of them, thus applying the corresponding Kraus operator (bottom). 
        Before that, the unraveling is chosen. 
        In our case, we do so based on the computation of the entanglement of formation of the affected effective two-qubit state (top right). 
        }
        \label{fig:overview}
    \end{figure*}

\section{Simulation through unraveled quantum trajectory sampling}
\label{sec:unraveling}

    This section introduces a general class of MPS-based classical simulation algorithms that contains our proposed algorithm as a special case.
    We call these trajectory sampling simulators, but these connect to long-established techniques~\cite {zollerQuantumJumpsAtomic1987, dumMonteCarloSimulation1992, carmichaelOpenSystemsApproach1993, molmerMonteCarloWavefunction1993, molmerMonteCarloWavefunctions1996, plenioQuantumjumpApproachDissipative1998, daleyQuantumTrajectoriesOpen2014} also referred to as quantum jump methods or Monte Carlo wave function methods. Trajectory sampling simulators have a controlled runtime but an uncontrolled simulation error parameter $\epsilon \in [0,1]$. That is, the runtime of the algorithm can be directly controlled by an initial parameter choice that, in the special case of our algorithm, is the MPS bond dimension. In contrast, the simulation error is not directly controlled. In our algorithm, the error is indirectly dependent on the MPS bond dimension and approaches zero at or before the MPS bond dimension approaches its maximal value of $2^{n/2}$.
    
    Algorithms in this class perform a task we call {Truncated Trajectory Sampling} to varying degrees of error $\epsilon$. In this section, we state one of our main results. 
    We show that all such simulators can be used to compute an a posteriori probabilistic upper bound on this, circuit-dependent, simulation error. 
    Our error bound is based on an estimator computed during the simulation procedure, which has a mean that can vary significantly depending on the specific properties of the simulator. 
    Here, we also relate the output of this class of algorithms to more commonly used notions of classical simulation, such as simulators that sample from the output distribution of a quantum circuit and simulators that estimate expectation values with respect to a given observable.

    After defining the concept of truncated trajectory sampling in Section~\ref{sec:TTS_def} and describing an abstract class of algorithms performing such a simulation in Section~\ref{sec:TTS_alg}, we give our bounds on the errors committed by such algorithms in Section~\ref{sec:error_bounds} and show how it results in bounds on errors in simulation tasks like sampling from the Born distribution or estimating expectation values of observables, Section~\ref{sec:reductions}.

\subsection{Truncated Trajectory Sampling Task}
\label{sec:TTS_def}

    At a high level, trajectory sampling simulators break state evolution into layers and perform the following process at the level of each layer:
    \begin{enumerate}
        \item Use an appropriate data structure to classically represent the quantum state inputted into the computation layer.
        \item Apply the quantum channel associated with the computation layer.
        \item Decompose the (mixed) state into a convex combination of (usually pure) states.
        \item Sample a state from the decomposition. 
        \item Approximate the sampled state with a nearby state that has preferable memory and runtime properties, and use this approximation as input to the next computation layer.
    \end{enumerate} 
    Our work focuses on utilizing the MPS data structure as described in point 1 above. 
    Matrix product states are a powerful tool to simulate pure quantum states with little entanglement, in particular 
    in \emph{one spatial dimension} (1D).
    This provides a natural data structure for the trajectory state vectors.
    As discussed in Section~\ref{sec:preliminaries}, mixed states provide an isometric degree of freedom in the choice of unraveling or decomposition in point 3 above. 
    Since, typically, states with low entanglement can be approximated more accurately by a low bond dimension MPS, we employ this isometric degree of freedom to choose unravelings that lead to lower entanglement states in point 4 (and hence also incur small approximation errors when approximated by a low bond dimension MPS in point 5). 
    We now properly define what we mean by an algorithm performing truncated trajectory sampling.
    
Let $\mathcal{D}$ be a classical description of a noisy quantum circuit on $n$ qubits consisting of 
        \begin{itemize}
            \item a qubit ordering: A linear ordering of the $n$-qubits by labeling them from $1$ to $n$,
            \item an initial state: An initial  state vector $\ket{\psi^{(0)}}$ given as an MPS (e.g. $\ket{0}^{\otimes n}$) with respect to the given qubit ordering,
            \item and a string of transformations: 
            A sequence of efficiently representable and applicable quantum channels on $n$ qubits (for instance, each comprised of one or two qubit unitary gates on disjoint sets of qubits composed with single-qubit noise channels).
        \end{itemize}
We use $\rho$ to denote the final quantum state associated with $\mathcal{D}$. Given $\chi\in \mathbb{N}$ and a classical description $\mathcal{D}$, the goal is to sample from an ensemble of MPS states with bond dimension at most $\chi$ such that the expected state with respect to this ensemble approximates $\rho$ in trace distance.

 To be more precise, we first define a key concept that we then use to define {Truncated Trajectory Sampling}.

\begin{definition}[$(\epsilon, \chi)$-convex MPS]
    \label{def:convex-MPS}
    Let $\rho$ be an $n$-qubit quantum state with a specified qubit ordering. Let $\epsilon\geq 0$ and $\chi\in \mathds{N}$. Then we say that $\rho$ is an $(\epsilon, \chi)$-convex MPS with respect to its qubit ordering iff there exists a probability distribution $\{q_i\}_i$ and a corresponding set of pure quantum states $\{\ket{\psi_i}\}_i$ such that
    \begin{enumerate}
        \item for all $i$, $\ket{\psi_i}$ admits a classical description in the form of an MPS with bond dimension $\leq \chi$ with respect to the specified qubit ordering, and
        \item the ensemble $\{q_i,\ket{\psi_i}\}_{i}$ satisfies
            \begin{equation}\label{eq:trace_distance_UB}
            \| \rho - \sum_i q_i \ketbra{\psi_i}{\psi_i} \|_{\mathrm{Tr}} \leq \epsilon \, .
            \end{equation}
    \end{enumerate}
    We call any such ensemble, $\{q_i,\ket{\psi_i}\}_{i}$, an $(\epsilon, \chi)$-convex MPS witness for $\rho$.
\end{definition}

Note that each MPS representation has an inherent qubit ordering, and different choices of qubit ordering are possible. Further, the minimal bond dimension of the MPS representations of the same state with respect to different qubit orderings can be extremely different. For this reason, we specify the qubit ordering of the MPS representations.
In many settings, however, a natural ordering is already imposed by the system's geometry.
That is the case when studying one-dimensional quantum circuits or open systems with nearest-neighbor interactions on a line. 

\begin{definition}[Truncated Trajectory Sampling]
    \label{def:TTS}
    Given a noisy circuit description $\mathcal{D}$ on $n$ qubits (as specified above), let $\rho$ be the final quantum state. Given a maximal bond dimension $\chi$, we say that an algorithm 
    performs $\epsilon\textsc{-tts}(\mathcal D, \chi)$ 
    iff $\rho$ is an $(\epsilon, \chi)$-convex MPS, and for some $\{q_i,\ket{\psi_i}\}_{i}$ that is an $(\epsilon, \chi)$-convex MPS witness for $\rho$, with probability $q_i$, the algorithm returns a classical description of a state vector $\ket{\psi_i}$ in the form of an MPS with bond dimension at most $\chi$.
\end{definition}
Below we define a general class of algorithms that perform \emph{Truncated Trajectory Sampling} to varying degrees of accuracy $\epsilon \in [0,1]$. We discuss how these algorithms are constructed and how to bound $\epsilon$ for any given algorithm, $\chi$, and $\mathcal{D}$.

\subsection{Class of Truncated Trajectory Sampling Algorithms}
\label{sec:TTS_alg}

    We now describe a general class of classical simulation algorithms that perform truncated trajectory sampling. 
    The simulators in this class use three key subprocedures, which we refer to as the Layering, Kraus Tree, and the Approximation subprocedures. 
    Although other implementation details and inputs may vary, conditioned on the details of these subprocedures, these algorithms all perform an equivalent high-level procedure, which we will describe shortly.
\begin{enumerate}
    \item Layering Subprocedure: This algorithm specifies how the circuit appearing in the circuit description $\mathcal{D}$ is to be broken up into a string of transformations (unitary gates or noisy processes) that will be applied, in order, to a quantum state representation.
    \item Kraus Tree Subprocedure: Each time a transformation is applied, this can be represented by one of (potentially) many Kraus decompositions\footnote{We note that unitary transformations are represented as a single-element Kraus decomposition.}. We will consider a simulation procedure where a Kraus element is sampled from a particular Kraus decomposition associated with each noisy transformation. After the first Kraus operator is sampled from the decomposition of the first noisy transformation, subsequent Kraus decompositions may depend on the sampled Kraus element. This defines a tree structure where:
    \begin{itemize}
        \item Each node is associated with an input and output state (for now, these can be treated identically; this distinction will become relevant when we introduce the approximation subprocedure), with the initial state of the system equal to the input state of the root node.
        \item Each edge from a node to one of its children is associated with exactly one Kraus element, where the choice of Kraus decomposition defines the set of all such edges. The latter can depend on the unique path from the root to the parent node. 
        The update of the state along an edge follows Eq.~\eqref{eq:Kraus_to_ensemble}.
        Given the output state of the parent node, $\ket{\psi}$, and a Kraus element $K_i$ associated with an edge to a child node, the probability of that edge is $p_i = \norm{K_i \ket{\psi}}_2$ and a state $ p_i^{-1/2} K\ket{\psi}$ is associated with the input state of the child node.
    \end{itemize} 
    \begin{figure}
        \centering
        \includegraphics[scale=1]{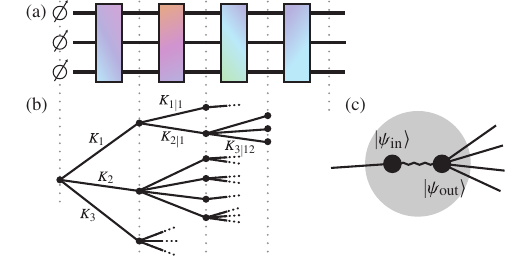}
        \caption{(a) Simple layered process of quantum channels (represented as colored boxes, although they may be non-unitary) and (b) the corresponding  Kraus tree, with Kraus elements of each edge possibly depending on the history of the path. 
        (c) Each node has an input and an output state, related by some approximation procedure (represented as a wiggly line).}
        \label{fig:kraus_tree}
    \end{figure}
    We refer to the tree labeled with Kraus operators, input/output states, and transition probabilities as the Kraus tree. Although this tree may have a number of nodes that is growing exponentially in the circuit gate count, it never needs to be explicitly constructed and can be fully specified by an algorithmic procedure that, given a node and the output state associated with that node, determines how to choose the Kraus decomposition for the next transformation to be applied. We refer to this as the Kraus Tree Subprocedure. As a simple example, one can consider Kraus Tree Subprocedures, where the Kraus decomposition is independent of the node and state, depending only on the next transformation to be applied, with each type of transformation in the circuit assigned a fixed Kraus decomposition. 

    \item Approximation Subprocedure: This subprocedure determines the mapping from the input state to the output state of each node of the Kraus tree. This subprocedure is intended to be used to simplify a state's representation by employing a suitable approximation. In our setting, where an MPS data structure is employed, a natural example is an approximation subprocedure that truncates the input state to one with a maximum permitted bond dimension.
\end{enumerate}

We consider the set of all algorithms that specify a Layering, Kraus Tree, and Approximation subprocedure and implement a simulation by sampling a complete path from the root to a leaf of the Kraus tree based on the transition probabilities assigned to the edges. 
These algorithms 
ultimately return the output state at the leaf of the sampled branch of the Kraus tree.
We call this one (truncated) \emph{trajectory}. 

We observe that if the approximation subprocedure is chosen to be the trivial identity map (i.e., the state is always represented exactly), then the expectation value over trajectories
is exactly the (mixed) state $\rho$ produced immediately prior to measurement. However, if non-trivial approximations are made at any point in the Kraus tree, the output state $\rho'$ will generally differ from $\rho$.  

In Appendix~\ref{proof:sampled_trajectories}, we show how to construct a probabilistic upper bound for the total trace distance between $\rho$ and $\rho'$ using only a polynomial number of samples from branches of the Kraus tree. 
At a high level, we sample a branch from the Kraus tree and consider the sum of trace distances between the input and output state of each node along the branch. 
We show that the expectation value of this random variable (and related constructions) acts as an upper bound on the trace distance between $\rho$ and $\rho'$. 
This permits us to estimate an upper bound of the expectation value as the mean of a number of samples of these random variables. 
To produce a probabilistic upper bound, we add a small buffer term to this mean to ensure that, with high probability, we are overestimating the expectation. 
This buffer term can be made smaller by increasing the number of samples used. 

In this work, we focus on the setting of the MPS data structure where the approximation subprocedure is used to impose an upper bound on the bond dimension of the output states at all nodes. With this additional restriction, the set of algorithms consistent with the above description perform \emph{Truncated Trajectory Sampling} to varying degrees of accuracy $\epsilon \in [0,1]$. Theorem~\ref{result:sampled_trajectories} adapts our probabilistic upper bound on the accuracy $\epsilon$ to the Truncated Trajectory Sampling setting. 
We emphasize that although our accuracy upper bounds apply to a broad class of algorithms, the size of the upper bound is strongly dependent on the specific details of the algorithm. 
Hence, sharper algorithms can achieve a smaller upper bound on the trace distance between $\rho$ and $\rho'$. 
In Section~\ref{sec:optimal_single_qubit_unraveling}, we present our algorithm, a special case of the set of algorithms described here. 
This provides a specific example of how unraveling can be used to reduce the accuracy parameter $\epsilon$ associated with Truncated Trajectory Sampling. 
But first, we introduce some more commonly used notions of classical simulation and discuss their connection to Truncated Trajectory Sampling.

\subsection{Upper bounds on the simulation errors}
\label{sec:error_bounds}

    In this section, we present our main contribution to the rigorous analysis of the errors committed in such quantum trajectory protocols (Theorem~\ref{result:sampled_trajectories}).
    Let us start by giving labels to the key quantities, and then state the main result.

    Each trajectory $\ket{\hat{\psi}_i}$ is drawn from the decomposition of some state $\sigma = \sum_i q_i \ketbra{\psi_i}{\psi_i}$. 
    This equates to saying 
    $\mathds{E} [ \ketbra{\hat{\psi}_i}{\hat{\psi}_i} ] = \sigma$.
    In the case of no approximations, one has $\sigma = \rho$, the exact state resulting from the noisy evolution, 
    while $\| \rho - \sigma \|_{\mathrm{Tr}} \geq 0$ otherwise.
    And although the error is hard to compute exactly efficiently, given $N$ trajectories sampled independently, 
    there is an estimator $\hat{\varepsilon}_N$ that can be computed in $\poly(n,\chi,N)$ time. 
    This is associated with the accuracy of trajectory sampling simulation algorithms.

    For a trajectory $\ket{\hat{\psi}_i}$ with $\varepsilon_{\mathrm{approx}}(\hat{\psi}_i, \ell)$ the approximation error at layer $\ell$, we can define a function of the errors at each layer bounding the cumulated error of the trajectory.
    This (bounded) total error of the trajectory is the sum of errors  
    $\varepsilon_{\mathrm{bound}} (\hat{\psi}_i) = \sum_{\ell=1}^L \varepsilon_{\mathrm{approx}}(\hat{\psi}_i, \ell)$ if this sum is smaller than some threshold $\varepsilon_{\mathrm{max}}-2$ and is set to $\varepsilon_{\mathrm{max}}$ otherwise. 
    With this, averaging over the sampled trajectories, one can compute an estimator of the error on the simulated state
    \begin{equation}
        \hat{\varepsilon}_N = \frac{1}{N} \sum_{i=1}^N \varepsilon_{\mathrm{bound}} (\hat{\psi}_i) \, . 
    \end{equation}
    Here, $\varepsilon_{\mathrm{max}}$ is the maximum tolerated approximation error in a single trajectory, which can be freely chosen between $2$ and $2L$.
    This allows us to state our first main result (selecting $\varepsilon_{\mathrm{max}} = 2$).

    \begin{restatable}[Truncated Trajectory Sampling error]{theorem}{trajectoryerror}
    \label{result:sampled_trajectories}
        Given a noisy circuit description $\mathcal{D}$ with a final state $\rho$ and a maximal bond dimension $\chi\in \mathbb{N}$, there exist an $\epsilon\geq 0$ such that with runtime $\poly(n, \chi, N)$ the trajectory sampling algorithm returns samples from an $(\epsilon, \chi)$-convex MPS witness for $\rho$. 
        Furthermore, for any $\delta>0$, with probability at least $1-\delta$,
        \begin{equation}
            \epsilon \leq \hat{\varepsilon}_N + \sqrt{\frac{2}{N} \log \left( \frac{1}{\delta} \right)}.
        \end{equation}
    \end{restatable}

    We note that when $\chi\geq 2^{n/2}$, for all $n=1,2,\ldots$, $\hat{\varepsilon}_N=0$. 
   
    This bound is composed of a systematic error and a statistical error.
    First, the \emph{systematic error} is due to the truncations. 
    Even in the hypothetical case where one would compute all (exponentially many) possible trajectories, each would be only an approximate representation of the real trajectories. 
    The difference between the real noisy state $\rho$ and the mixture of \emph{all} unraveled trajectories $\sigma = \mathds{E}[\ketbra{\hat{\psi}_i}{\hat{\psi}_i}]$ is the systematic error of the simulation. 
    This error can be bounded by $\lim_{N \rightarrow \infty} \hat{\varepsilon}_N$. 
    The resulting accuracy can be improved by reducing the approximation errors of each trajectory, i.e., by increasing the bond dimension of the MPS.
    Second, there is the uncertainty due to finite sampling, the \emph{statistical error}. 
    The systematic error is constructed from the sampled trajectories, as a cumulation of the errors in each. 
    The exact computation would require computing all trajectories, which can not be done efficiently. 
    As a result, one needs to estimate the systematic error itself. 
    The bound $\hat{\varepsilon}_N$ obtained by averaging the errors of the sampled trajectories holds with high probability when adding the ``statistical buffer'' $\sqrt{{(\varepsilon_{\mathrm{max}}^2}/{(2N)}) \log \left( {1}/{\delta} \right)}$. 
    The tightness of the bound (or the probability of failure) can be reduced by increasing the number of samples.
    When using MPS, the approximation error $\varepsilon_{\mathrm{approx}}(\hat{\psi}_i, \ell)$ of the trajectory $\ket{\hat{\psi}_i}$ at depth $\ell$ due to the truncations
    can be computed using that 
    $\Norm{\ketbra{\psi}{\psi} - \ketbra{\phi}{\phi}}_{\mathrm{Tr}}=\sqrt{1-\abs{\braket{\psi|\phi}}^2}$ 
    at the cost of having to also compute the non-truncated evolution, 
    or using the bound from Eq.~\eqref{eq:MPS_truncation_error}, resulting in
    \begin{equation}
    \label{eq:MPS_truncation_error_w_sqrt}
        \varepsilon_{\mathrm{approx}}(\hat{\psi}_i, \ell) \leq 4 \sqrt{2 \sum_{k} \epsilon_k(\chi)}.
    \end{equation}
    The full proof of Theorem~\ref{result:sampled_trajectories} and details about the MPS-based errors can be found in Appendix~\ref{sec:upper_bound_proofs}.

    Note that, in general, the bound from Theorem~\ref{result:sampled_trajectories} is not tight.
    For some noise models of interest, it can be made tighter using concentration arguments.
    Then, the error contribution of each layer is scaled by a factor decaying exponentially with the depth of the circuit. 
    In this case, we obtain a new bound (omitting the bounding with $\varepsilon_{\mathrm{max}}-2$ for simplicity)
    \begin{align}
        \hat{\varepsilon}_N = \frac{1}{N} \sum_{i=1}^N \left( \sum_{\ell=1}^L \alpha^{L-\ell}\varepsilon_{\mathrm{approx}}(\hat{\psi}_i, \ell) \right)
    \end{align}
    for some noise-dependent factor $\alpha < 1$. 
    See Appendix~\ref{sec:concentration_bounds} for a presentation of the argument and examples of specific values of $\alpha$.

    We are now ready to introduce some more commonly used notions of classical simulation and demonstrate how our results on Truncated Trajectory Sampling result in bounds on simulation errors.

\subsection{Target classical simulation tasks}
\label{sec:reductions}

In this section, we define two commonly used notions of classical simulation, which we call {Output Distribution Sampling} and {Observable Estimation}, and show their connection to Truncated Trajectory Sampling.  

Given:
\begin{enumerate}
    \item A classical description $\mathcal{D}$ of a noisy quantum circuit on $n$ qubits,
    \item An accuracy parameter $\epsilon>0$ and a failure probability $\delta>0$,
\end{enumerate} 
the goal is to perform the following tasks with probability at least $1-\delta$, over the randomness of the algorithm:

\begin{itemize}
    \item {Output Distribution Sampling:} 
        Sample from any probability distribution within the $\epsilon$-ball (in total variation distance) of the outcome distribution
        \begin{align}
        P(x \mid \rho) & \coloneqq \langle x \mid \rho \mid x \rangle
        \end{align}
        associated with measuring the final state $\rho$ in the computational basis $\{\ket{x}\}_x$.
    \item {Observable Estimation:} 
        Given an observable $O$ in the form of an MPO of bond dimension upper bounded by $\poly(\log n)$, estimate $\braket{O} = \Tr{}{O \rho}$ up to $\pm \epsilon$ in additive error.
\end{itemize}

We denote an instance of each of these problems using $\textsc{Samp}(\mathcal{D},\epsilon,\delta)$ and $\textsc{OEstim}(\mathcal{D},O,\epsilon,\delta)$ respectively. 

    In general\footnote{Without placing specific restrictions on the level and nature of the noise model, this class of circuits includes universal noiseless circuits.}, we cannot expect to find an efficient classical algorithm for these simulation tasks that works for all instances. This is because solving the {Output Distribution Sampling} task in $\poly(n,\epsilon^{-1}, \log \delta^{-1})$ time is as hard as $\mathrm{SampBQP}$, the class of sampling problems efficiently solvable by a (noiseless) universal quantum computer. 
    It is also easy to show that solving the {Observable Estimation} task in $\poly(n)$ time for any $\epsilon + \delta<1/6$ and $O=\ketbra{0}{0}\otimes I^{\otimes n-1}$ is $\mathrm{BQP}$-hard\footnote{$\mathrm{BQP}$ is the class of decision problems efficiently solvable by a (noiseless) universal quantum computer.}.
    Note that the estimation of Born probabilities and marginals (the probabilities of outcomes in the computational basis) is a special case of {Observable Estimation}. For a detailed discussion of how these notions of classical simulation relate to each other, see Ref.~\cite{pashayanEstimationQuantumProbabilities2020}.

Before introducing our algorithm in more detail, we outline how algorithms for {Truncated Trajectory Sampling} can be used to perform  {Output Distribution Sampling} and {Observable Estimation}. 

\begin{restatable}[Output distribution sampling from trajectory sampling]{lemma}{samplingreduc}
    \label{result:sampling_reduction}
    An algorithm performing $\epsilon \textsc{-tts}(\mathcal{D}, \chi)$ with probability $1-\delta$ can, with $\poly(n,\chi)$-time classical processing, produce a solution to $\textsc{samp}(\mathcal{D}, \frac{1}{2} \epsilon, \delta)$.  
\end{restatable}

\begin{restatable}[Expectation value estimation from trajectory sampling]{lemma}{observablereduc}
    \label{result:expectation_reduction}
    For any $\eta>0$, $N = 2\|O\|_{\mathrm{op}}^2 \eta^{-2} \log (2\delta'^{-1})$ independent outputs of an algorithm performing $\epsilon \textsc{-tts}(\mathcal{D},\chi)$ with probability $1-\delta$ can, with $\poly(n,\chi)$-time classical processing, produce a solution to
    \begin{itemize}
        \item $\textsc{Oestim}(\mathcal{D}, O, \epsilon', \delta')$ where $\epsilon' = [\epsilon + 2\delta] \|O\|_{\mathrm{op}} + \eta$ or 
        \item $\textsc{Oestim}(\mathcal{D}, O, \epsilon', \delta+\delta')$ where $\epsilon' = \epsilon \|O\|_{\mathrm{op}} + \eta$.
    \end{itemize}
\end{restatable}

    Lemma~\ref{result:sampling_reduction} derives from the fact that each trajectory, independent of the choice of unraveling, is a pure state drawn from a valid ensemble decomposition of the approximate output state of the circuit. 
    With this, the measurement statistics averaged over the elements of the decomposition match those of the mixed state. 
    Leveraging the fact that any bound on the trace distance between states is a bound on the TV distance between the respective distributions leads to the result above.
    
    For Lemma~\ref{result:expectation_reduction}, one also has to take into account that the finite sampling of trajectories does not fully reconstruct the approximate state. 
    Since not all trajectories are sampled but only a subset, one can only guarantee that the unraveled mixed state is within a ball around the target state with high probability. 
    The incurred statistical error is reflected in the additional term.
    The detailed proofs of these statements can be found in Appendix~\ref{sec:reduction_proofs}.

    With this, Theorem~\ref{result:sampled_trajectories} allows us to state bounds on the errors of the simulation tasks presented above. 
    \begin{corollary}[Output distribution sampling]
    \label{result:sampling_task}
        Given the setting of Theorem~\ref{result:sampled_trajectories}, $N$ samples from the trajectory sampling algorithm allow for computing $N$ samples from a probability distribution $\mathcal{P}'$ which is close to $\mathcal{P}$
        \begin{equation}
            TV(\mathcal{P},\mathcal{P}')
            \leq \hat{\varepsilon}_N + \sqrt{\frac{\varepsilon_{\mathrm{max}}^2}{2N} \log \left( \frac{1}{\delta} \right)}
        \end{equation}
        with probability $1-\delta$ 
        where $\mathcal{P} \coloneqq \{ P(x|\rho) \}_x$ is the distribution of outcomes generated when the noisy circuit described above is applied to the $\ket{\psi^{(0)}}$ state and all qubits are subsequently measured in the computational basis.
    \end{corollary}
    \begin{proof}
        This is a direct application of Lemma~\ref{result:sampling_reduction} to the result of Theorem~\ref{result:sampled_trajectories}.
    \end{proof}
    
    \begin{corollary}[Expectation value estimation]
    \label{result:expectation_task}
        Given the setting of Theorem~\ref{result:sampled_trajectories} and a bounded observable $O$ with $\|O\|_{\mathrm{op}} \leq 1$, 
        $N$ samples from the trajectory sampling algorithm can be used to construct an estimator $\hat{O}$ for the expectation value $\braket{O}$  with
        \begin{equation}
            \left| \hat{O} - \braket{O} \right| 
            \leq \hat{\varepsilon}_N + \sqrt{\frac{\varepsilon_{\mathrm{max}}^2}{2N} \log \left( \frac{1}{\delta} \right)} + \sqrt{\frac{2}{N} \log \left( \frac{2}{\delta'} \right)}
        \end{equation}
        with probability at least $1-\delta-\delta'$.
    \end{corollary}
    \begin{proof}
        This is a direct application of Lemma~\ref{result:expectation_reduction} to the result of Theorem~\ref{result:sampled_trajectories}.
    \end{proof}

    We can now focus entirely on how MPS-based quantum trajectories allow us to compute solutions to truncated trajectory sampling tasks. 
    We argue that our algorithm achieves a state-of-the-art performance among unraveling-based simulators of noisy quantum circuits.

\section{Our algorithm}
\label{sec:our_algorithm}

    In this section, we will dive into the details of our algorithm. 
    For this, we review MPS-based trajectory sampling algorithms (as presented in Section~\ref{sec:TTS_alg} and studied in Refs.~\cite{vovkEntanglementOptimalTrajectoriesManyBody2022, kolodrubetzOptimalityLindbladUnfolding2023, chengEfficientSamplingNoisy2023, vovkQuantumTrajectoryEntanglement2024, chenOptimizedTrajectoryUnraveling2024, darabanNonunitarityMaximizingUnraveling2025}) in Section~\ref{sec:MPS_based_details}.
    We then present our locally entanglement-optimal choice of unraveling and how it is derived in Section~\ref{sec:optimal_single_qubit_unraveling}.
    In Section~\ref{sec:unraveling_comparison}, we give an overview and comparison of state-of-the-art choices of unraveling for MPS-based trajectory sampling.
    
\subsection{Details of MPS-based trajectory sampling}
\label{sec:MPS_based_details}
    
    The central idea of quantum trajectory unraveling is to interpret open system evolution as a stochastic process over (non-unitary) normalized pure-state evolutions. 
    For each layer, starting from a state vector $\ket{\psi}$ the effect of the computational layer $\Lambda$ is unraveled by specifying a particular Kraus representation $\{ K_i \}_i$ of the channel and obtaining the induced ensemble decomposition as per Eq.~\eqref{eq:Kraus_to_ensemble}. 
    Put into words, Eq.~\eqref{eq:Kraus_to_ensemble} is a statement that the output state is a probabilistic average over the distribution $\boldsymbol{p} = \{ p_i \}_i$ of normalized state vectors $\{ \ket{\phi_i} = p_i^{-1/2}K_i \ket{\psi} \}_i$.
    This procedure is represented graphically in Fig.~\ref{fig:overview} (bottom left).
    Note that one has the freedom to choose the Kraus decomposition of the channel. 
    This different choice of unraveling is synonymous with a different $U$ in Eq.~\eqref{eq:unitarity_kraus}, and it is precisely this freedom we will leverage to optimize the algorithm in Section~\ref{sec:optimal_single_qubit_unraveling}.

    We now expose how to perform the update efficiently when using MPS and few-qubit unitaries, followed by single-qubit noise channels.
    From here on out, we consider 1D architectures of geometrically local, few-qubit, unitary gates composed with local (single-qubit) noise channels, like for instance the brickwork-
like circuit structure shown in Fig.~\ref{fig:overview} (top left). 
    This setting could be made more general, but it already contains all the key elements to set the stage for our locally entanglement-optimal unraveling.
    Let us describe each block individually.
    
    Given a pure state represented as an MPS of bond dimension $\chi$, the action of a unitary gate is obtained by contracting the (four-legged) tensor with the affected MPS tensors and then recovering the MPS structure through SVD.
    \begin{equation*}
        \includegraphics{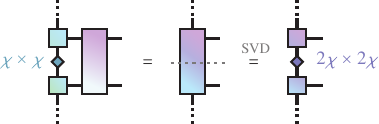}
    \end{equation*}
    One layer of two-qubit, nearest-neighbor, unitaries acting on disjoint sets of qubits will result in a new MPS with increased bond dimension of at most $2\chi$. 
    
    Each single-qubit noise channel is also applied locally. 
    They can be represented by a three-legged tensor $K = \sum_i K_i \otimes \bra{i}$, 
    then Eq.~\eqref{eq:def_kraus} can be formulated as
    \begin{equation*}
        \includegraphics{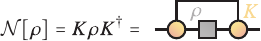}
    \end{equation*}
    After selecting the Kraus decomposition (the unraveling), the probabilities of each path (at most four in our choice of unraveling) are computed
    \begin{align}
        p_i = & 
        \braket{K_i^\dagger K_i} = \bra{\psi} K_i^\dagger K_i \ket{\psi} \\ & 
        \includegraphics{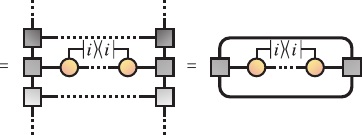} 
    \end{align}
    and one of them is sampled. 
    The non-unitary evolution is then implemented by contraction of the respective Kraus operator and renormalization by the corresponding probability. 
    \begin{equation*}
        \includegraphics{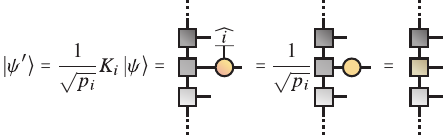}
    \end{equation*}
    This sampling is repeated for each qubit, giving a valid trajectory of the full first layer.
    
    The MPS may then be truncated to remain within the maximum allowed bond dimension.
    This is done by discarding the smallest singular values of each bond (i.e., the smallest Schmidt coefficients of each bipartition)~\cite{verstraeteMatrixProductStates2006}.
    \begin{equation*}
        \includegraphics{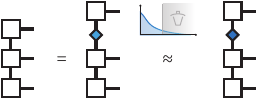}
    \end{equation*}
    One is thus given a trade-off between cost and simulation error.
    The approximation error for the layer, computed efficiently, is added to the cumulated error of simulating the trajectory.
    Note that in practical implementations, not focused on deterministic calculations, one could variationally look for the closest MPS with the target bond dimension and compute an approximation error from the resulting overlap~\cite{vanheckeTangentspaceMethodsTruncating2021,Hintermueller}.
    
    Repeating this procedure (unitary evolution, noise process, truncation) for every layer returns one of the (at most $4^{nL}$) trajectories of the full noisy circuit, as visualized in Fig.~\ref{fig:overview}. 
    An algorithmic description of the MPS-based unraveling of a single noise channel is presented in Algorithm~\ref{alg:channel_unraveling}.

    Note that the layering process also affects the accuracy and cost of the simulation.
    Although not central to our work, we outline two methods for interleaving unitary action, noise process, and truncation to illustrate the concept.
    One layering option is to define the unitary layer as $\lfloor n/2 \rfloor$ two-qubit unitary gates, the noise layer as the tensor product of $n$ single-qubit noise channels, and then a truncation on all bonds
    \begin{equation}
        \operatorname{Trunc}^{\otimes n-1} \circ \mathcal{N}^{\otimes n} \circ \mathcal{U}^{\otimes \lfloor n/2 \rfloor} \, .
    \end{equation}
    Alternatively, one can choose to perform the unitary, noise, and truncation locally on the two qubits affected by the unitary
    \begin{equation}
        \left( \operatorname{Trunc}^{\otimes 1} \circ \mathcal{N}^{\otimes 2} \circ \mathcal{U} \right)^{\otimes \lfloor n/2 \rfloor} \, .
    \end{equation}
The first strategy will result in lower truncation errors, as the noise on the last qubit may affect the first bond. 
    However, it requires an intermediate representation of the state as an MPS with bond dimension $2\chi$ on every second bond.
    The first approach does not leverage the effect of the noise on all qubits to reduce the bond dimension before truncating, but never needs to construct an MPS with a bond dimension larger than $\chi$.
    In general, the intermediate bond dimension will depend on the interplay between the channel layers and the truncation imposed by the choice of layering. 
    This results in a trade-off between computational cost (bond dimension) and precision (truncation error).

 \begin{figure}          
    \begin{algorithm}[H]
    \caption{Channel unraveling subprocedure} 
    \label{alg:channel_unraveling}
    \begin{algorithmic}[1]
        \State \textbf{\MakeUppercase{input}} 
            \State initial state vector $\ket{\psi}$ as an MPS with bond dimension $\chi$
            \State noise channel description $\mathcal{N}$ 
            % \Comment{Kraus decomposition}
            \State location of qubit affected by noise $q$
        \State \textbf{\MakeUppercase{output}} 
            \State state $\ket{\psi'}$ as an MPS with bond dimension $\chi'$
        \State \textbf{\MakeUppercase{procedure}} 
                \State choose Kraus decomposition:
                $\{ K_i \}_i \gets \operatorname{KrausUnraveling}(\mathcal{N}, \ket{\psi_0}, q)$
                % \Comment{choose Kraus decomposition}
                \State compute path probabilities:
                $\boldsymbol{p}  \gets \{ \bra{\psi} K^{\dagger}_i K_i \ket{\psi}\}_i$
                % \Comment{path probabilities}
                \State sample a trajectory:
                $i \gets i \sim \boldsymbol{p}$
                % \Comment{sample $i$ with prob $p_i$}
                \State apply Kraus operator:
                $ \ket{\psi'} \gets \frac{1}{\sqrt{p_i}} K_i \ket{\psi}$
                % \Comment{apply Kraus operator}
    \end{algorithmic}
    \end{algorithm}
    \end{figure}

\subsection{Locally entanglement-optimal unraveling}
\label{sec:optimal_single_qubit_unraveling}

    Like previous works, we aim to minimize the entanglement in the trajectories to reduce the computational cost of simulating those with MPS. 
    In this section, we present the method for obtaining a locally entanglement-optimal unraveling.
    We will specify the optimization problem we aim to solve and state our result on the optimality of our unraveling in Section~\ref{sec:optimal_unraveling:setting_result}. 
    The remainder of this section will serve to elucidate the steps of the argument by showing how to map the $n$-qubit problem to a $2$-qubit problem (Section~\ref{sec:effective_problem}), for which one can obtain the optimal decomposition of the mixed state using long-standing results (Section~\ref{sec:optimal_decomp}) and finally how to recover the optimal unitary (the optimal unraveling) from that (Section~\ref{sec:optimal_unitary}).

\subsubsection{Problem statement and main result}
\label{sec:optimal_unraveling:setting_result}

    As their name suggests, bipartite measures of entanglement, as defined in Section~\ref{sec:preliminaries}, are always relative to a partition of the state.
    One could pick any bipartition, e.g., the one between the target qubit and the next, or the half-system partition. 
    Both of these translate to optimizing one specific bond of the MPS.
    In this work, we target the minimization of the entanglement between the noisy qubit and the rest of the state (both right and left). This local optimization objective was previously independently proposed in Ref.~\cite{chenOptimizedTrajectoryUnraveling2024} in the context of a heuristic variational search. Here, we show that this minimization problem admits an exact analytical solution.
    
    Formally, given an initial state vector $\ket{\psi}$ and a Kraus description of a single-qubit noise channel $\{K_i\}_i$, we find the optimal unitary $U_{\mathrm{opt}}$ for the noisy state $\rho = \sum_i K_i \ketbra{\psi}{\psi} K_i^\dagger$ such that 
    \begin{align}
        & E_{\mathrm{av}}(\{ p^{\mathrm{opt}}_i, \ket{\phi^{\mathrm{opt}}_i}_{\mathrm{target:rest}} \}) = E_{\mathrm{oF}}(\rho_{\mathrm{target:rest}}) \nonumber \\
        & \text{with }
        \ket{\phi^{\mathrm{opt}}_i} = \frac{1}{\sqrt{p^{\mathrm{opt}}_i}} K^{\mathrm{opt}}_i \ket{\psi}
        \text{ and } 
        p^{\mathrm{opt}}_i = \| K^{\mathrm{opt}}_i \ket{\psi} \|_2^2 \label{eq:target_optimal_unitary} \\ 
        & \text{for }
        K^{\mathrm{opt}}_i = \sum_j [U_{\mathrm{opt}}]^*_{j,i} K_{j} \, . \nonumber
    \end{align}

    \begin{theorem}
    \label{result:wootters_unraveling}
        Given some state vector $\ket{\psi}$, in a representation where computing its Schmidt decomposition is efficient (e.g., an MPS), 
        \begin{equation}
        \label{eq:schmidt_1_n-1}
            \ket{\psi} = s \ket{u_0}_{\mathrm{target}} \otimes \ket{v_0}_{\mathrm{rest}} + \sqrt{1-s^2} \ket{u_1}_{\mathrm{target}} \otimes \ket{v_1}_{\mathrm{rest}}
        \end{equation}
        and a single-qubit noise channel
        $\mathcal{N}[\cdot] = \sum_i K_i (\cdot) K_i^\dagger$, there exists an efficient and deterministic algorithm that returns the von Neumann optimal local unraveling.
        That is, it returns the unitary $U_\mathrm{opt}$ transforming the Kraus representation
        (according to Eq.~\eqref{eq:unitarity_kraus})
        such that
        \begin{equation}
            U_\mathrm{opt} = \argmin_{U} \left\{ \sum_{i=1}^{r'} p'_{i} E \left(\frac{1}{\sqrt{p'_{i}}} K'_{i} \ket{\psi}_{\mathrm{target:rest}} \right) \right\} \,,
        \end{equation}
        where $E$ is the von Neumann entanglement entropy.
    \end{theorem}
    This unraveling is optimal in that it minimizes the ensemble-averaged von Neumann entanglement entropy between the noisy qubit and the rest of the state by reaching the entanglement of formation. 
    To build up to this result, we will first show how to reduce the problem stated in Eq.~\eqref{eq:target_optimal_unitary} to an 
    effective two-qubit problem.
    With this observation and some long-standing results on the computation of the entanglement of formation at hand, we show that this problem can be solved exactly, resulting in the optimal choice for the unitary freedom.

\subsubsection{Reduction to the effective two-qubit problem}
\label{sec:effective_problem}

    We start from an arbitrary pure state on $n$ qubits and compute the Schmidt decomposition with respect to the cut isolating the target noisy 
    qubit from the remaining $n-1$ qubits, i.e.,  
    \begin{equation}
        \ket{\psi} = s \ket{u_0}_{\mathrm{target}} \otimes \ket{v_0}_{\mathrm{rest}} + \sqrt{1-s^2} \ket{u_1}_{\mathrm{target}} \otimes \ket{v_1}_{\mathrm{rest}} \, .
    \end{equation} 
    Here, we have that $\ket{u_i} \in \CC^2$ and $\ket{v_i} \in \CC^{2^{n-1}}$ for all $i$.
    We are then interested in the action of 
    $\mathcal{N}_{\mathrm{target}} \otimes \mathcal{I}_{\mathrm{rest}}$ on $\ketbra{\psi}{\psi}$,
    where $\mathcal{I}$ is the identity channel.
    Since the noisy map only affects the first qubit, the resulting noisy state remains within Hilbert space of $\CC^2 \otimes \operatorname{span}\{\ket{v_0},\ket{v_1}\}$, which is four-dimensional.
    This is a consequence of the Schmidt decomposition having (at most) two non-zero singular values.
    With this, we can interpret the process as an effective two-qubit evolution. 
    We make it explicit by constructing the isometry which rotates the single-qubit states $\{\ket{0},\ket{1}\}$ to $\{\ket{v_0},\ket{v_1}\}$ as
    \begin{equation}
        V = \ketbra{v_0}{0} + \ketbra{v_1}{1} \, .
    \end{equation}
    With that, we can write
    \begin{equation}
        \ket{\psi} = (\II \otimes V)(
        \underbrace{s \ket{u_0} \otimes \ket{0} + \sqrt{1-s^2} \ket{u_1} \otimes \ket{1}}_{\text{two-qubit state}}
        ) \, .
    \end{equation}
    This isometry could be completed to form a unitary, but it is unnecessary and would even be detrimental to our analysis. 
    Note, however, that for any state in $\CC^2 \otimes \operatorname{span}\{\ket{v_0},\ket{v_1}\}$, the operation $(\II \otimes V)(\II \otimes V^\dagger)$ leaves the state invariant. 
    It acts as the identity on the subspace of interest.
    
    Now let us define the channels $\mathcal{V}[\cdot] = V(\cdot)V^\dagger$ and $\mathcal{V}^\dagger[\cdot] = V^\dagger(\cdot)V$ which respectively map from $\{\ket{0},\ket{1}\}$ to $\{\ket{v_0},\ket{v_1}\}$ and back.
    Let us also slightly abuse notation and use $\mathcal{V}$ for $(\mathcal{I}_{\mathrm{target}} \otimes \mathcal{V}_{\mathrm{rest}})$ and $\mathcal{N}$ for $(\mathcal{N}_{\mathrm{target}} \otimes \mathcal{I}_{\mathrm{rest}})$ when dealing with either the two or $n$ qubit state. 
    Then, the noise channel and the isometric transformation act non-trivially on disjoint sets of qubits (the noise on the first qubit and the isometry on the rest) and therefore commute.
    Consequently, 
    \begin{equation}
        \mathcal{V} \circ \mathcal{N} \circ \mathcal{V}^\dagger [\ketbra{\psi}{\psi}]
        = \mathcal{V} \circ \mathcal{V}^\dagger \circ \mathcal{N} [\ketbra{\psi}{\psi}]
        = \mathcal{N} [\ketbra{\psi}{\psi}] \, .
    \end{equation}
    The first equality is due to the commutation relation, and the second is due to the restriction to the four-dimensional subspace.
    We can now see how this helps us find the optimal decomposition of the $n$-qubit state when the noise channel acts only on a single qubit.
    Defining the effective two-qubit state vector $\ket{\psi^{\mathrm{eff}}} \in \CC^2 \otimes \CC^2$ as
    \begin{equation}
        \ket{\psi^{\mathrm{eff}}} = (\II \otimes V^\dagger)\ket{\psi}
         = s \ket{u_0} \otimes \ket{0} + \sqrt{1-s^2} \ket{u_1} \otimes \ket{1} \, ,
    \end{equation}
    the optimal unraveling of the channel $\mathcal{N}$ for $\ket{\psi^{\mathrm{eff}}}$ is also optimal for $\ket{\psi}$.

    Now, $V$ is in principle a $2 \times 2^{n-1}$ matrix, thus not practical. 
    We can, however, take advantage of the matrix product state structure.
    When starting with an MPS, obtaining the isometry $V$ is almost free: one only has to sever a bond and later stitch it back together.
    Graphically, on an MPS, the Schmidt decomposition above is given by the SVD of a single tensor between the physical and the bond indices (up to bending some wires).
    \begin{equation}
        \label{eq:MPS_SVD_cut}
        \vcenter{\hbox{\includegraphics{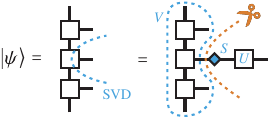}}}
    \end{equation}
    Here $U = \ketbra{u_0}{0} + \ketbra{u_1}{1}$, $S = s \ketbra{0}{0} + \sqrt{1-s^2} \ketbra{1}{1}$ and $V = \ketbra{v_0}{0} + \ketbra{v_1}{1}$.
    The state vector $\ket{\psi^{\mathrm{eff}}}$ is then simply obtained by keeping only the $U$ and $S$ tensors, and bending the index between $S$ and $V$ to act as the effective second qubit
    \begin{equation}
    \label{eq:effective_two_qubit_state}
        \includegraphics{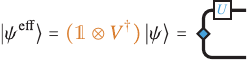}\,.
    \end{equation}
    In that sense, the isometry $V$ is already defined by the rest of the matrix product state and does not need to be calculated (nor applied) explicitly.
    
\subsubsection{Wootters' optimal decomposition}
\label{sec:optimal_decomp}

    The optimal unraveling can then be obtained using Wootters' method for computing optimal convex decompositions of two-qubit states~\cite{woottersEntanglementFormationArbitrary1998}.
    Given a two-qubit state in its eigendecomposition $\rho = \sum_j \lambda_j \ketbra{\lambda_j}{\lambda_j}$, Wootters shows how to find the optimal decomposition $\rho = \sum_i p^{\mathrm{opt}}_i \ketbra{\phi^{\mathrm{opt}}_i}{\phi^{\mathrm{opt}}_i}$ such that 
    \begin{equation}
        E_{\mathrm{av}}(\{ p^{\mathrm{opt}}_i, \ket{\phi^{\mathrm{opt}}_i} \}) = E_{\mathrm{oF}}(\rho)
    \end{equation}
    as well as the unitary as per Eq.~\eqref{eq:freedom_ensembles_from_eigen} that maps from one to the other.
    See Appendix~\ref{sec:wootters_optimal_decomp} for an introduction to Wootters' decomposition and a practical (numerically-friendly) implementation.
    The idea of how to use this to find unravelings of channels is straightforward. 
    We begin by rotating the non-noisy part from the pure state to its single-qubit equivalent subspace, obtaining $\ket{\psi^{\mathrm{eff}}}$ as per Eq.~\eqref{eq:effective_two_qubit_state}.
    Then, we apply the noise to the target qubit, resulting in the mixed state
    \begin{equation}
        \rho^{\mathrm{eff}} = (\mathcal{N} \otimes \mathcal{I})[\ketbra{\psi^{\mathrm{eff}}}{\psi^{\mathrm{eff}}}]
        \, .
    \end{equation}
    From this, one can apply Wootters' method to obtain the optimal convex decomposition
    \begin{equation}
    \label{eq:optimal_effective_decomp}
        \rho^{\mathrm{eff}} = \sum_{i=1}^4 p^{\mathrm{opt}}_i \ketbra{\phi^{\mathrm{eff}}_i}{\phi^{\mathrm{eff}}_i}
    \end{equation}
    where $\ket{\phi^{\mathrm{eff}}_i}$ are two-qubit states and this decomposition minimizes the average entanglement entropy. 
    Rotating back gives an $n$-qubit state
    \begin{equation}
        \rho = (\II \otimes V) \rho^{\mathrm{eff}} (\II \otimes V^\dagger)
        = \sum_{i=1}^4 p^{\mathrm{opt}}_i (\II \otimes V) \ketbra{\phi^{\mathrm{eff}}_i}{\phi^{\mathrm{eff}}_i} (\II \otimes V^\dagger).
    \end{equation}
    Since the rotation is a local operation with respect to the bipartition of interest (it does not affect the target qubit in any way), it also does not change the entanglement between the two subsystems. 
    We have therefore found an optimal decomposition of the noisy $n$-qubit state when the noise channel is only applied to a single qubit.

    In practice, on MPS, after applying the single qubit noise model and finding the optimal decomposition, one can sample one state from the decomposition~\eqref{eq:optimal_effective_decomp} and then contract the non-physical index back with the rest of the MPS from Eq.~\eqref{eq:MPS_SVD_cut}. 
    This is sampling one state vector $\ket{\phi^{\mathrm{eff}}_i}$ and retrieving the sample of the $n$-qubit evolution by computing $\ket{\phi^{\mathrm{opt}}_i} = (\II \otimes V) \ket{\phi^{\mathrm{eff}}_i}$.
    A graphical overview of the complete procedure is shown in Fig.~\ref{fig:optimal-single-qubit-unraveling-procedure}.

    \begin{figure}
        \centering
        \includegraphics[width=\linewidth]{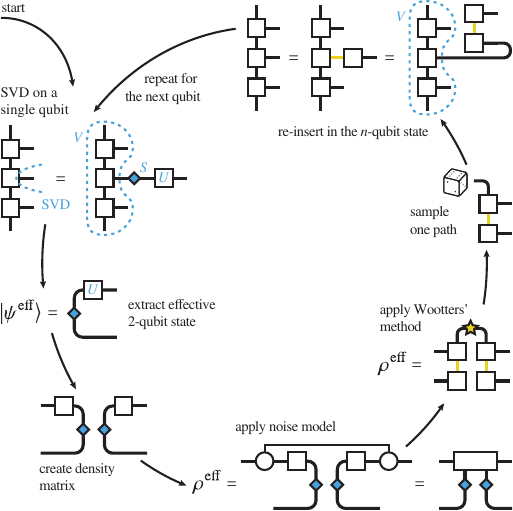}
        \caption{Graphical representation of the procedure to find the optimal single-qubit unraveling.}
        \label{fig:optimal-single-qubit-unraveling-procedure}
    \end{figure}
    
    \begin{table*}[!ht]
        \centering
        \begin{tabular}{l l c | c c c c c}
            \textbf{Work} & & & \textbf{System} & \textbf{Adaptive} & \textbf{Analytical} & \textbf{Cost function} & \textbf{Noise model} \\
            Vovk \& Pichler \# 1 & 2022 & \cite{vovkEntanglementOptimalTrajectoriesManyBody2022} & TITIL & \cmark & \cmark & entropy rate (MC) & any \\
            Kolodrubetz & 2023 & \cite{kolodrubetzOptimalityLindbladUnfolding2023} & RC & \xmark & \xmark & N/A & DF \\
            Cheng \& Ippoliti & 2023 & \cite{chengEfficientSamplingNoisy2023} & RC & \xmark & \cmark & quasientropy & unital \& AD \\
            Chen et al. & 2024 & \cite{chenOptimizedTrajectoryUnraveling2024} & RC & \xmark & \cmark & quasientropy & DF \& AD \\
            Chen et al. & 2024 & \cite{chenOptimizedTrajectoryUnraveling2024} & any & \cmark & \xmark & entropy (NQR) & any\textsuperscript{+} \\
            Vovk \& Pichler \# 2 & 2024 & \cite{vovkQuantumTrajectoryEntanglement2024} & any & \cmark & \xmark & entropy (MC) & any\textsuperscript{+} \\
            Daraban et al. & 2025 & \cite{darabanNonunitarityMaximizingUnraveling2025} & any & \cmark & \xmark & ``unitarity'' & any\textsuperscript{+} \\
            This work & & & any & \cmark & \cmark & entropy (NQR) & any\textsuperscript{+} \\
        \end{tabular}
        \caption{
            Comparison overview of the state-of-the-art choices of unraveling based on several criteria: 
            {\bf System:} (restricted vs general) the kind of system the results apply to (time-independent \emph{translationally-invariant Lindbladian} systems (TITIL), \emph{random circuits} (RC) or arbitrary systems), 
            {\bf Adaptive:}(fixed vs adaptive) whether the unraveling is fixed or adaptive (state-dependent), 
            {\bf Analytical:}(numerical vs analytical) whether the unraveling optimisation (if any) is computed using analytic or numerical methods,
            {\bf Cost function:} what measure of quality of the unraveling is used and optimized (for entropy based measures, the bipartition is specified: MC for middle cut or NQR for noisy qubit vs rest of the system), and 
            {\bf Noise model:} to which single-qubit noise models the results apply (DP: depolarizing, DF: dephasing, AD: amplitude damping, \textsuperscript{+}: applicable even to non-fixed noise models, possibly changing for each qubit and each depth).
            Note that Ref.~\cite{chenOptimizedTrajectoryUnraveling2024} appears twice: once for their analytical fixed unraveling (spin-model mapping) and once for their adaptive heuristic strategy (numerical optimization).
        }
        \label{tab:unraveling_comparison}
    \end{table*}
    
\subsubsection{Optimal unraveling unitary}
\label{sec:optimal_unitary}

    The procedure presented is fully functional and implemented for the numerical simulations shown in Section~\ref{sec:numerical_results}.
    However, Wootters' method for computing the optimal decomposition of a two-qubit state takes the eigendecomposition as input, not any arbitrary decomposition. 
    That is, it returns a unitary $U_{\mathrm{Wootters}}$ such that the decomposition
    \begin{equation}
        \sqrt{p^{\mathrm{opt}}_i} \ket{\phi^{\mathrm{opt}}_i} = \sum_{j=1} [U_{\mathrm{Wootters}}]^*_{j,i} \sqrt{\lambda_j} \ket{\lambda_j}
    \end{equation}
    is optimal.
    This is no major issue since the eigendecomposition of a two-qubit state can be computed, 
    but it does not fully solve the problem stated in \eqref{eq:target_optimal_unitary} of finding the optimal unitary $U_{\mathrm{opt}}$, since in general the induced decomposition from Eq.~\eqref{eq:Kraus_to_ensemble} will not be the eigendecomposition.
    Thus, we want to find the unitary that maps the induced decomposition to the eigendecomposition
    \begin{equation}
        \rho = \sum_i K_i \ketbra{\psi}{\psi} K_i^\dagger = \sum_i p_i \ketbra{\phi_i}{\phi_i} 
        \xrightarrow{U_0} 
        \rho = \sum_j \lambda_j \ketbra{\lambda_j}{\lambda_j} \, .
    \end{equation}
    We will show here that this can also be done efficiently.
    Doing so, we will obtain the optimal unitary to apply directly to the Kraus decomposition of the noise channel instead of only one valid sample from the optimal distribution, thus allowing us to implement the procedure depicted in Fig.~\ref{fig:overview}. 

    For the rest of this section, we will use the matrix notation of ensemble decompositions presented in Appendix~\ref{sec:matrix_formulation}.
    Wootters' method starts with the eigendecomposition of some input state $\rho$, i.e.,
    \begin{equation}
        \rho = \sum_i\lambda_i \ketbra{\lambda_i}{\lambda_i} 
        = v v^\dagger
        \quad \text{with} \quad 
        v = \sum_i \sqrt{\lambda_i} \ketbra{\lambda_i}{i} \, .
    \end{equation}
    However, in general, we are given an arbitrary decomposition of the state
    \begin{equation}
        \rho = \sum_i p_i \ketbra{\phi_i}{\phi_i} 
        = \phi \phi^\dagger
        \quad \text{with} \quad 
        \phi = \sum_i \sqrt{p_i} \ketbra{\phi_i}{i} \, ,
    \end{equation}
    and want to apply Wootters' method.
    The question now is, how we can find a mapping from the arbitrary decomposition $\phi$ to the eigendecomposition $v$, i.e., some unitary $U_0$, such that
    \begin{equation}
    \label{eq:v_to_phi}
        v = \phi U_0^\dagger \, .
    \end{equation}
    To this end, we identify
    \begin{equation}
        v v^\dagger = V \Lambda V^\dagger
        \quad \text{with} \quad 
         V = \sum_{i=1}^r \ketbra{\lambda_i}{i}
        \ \text{ and } \ 
        \Lambda = \sum_{i=1}^r \lambda_i \ketbra{i}{i} \,
    \end{equation}
    and find $v = V \sqrt{\Lambda}$.
    Combining this with the target unitary freedom~\eqref{eq:v_to_phi} results in the equation to satisfy
    \begin{equation}
        V \sqrt{\Lambda} = \phi U_0^\dagger 
        \quad \Leftrightarrow \quad 
        \phi = V \sqrt{\Lambda} U_0 \, ,
    \end{equation}
    which resembles a singular value decomposition of $\phi$. 
    Indeed, we can associate
    \begin{equation}
        \mathrm{SVD}(\phi) 
        = U_\phi^{\vphantom{\dagger}} \Sigma_\phi^{\vphantom{\dagger}} V_\phi^\dagger 
        = \underbrace{V\sqrt{\Lambda}}_{v}U_0
        \quad \Rightarrow \begin{cases}
            V & = U_\phi^{\vphantom{\dagger}} \, , \\
            \Lambda & = \Sigma^2 \, , \\
            U_0 & = V_\phi^\dagger \, .
        \end{cases}
    \end{equation}
    From this, we find that the optimal unitary freedom to apply to the Kraus decomposition of the channel to solve~\eqref{eq:target_optimal_unitary} is given by $U_{\mathrm{opt}} = U_{\mathrm{Wootters}} U_0$.

\subsection{Comparing unraveling strategies}
\label{sec:unraveling_comparison}

    In recent years, several works have explored trajectory sampling algorithms, aiming to reduce computational cost and ideally find or improve noise thresholds, resulting in efficient classical simulations.
    As mentioned above, sampling a trajectory from an evolution under a quantum channel admits a degree of freedom in the decomposition to sample from (the choice of unraveling) and 
    the precise value of the heralded error from Theorem~\ref{result:sampled_trajectories} strongly depends on this choice. 
    For the same evolution, some unravelings result in large truncation errors while others produce states that can be well represented as MPS.
    Some motivating examples of the impact of this choice of unraveling are presented in Appendix~\ref{sec:unitary_unraveling_choice}. 
    In this section, we present and compare the state-of-the-art choices of unraveling. 
    An overview of the unravelings considered here is presented in Table~\ref{tab:unraveling_comparison}.
    
    We say that statements about an unraveling are \emph{general} if they hold for any evolution within the setting considered in Section~\ref{sec:MPS_based_details}, while they are \emph{system-restricted} if further restricting assumptions are made. 
    That would be the case for average-case results over brickwork circuits of Haar random gates, or if the statements only hold for translationally invariant systems.
    We call an unraveling choice \emph{fixed} if it is independent of the current state of the system. 
    Nonetheless, one is generally not restricted to unraveling the noise in the same way at every application of the noise.
    An unraveling choice is said to be \emph{adaptive} if it depends on the state to which the noise is applied.
    We also distinguish methods that deterministically find a cost-function optimal unraveling
    from those obtained from a numerical optimization without guarantees (e.g., gradient descent).
    This is particularly relevant for the adaptive unravelings, where the optimization happens at every application of the noise.
    We also emphasized whether the obtained unraveling is computed by \emph{analytical} methods or by numerical optimization and describe the target cost function used to optimize the unraveling. Finally, we specify the noise models that can be handled by each of the unraveling methods. 

    Physically motivated or handpicked unravelings have been studied, in particular in the setting of time-independent translationally invariant Lindbladian systems~\cite{vovkEntanglementOptimalTrajectoriesManyBody2022} or noisy random circuits~\cite{kolodrubetzOptimalityLindbladUnfolding2023}. 
    Different fixed unravelings have been shown numerically to result in significantly different performance, in particular different growths of average entanglement.
    To provide analytical insights on how to choose a good fixed unraveling, approximate mappings of random circuits to classical spin models have been leveraged~\cite{chengEfficientSamplingNoisy2023, chenOptimizedTrajectoryUnraveling2024}, resulting in fixed unravelings with improved performance for random circuits. 
    However, these methods really optimize a quasientropy which is only indirectly related to the cost of the simulation.
    On the other hand, adaptive unravelings have first been considered in restricted settings (time-independent translationally invariant Lindbladian systems), where it was shown that to minimize the (continuous) entanglement change rate it is sufficient to consider only two choices (homodyne and number measurements), resulting in an efficient greedy optimization algorithm~\cite{vovkEntanglementOptimalTrajectoriesManyBody2022}.
    Adaptive unravelings can also be computed with numerical methods, like gradient descent on the unitary freedom in Eq.~\eqref{eq:unitarity_kraus}.
    Versions thereof have been implemented to pick adaptive unravelings minimizing different cost functions, like entanglement entropies~\cite{vovkQuantumTrajectoryEntanglement2024, chengEfficientSamplingNoisy2023} or measures of non-unitarity~\cite{darabanNonunitarityMaximizingUnraveling2025}.

    Adaptive unravelings offer greater freedom than fixed ones, potentially leading to better performance. 
    In particular, in addition to their analytical fixed unraveling, \cite{chenOptimizedTrajectoryUnraveling2024} also present an adaptive heuristic strategy that employs the same cost function as our work but uses numerical optimization to compute the unraveling.
    However, numerical approaches such as gradient decent are often non-deterministic, have the risk of not finding a good minimum for their cost function and can be significantly slower to implement on each of the numerous calls of the unraveling subprocedure.
    Thus, it would be valuable from both practical and conceptual perspectives to find analytically motivated choices of adaptive unraveling.
    Our unraveling achieves precisely that.
    We provide an efficient adaptive unraveling strategy with analytical guarantees for any single-qubit noise channel.
    Our unraveling can be computed deterministically in constant time 
    (as opposed to the $\mathcal{O}(\chi^2)$ or $\mathcal{O}(\chi^3)$ scaling of previous adaptive unravelings)
    and optimally disentangles the target qubit from the rest of the state.

\section{Application to specific systems}
\label{sec:restricted_settings}

    In this section, we discuss the application of trajectory sampling methods and the impact of our results on two frequently considered settings: time-independent Lindbladian systems (Section~\ref{sec:time_indep_Lindblad}) and noisy random circuits (Section~\ref{sec:random_circuits}). 
    In particular, we define a measure of unitarity of (Kraus) operators, demonstrate how to compute an unraveling that minimizes this measure, and show that this fixed unraveling is near-optimal for random circuits.

\subsection{Time-independent Lindbladian systems}
\label{sec:time_indep_Lindblad}

    Our method can also be applied to Lindblad equations with single-qubit jump operators.
    In this section, we discuss how to construct a quantum circuit that matches the language above, starting with Trotterizing the evolution of the open system.
    Then, the choice of unraveling can be directly applied to the resulting Kraus operators.
    
    The Lindbladian description of the evolution of open quantum systems 
    is given by
    equations of motion of the form \begin{equation}
        \frac{\partial \rho}{\partial t} 
        = -\ii[H, \rho] + \sum_{j} \gamma_j \left( c_j \rho c_j^\dagger -\frac{1}{2} \left\{ c_j^\dagger c_j, \rho \right\}\right)
    \end{equation}
    where $H$ is the Hamiltonian governing the coherent evolution and $\{ c_j \}_j$ are the jump operators.
    We explicitly show in Appendix~\ref{sec:Lindblad_to_circuits} how to compute the noisy quantum circuit corresponding to a Trotterization of the evolution above exactly. 
    For a coherent part composed of two-local nearest-neighbor 
    interactions
    \begin{equation}
        H = \sum_{i=1}^n h_{i, i+1}
        \quad \text{with} \quad 
        h_{i, i+1} \in \CC^{4 \times 4} \ \forall i 
        \, ,
    \end{equation}
    the gates from the resulting brickwork circuit are simply $U_{\mathrm{coh}} = \ee^{-\ii h dt}$.
    For the incoherent part, the Kraus operators of the corresponding noise channel are obtained by computing the eigendecomposition of the Choi state of the local jump process. 
    The Choi state in this setting is given by a reshaping of the vectorized jump superoperator
    $\ee^{L_j dt}$
    with 
    \begin{equation}
    L_j = \gamma_j \left( c_j \otimes c_j^* -\frac{1}{2} c_j^\dagger c_j \otimes \II -\frac{1}{2} \II \otimes c_j^\top c_j^* \right).
    \end{equation}
    This standard procedure returns an exact (and orthogonal) Kraus decomposition of the noise channel resulting from the Trotterized jump process.
    
    We also show in Appendix~\ref{sec:noise_jump_equivalence} that, for a jump rate $\gamma$ and Trotterization time step $dt$, choosing 
    $c = \left( \begin{smallmatrix} 1 & 0 \\ 0 & 0 \end{smallmatrix}\right)$
    results in a dephasing channel with rate $p^{\mathrm{DF}} = 1 - \ee^{-\gamma dt/ 2}$
    and choosing 
    $c = \left( \begin{smallmatrix} 0 & 1 \\ 0 & 0 \end{smallmatrix}\right)$
    results in an amplitude damping channel with noise-rate $\gamma^{\mathrm{AD}} = 1 - \ee^{- \gamma dt}$.
    Additionally, the resulting Kraus decomposition coincides with the decomposition typically found in textbooks (labeled in the remainder as {mixed-unitary} decomposition for dephasing and {Orthogonal} decomposition for amplitude damping, see also Table~\ref{tab:unraveling_comparison}).

\subsection{Noisy random circuits}
\label{sec:random_circuits}

    In this section, we look at the example of random quantum circuits. 
    This setting has been heavily studied in the trajectory unraveling literature and the related field of measurement-induced phase transitions. 
    We show in Section~\ref{sec:random_circuits_unraveling} how our unraveling fits in this literature by explaining how previous unravelings can be shown to be near-optimal using our arguments based on entanglement of formation.
    This is done by observing that, when restricting to random circuits, our protocol can be simplified by approximating the effective two-qubit states with a Bell state (maximally entangled state on two qubits). 
    With this, we give analytical footing to the proxy proposed in Ref.~\cite{kolodrubetzOptimalityLindbladUnfolding2023} and obtain that the unraveling choices presented by Chen et al.~\cite{chenOptimizedTrajectoryUnraveling2024} and Cheng et al.~\cite{chengEfficientSamplingNoisy2023} happen to be optimal in our sense for the noisy Bell state.
    We also address one open question from Ref.~\cite{chengEfficientSamplingNoisy2023} about the improvement of thresholds when using locally and adaptively optimized unravelings. 
    In the setting of Haar random circuits, such improvement appears to be minimal, as the two strategies converge.
    In Section~\ref{sec:least_unitary_unraveling}, we then relate this optimality statement to the ``projectiveness'' or ``unitarity'' of the obtained Kraus decompositions.

\subsubsection{Near-optimal fixed unraveling for random circuits}
\label{sec:random_circuits_unraveling}

    This argument is based on the observation that states in random circuits are locally approximately maximally entangled. 
    As presented in Section~\ref{sec:optimal_single_qubit_unraveling}, the problem of finding a good unraveling for single-qubit channels can be reduced to a two-qubit problem. 
    The Schmidt coefficients in Eq.~\eqref{eq:schmidt_1_n-1} characterize the entanglement of the initial effective two-qubit state.

    If the singular values are equal (thus both $1/\sqrt{2}$), then the state is effectively maximally entangled. 
    This case may seem very specific and unlikely, but it is close to what happens for Haar random states.
    Indeed, it is already well known that, for sufficiently large Haar random states, any subset of qubits is ``more or less maximally entangled'' with the rest of the system.
    For a Haar random state, the distribution of eigenvalues of the reduced density matrix (squared singular values) follows a Marchenko-Pastur distribution~\cite{marcenkoDistributionEigenvaluesSets1967, znidaricEntanglementRandomVectors2006}.
    For states drawn from a $2$-design ensemble~\cite[Example 50]{meleIntroductionHaarMeasure2024} and a single target qubit, we have
    \begin{equation}
        \operatorname*{\mathds{E}}_{\ket{\psi} \sim \nu} \left[ \Tr{}{\rho_1^2}\right] = \frac{2+2^{n-1}}{2^n + 1}=\frac{1}{2}+\mathcal{O}\left(\frac{1}{2^n}\right)
    \end{equation}
    where $\nu$ is a $2$-design ensemble and $\Tr{}{\rho_1^2} = \lambda_1^2 + \lambda_2^2$ is the purity of the reduced density matrix on that qubit.
    Now, since $\Tr{}{\rho_1} = \lambda_1 + \lambda_2 = 1$, we have 
    $\lambda_i = \frac{1}{2} + \mathcal{O}\left( \frac{1}{\sqrt{2^n}} \right)$.
    Since the eigenvalues of the reduced density matrix are the squared Schmidt coefficients, 
    \begin{equation}
        s_i = \sqrt{\lambda_i}= \frac{1}{\sqrt{2}} + \mathcal{O}\left( \frac{1}{\sqrt{2^n}} \right) \, .
    \end{equation}
    Interpreting the rest of the $n-1$ qubits as a single effective qubit as we did in Section.~\ref{sec:effective_problem}, we obtain
    \begin{equation}
        \ket{\psi^{\mathrm{eff}}} \approx \frac{1}{\sqrt{2}}(\ket{u_1} \otimes \ket{v_1} + \ket{u_2} \otimes \ket{v_2}) \, ,
    \end{equation}
    which is a maximally entangled state.
    Furthermore, in the case of two degenerate Schmidt coefficients $s_1 = s_2 = 1/\sqrt{2}$, the local basis is irrelevant and one can enforce that one of the two Schmidt bases is the computational basis. 
    Using this, we choose the form
    $\ket{\psi} \approx \frac{1}{\sqrt{2}}(\ket{0} \otimes \ket{w_1} + \ket{1} \otimes \ket{w_2})$. 
    This can be seen from the singular value decomposition of the matrix of coefficients (from which the Schmidt decomposition is computed), as the matrix of singular values is proportional to the identity, and thus commutes with the two basis-defining unitaries.
    With this, the case $s_1 = s_2 = 1/\sqrt{2}$ reduces to considering $\ket{\psi^{\mathrm{eff}}} = \ket{\Phi^+}$.
    \begin{equation*}
        \includegraphics{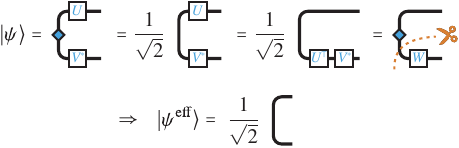}
    \end{equation*}

    When dealing with random circuits, one can simplify the unraveling procedure: 
    Computing the optimal unraveling of the noise channel of interest for the Bell state once.
    We call this unraveling \emph{Haar Optimal}.
    This yields a near-optimal unraveling in a state-independent manner.
    We show in Appendix~\ref{sec:choi_optimal_kraus_examples} that the unravelings proposed in Refs.~\cite{chenOptimizedTrajectoryUnraveling2024} and~\cite{chengEfficientSamplingNoisy2023} are optimal in this sense.
    Our argument offers a different intuition for why the Kraus decompositions found in the literature perform well in unraveling noisy random circuits based on MPS.
    Additionally, if one lacks information about the state and consequently assumes it to be random, this decomposition is the best estimate.
    In a sense, they are the best state-agnostic choice of unraveling.

\subsubsection{Least unitary unraveling}
\label{sec:least_unitary_unraveling}

    As a consequence of the result in the previous section, we also find the (state-independent) unraveling, which can be interpreted as being as far as possible from a mixed-unitary unraveling. 
    Finding the optimal unraveling for the Bell state (and thus also for random circuits) is equivalent to computing the optimal decomposition of the Choi state of the noise channel (see Appendix~\ref{sec:optimizing_Kraus_projectiveness} for details). 
    One can then ask about input-independent properties of unravelings.
    In particular, single-qubit channels are special in that they cannot create entanglement and only reduce entanglement by projecting out part of the state.
    Multi-qubit channels, on the other hand, can disentangle states in a unitary way.
    We ask if one can find the Kraus decomposition of a single-qubit channel that is as far from being mixed-unitary as possible, i.e., where each Kraus operator is as close to rank-$1$ as possible.

    \begin{corollary}[Least unitary unraveling]
    \label{result:average_projective_unraveling}
        Let us quantify the unitarity of an operator $K$ by the uniformity of its normalized singular values $\{ \sigma_i / \norm{K}_F\}$, measured by the Shannon entropy
        \begin{equation}
        \label{def:operator_unitarity}
            \operatorname{unitarity}(K) = - \sum_{i=1}^d \left( \frac{\sigma_i}{\norm{K}_F} \right)^2 \log_2 \left( \frac{\sigma_i}{\norm{K}_F} \right)^2 \, .
        \end{equation}
        Then, given a single-qubit channel as in Eq.~\eqref{eq:def_kraus}, the locally entanglement-optimal unraveling for the Bell state is also the Kraus decomposition which minimizes the average unitarity of the operators
    \begin{equation}
        \operatorname{unitarity}_{\mathrm{av}}(\{ K_i \}) = \frac{1}{d} \sum_i \norm{K_i}^2_F \operatorname{unitarity}(K_i) \, .
    \end{equation}
    \end{corollary}
    The normalization with the Frobenius norm guarantees that we compute the entropy of a valid probability distribution.
    Note that the unitarity, as defined above, reaches its maximal value of $1$ for unitary operators and the minimum $0$ for rank-$1$ operators, e.g., rank-$1$ projectors. 
    A longer reflection on the choice of definition as well as the full derivation of this result can be found in Appendix~\ref{sec:optimizing_Kraus_projectiveness}.
    During the writing of this manuscript, we were made aware of Ref.~\cite{darabanNonunitarityMaximizingUnraveling2025} where a different notion of (non-) unitarity is introduced, which is then optimized numerically for an adaptive unraveling.

    \begin{figure*}[t!]
        \centering
        \includegraphics[width=\textwidth]{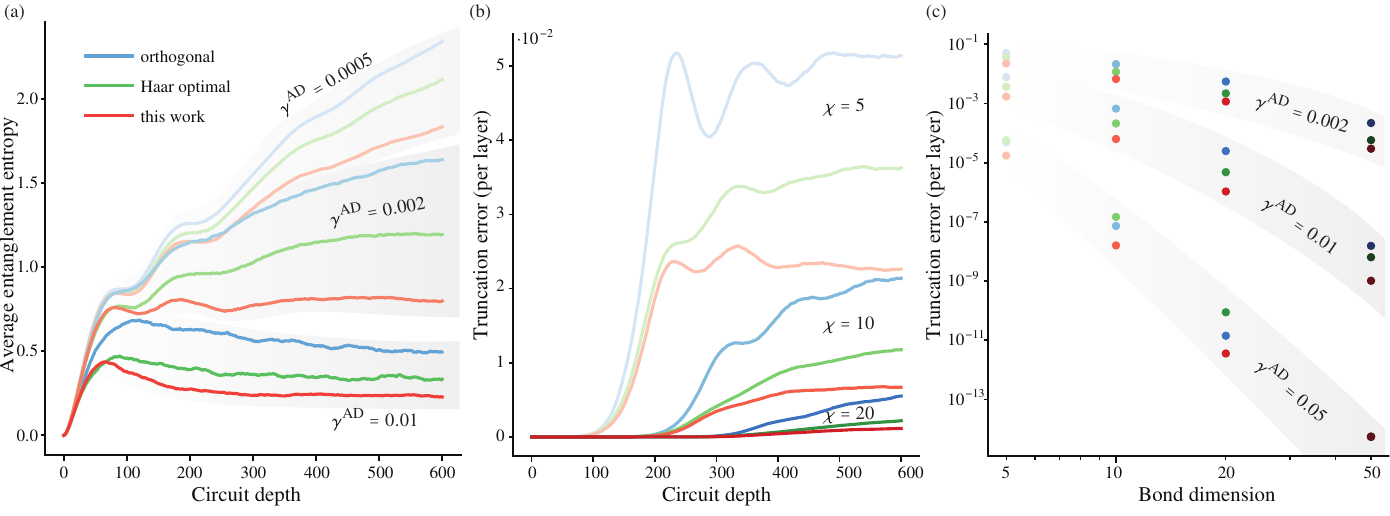}
        \caption{
        Simulation of a $16$ qubit system evolving under a Heisenberg Hamiltonian $H = \sum_i Y_iY_{i+1} + 0.35 X_i + 0.35 Y_i + 0.5 Z_i$ (open boundary conditions) with amplitude damping noise. 
        All results are an average of $500$ independently sampled trajectories and repeated for the \emph{Orthogonal} and \emph{Haar Optimal} fixed unravelings (see Table~\ref{tab:unravelings}) as well as our locally entanglement-optimal unraveling.
        (a) Average entanglement as a function of circuit depth plotted for various noise rates and unravelings in a simulation with bond dimension $\chi = 50$.
        (b) Average upper bound on truncation error (Eq.~\eqref{eq:MPS_truncation_error}) as a function of circuit depth, plotted for various bond dimensions in a simulation with noise rate $\gamma^{\mathrm{AD}} = 0.002$.
        (c) Scaling of the truncation error (at depth $600$) with bond dimension for various noise rates.
        }
        \label{fig:Fixed_Hamiltonian}
    \end{figure*}
    
    \begin{table*}[t!]
        \begin{tabular}{lccc}
            \textbf{Noise model} & 
            \textbf{Orthogonal} & 
            \textbf{Projective} & 
            \textbf{Haar Optimal} \\ 
            \vspace{0.15cm}
            Dephasing  & 
            $\Big \{ \sqrt{1-\frac{p}{2}} \II$, $\sqrt{\frac{p}{2}} Z \Big \}$ & 
            $\big \{ \sqrt{1-p} \II$, $\sqrt{p} \ketbra{0}{0}$, $\sqrt{p} \ketbra{1}{1} \big \}$ & 
            $U = H $ \\ 
            \vspace{0.50cm}
            Depolarizing & 
            \makecell[cc]{
                $\Big \{ \sqrt{1-\frac{3p}{4}} \II$, $\frac{\sqrt{p}}{2} X$, \\ $\frac{\sqrt{p}}{2} Y$, $\frac{\sqrt{p}}{2} Z \Big \}$
            } & 
            \makecell[cc]{
                $\big \{ \sqrt{1-p} \II$, $\sqrt{\frac{p}{2}} \ketbra{0}{0}$, $\sqrt{\frac{p}{2}} \ketbra{0}{1}$, \\ $\sqrt{\frac{p}{2}} \ketbra{1}{0}$, $\sqrt{\frac{p}{2}} \ketbra{1}{1} \big \}$ 
            } & 
            $U = H \otimes H$ \\
            \vspace{0.15cm}
            Amplitude damping & 
            $ \Big \{ \left(\begin{smallmatrix} 1 & 0 \\ 0 & \sqrt{1-\gamma} \end{smallmatrix} \right)$, $\left( \begin{smallmatrix} 0 & \sqrt{\gamma} \\ 0 & 0 \end{smallmatrix} \right) \Big \}$ & 
            --- &
            $U = H$ \\
        \end{tabular}
        \caption{
            Overview of some unravelings considered in the literature and implemented in the numerical simulations. 
            We also call the {Orthogonal} unraveling of the depolarizing and dephasing channels \emph{mixed-unitary}.
            The {Haar Optimal} unraveling is obtained by rotating the {Orthogonal} unraveling according to Eq.~\eqref{eq:unitarity_kraus} using the specified unitary, where $H$ is the Hadamard matrix $H = \frac{1}{\sqrt{2}} \left( \begin{smallmatrix} 1 & 1 \\ 1 & -1 \end{smallmatrix} \right)$. 
            More information can be found in Appendix~\ref{sec:common_unravelings}.
        }
        \label{tab:unravelings}
    \end{table*}
    
\section{Performance}
\label{sec:performance}
    In this section, we discuss the performance of our unraveling choice relative to other quantum trajectory sampling methods and in the broader context of classical simulation algorithms for noisy quantum dynamics.
    We first present numerical results for Lindbladian systems and random circuits and compare the numerical performance of our unraveling to alternative unravelings.
    To convey a sense of the scale of noise rate sufficient for efficient simulability using our algorithm, we discuss known noise thresholds from the trajectory sampling and measurement-induced phase transition literature. 
    Finally, we compare the family of methods based on sampling of trajectories with other simulation methods, including the Pauli path method and other tensor network methods based on matrix product operators and locally purified density operators.

\subsection{Numerical performance comparison}
\label{sec:numerical_results}

    \begin{figure*}[t!]
        \centering
        \includegraphics[width=\textwidth]{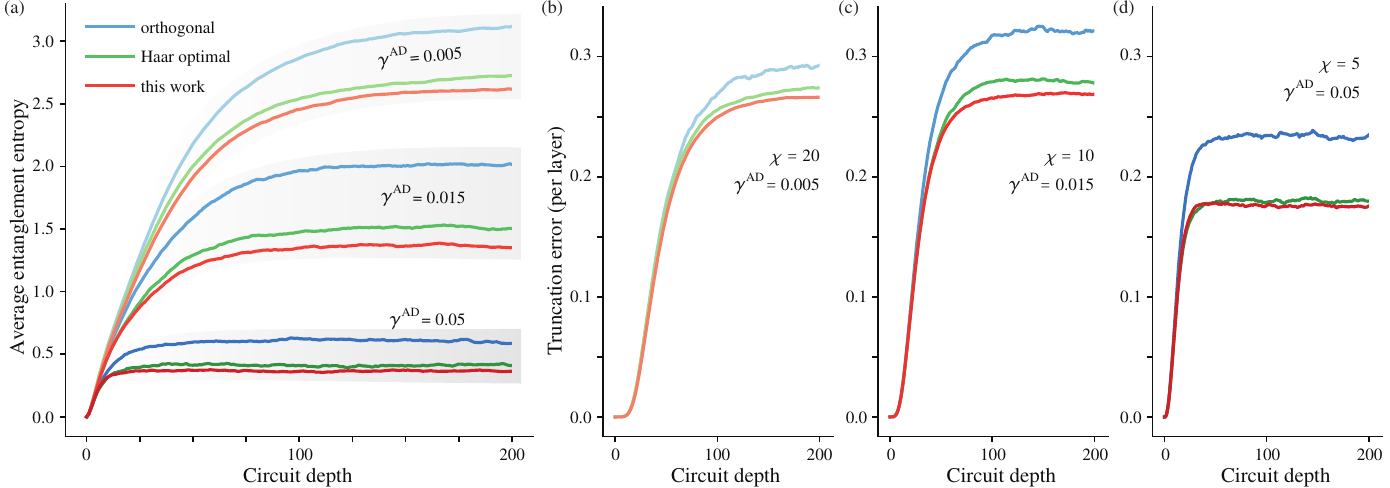}
        \caption{
            Simulation of a $12$ qubit open quantum system under random translationally invariant and time-independent Lindbladian evolutions with amplitude damping noise. 
            For each random evolution, we consider a brickwork circuit composed of a fixed gate repeated at each site.
            The gate itself is obtained by a short-time unitary evolution $U=\ee^{-\ii H dt}$ of a random Hamiltonian $H$ (a random Hermitian matrix with normally distributed real and imaginary parts of each entry) and a time step $dt=0.01$.
            Each trace is the average of 1000 independently computed trajectories truncated to a fixed bond dimension.
            (a) Average entanglement as a function of circuit depth plotted for various noise rates and unravelings in a simulation with bond dimension $\chi = 20$.
            (b) to (d) Average upper bound on truncation error (Eq.~\eqref{eq:MPS_truncation_error}) as a function of circuit depth, plotted for various combinations of bond dimensions and noise rates.
            }
        \label{fig:numerics_random_Hamiltonians}
    \end{figure*}

    To verify the performance of our algorithm,
    we present numerical simulations comparing our choice of unraveling to other unraveling algorithms under various settings. 
    In  Fig~\ref{fig:Fixed_Hamiltonian}, we considered an open quantum system evolving under a particular Hamiltonian, while in Fig~\ref{fig:numerics_random_Hamiltonians}, we present aggregated results based on many randomly generated Hamiltonians. In line with the literature on the simulation of noisy quantum circuits, in Fig~\ref{fig:numerics_random_circuits} we simulated random one-dimensional quantum circuits with independent local noise and variants thereof. 

    The various noisy quantum dynamics are simulated using our locally entanglement-optimal unraveling algorithm and three other types of 
    unraveling algorithms that we refer to as the \emph{Orthogonal}, \emph{Projective} and \emph{Haar Optimal} unravelings. Here, our numerical comparison focuses on three commonly considered noise models (depolarizing, dephasing and amplitude damping), Table~\ref{tab:unravelings} provides a Kraus decomposition associated with these unravelings (see Appendix~\ref{sec:common_unravelings} for more details).
    We point out that this comparison alone does not fully capture the performance advantages of our unraveling, which applies to arbitrary single-qubit noise models. In contrast, an explicit construction of the {Haar Optimal} unraveling
    —the more competitive of the three alternatives—
    is only known for specific noise channels and lacks this level of generality.

    Kraus decompositions can always be constructed from ensemble decompositions of the Choi state of the CPTP map~\cite[Section 6.4]{renesQuantumInformationTheory2022a},~\cite[Section 2.2]{watrousTheoryQuantumInformation2018}. The {Orthogonal} Kraus decomposition is defined as the set of operators obtained from the eigendecomposition, which are necessarily orthogonal.
    In the case of unital noise channels such as the depolarizing and dephasing noise, the eigenoperators are unitary and hence we call these \emph{mixed-unitary} unravelings. In the case of amplitude damping noise, the eigenoperators are not unitary. 
    This is the representation typically given in standard expositions of quantum noise for the three noise channels considered here.
    
    We also consider the \emph{Haar Optimal} unraveling, obtained by rotating the {Orthogonal} Kraus representation using the Hadamard matrix as the unitary freedom in Eq.~\eqref{eq:unitarity_kraus}. This unraveling is derived in Refs.~\cite{chengEfficientSamplingNoisy2023, chenOptimizedTrajectoryUnraveling2024} via an approximate mappings of random circuits to classical spin models. As discussed in Section~\ref{sec:random_circuits}, our locally entanglement-optimal unraveling also becomes the {Haar Optimal} unraveling when we restrict to Haar random states.
    Finally, in the case of depolarizing and dephasing noise, we also consider the \emph{Projective} unravelings defined in Table~\ref{tab:unravelings}.

    While Ref.~\cite{chenOptimizedTrajectoryUnraveling2024} also introduces an adaptive heuristic strategy based on minimizing local entanglement, the source code for their optimization routine is not publicly available and we do not include a direct numerical comparison here. However, because our method provides the exact analytical solution to the local entanglement minimization problem approximated by their heuristic, we expect the resulting trajectories 
    to be identical (assuming the heuristic successfully finds the optimum). Our analytical approach finds the optimum deterministically with $O(1)$ overhead, avoiding the convergence uncertainties and computational cost of numerical optimization.
    We focus our comparison to analytically optimal unraveling methods including the Haar optimal unraveling which is state-of-the-art for the noise models we consider.

\subsubsection{Time-independent Lindbladian systems}

    In Fig.~\ref{fig:Fixed_Hamiltonian} we demonstrate that there are settings where our locally entanglement-optimal unraveling results in a significant reduction of entanglement cost and truncation errors. Using a Heisenberg Hamiltonian $H = \sum_i Y_iY_{i+1} + 0.35 X_i + 0.35 Y_i + 0.5 Z_i$ with amplitude damping noise, we observe a marked performance improvements when going from the {Orthogonal} to the {Haar Optimal} unraveling and again going from the {Haar Optimal} to our unraveling. 
    We see this pattern more generally across many instances of Hamiltonians, although sometimes the separations can be less striking. 
    In Fig.~\ref{fig:numerics_random_Hamiltonians}, we demonstrate this by considering the entanglement and truncation error behavior of competing unravelings when averaged over many randomly selected Hamiltonians. As per Fig.~\ref{fig:Fixed_Hamiltonian}, we observe a hierarchy of performance where our unraveling improves on the {Haar Optimal} one, which itself outperforms the {Orthogonal} one. 
    We note that performance improvements within this hierarchy tend to be less marked for extremal levels of noise and bond dimension. As expected, our performance improvement over the {Haar Optimal} unraveling diminishes in settings with low noise and strongly coupled non-integrable Hamiltonians where the local entanglement structure of the system tends to be similar to that of Haar random states.

    We note that for some instances, our unraveling does not result in the best performance. 
    This may be explained by the greediness of our algorithm.
    As we will discuss in Section~\ref{sec:limitations}, a somewhat sub-optimal local sampling (such as given by the fixed {Haar Optimal} unraveling) could result in more favorable initial conditions for the unraveling in the next time step.
    
    \begin{figure*}[!ht]
        \centering
        \includegraphics[width=\textwidth]{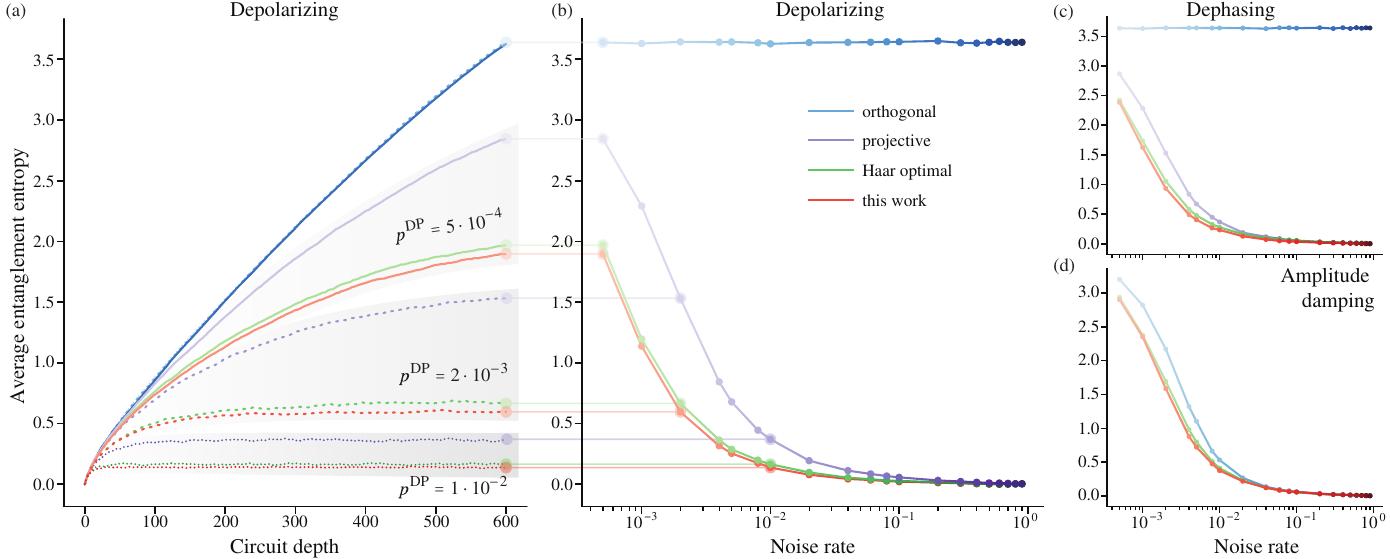}
        \caption{
            Average entanglement entropy across the middle bond of the MPS in the numerical implementation of the trajectory unraveling algorithm to different random circuits on 12 qubits with noise for different choices of unraveling. 
            Each run results from one unraveled trajectory of a newly sampled circuit, and each data point is the average over 1000 runs. 
            For all settings, we computed the average entanglement after each layer. 
            Panel (a) shows the growth of the entanglement with circuit depth for some noise rates of depolarizing noise on low-entangling random circuits. 
            Note that the shaded area labeling noise rates excludes the {mixed-unitary} unraveling, as those curves overlap for all noise rates.
            Panels (b) to (d) show the value of the entropy after the last layer ({600}) as a function of noise rates for
            (b) low-entangling random circuits with depolarizing noise (as in (a)), 
            (c) dephasing noise and 
            (d) low-entangling random circuits with local random rotations with amplitude damping noise. In these settings, the final value at depth {600} is representative of the relation between the different unravelings for all depths, as in panel (a).
            The entangling angle is set at $0.05$ for all simulations.
        }
        \label{fig:numerics_random_circuits}
    \end{figure*}

\subsubsection{Noisy random circuits}

    \begin{figure}
        \centering
        \includegraphics[scale=1]{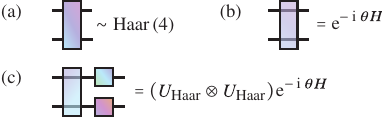}
        \caption{
            Construction of the $2$-qubit unitaries for the three kinds of random circuits considered in the numerical simulations. 
            (a) Two-qubit unitaries drawn from the Haar distribution. 
            (b) Little-entangling gates $U(H, \theta) = \ee^{-\ii \theta H}$ where $H$ is a Hermitian matrix with real and imaginary part of the entries drawn from a normal distribution with mean $0$ and standard deviation $1$. 
            The rotation angle $\theta$ is picked to be small, e.g., $0.1$.
            (c) Little-entangling gates with local rotation $U_{\mathrm{tot}} = (U_{\mathrm{Haar}} \otimes U_{\mathrm{Haar}}) U(H, \theta)$ where the local (single-qubit) unitaries are drawn from the Haar distribution while the entangling gate is constructed like in (b).
        }
        \label{fig:random_circuit_blocks}
    \end{figure}
    
    In the case of noisy quantum circuits, we consider one-dimensional circuits with brickwork layouts of two-qubit unitary gates, each followed by single-qubit noise channels on the affected qubits. 
    The unitary gates, as shown in Fig.~\ref{fig:random_circuit_blocks}, are either drawn from the Haar random distribution, small rotations around the axis defined by a random Hermitian matrix, or the latter followed by local Haar random unitaries.
    A graphical representation of such circuits can be found in Fig.~\ref{fig:overview} (top left).
    We simulated systems of {12} qubits with {600} layers of unitary gates and noise channels. 
    For each noise rate, we show the average of {1000} trajectories. 
    We repeated such simulations for depolarizing, dephasing, and amplitude-damping noise models.

    Let us first discuss the case where each gate is drawn randomly from the Haar distribution. As discussed in Section~\ref{sec:random_circuits}, for Haar random states, our unraveling and the {Haar Optimal} unraveling become equivalent. As a consequence, in this Haar random gates setting we expect and observe these two unraveling to perform similarly. For all three noise models considered here, the simulations yielded similar ensemble-averaged entanglement entropy values for our adaptive unraveling and the {Haar Optimal} one (see Appendix~\ref{sec:auxiliary_numerics}).
    For the two unital channels (depolarizing and dephasing), the {Orthogonal} unraveling is a mixed-unitary unraveling and hence does not affect the entanglement growth, while our unraveling perform best.
    As previously shown in Refs.~\cite{kolodrubetzOptimalityLindbladUnfolding2023, chengEfficientSamplingNoisy2023, chenOptimizedTrajectoryUnraveling2024}, we observed that the {Projective} unravelings result in smaller values of ensemble-averaged entanglement entropy than the {mixed-unitary} unravelings but under-perform against our unraveling. 
    For amplitude damping, our unraveling also outperforms the {Orthogonal} one.
    
    We now consider the case of small rotations of random circuits, where each gate is a small-angle rotation around a randomly generated Hermitian matrix (as per Fig.~\ref{fig:random_circuit_blocks}).
    For such circuits, when the acting noise is depolarizing or dephasing, our  unraveling shows consistently lower values of entanglement entropies than all other unravelings, including the {Haar Optimal} one, as shown in Fig.~\ref{fig:numerics_random_circuits} (a) to (c).
    First, note that, again, the {mixed-unitary} unraveling exhibits no dependence on noise rate. 
    Then, in most simulations, we observed that the {Haar Optimal} unraveling outperformed the {Projective} unraveling (or {Orthogonal} unraveling for amplitude damping).
    In particular, note that the saturation value of the entanglement when using the {Projective} unraveling for a depolarizing rate of $\gamma = 5 \cdot 10^{-3}$ is nearly the same as that of the {Haar Optimal} unraveling for $\gamma = 2 \cdot 10^{-3}$.
    In this case, an evolution with $2$ times less noise can be simulated at the same ``entanglement cost''.
    The {Haar Optimal} unraveling itself is being outperformed by our locally entanglement-optimal unraveling. 
    
    Surprisingly, the improvement from the locally entanglement-optimal over {Haar Optimal} unraveling was not to be found for amplitude damping in the low-entangling random circuits (see Appendix~\ref{sec:auxiliary_numerics}).
    We associate this observation with the fact that the {Haar Optimal} unraveling is not only optimal for Bell states under amplitude damping noise but for any state whose Schmidt decomposition is in the computational basis (see Appendix~\ref{sec:opt_AD_computational}).
    Since the small rotation unitaries do not rotate the state much, and the noise pushes it towards the all-zero state, the states remain close to states with Schmidt decompositions in the computational basis.
    This hypothesis is confirmed by slightly changing the circuit definition to include Haar random single-qubit rotations after each two-qubit small-angle unitary, as shown in Fig.~\ref{fig:random_circuit_blocks} (c).
    The single-qubit gates do not affect the entanglement between qubits, but they randomize the basis of the Schmidt decomposition.
    For random circuits composed of such gates, our unraveling again outperforms the {Haar Optimal} unraveling, as shown in Fig.~\ref{fig:numerics_random_circuits} (d).

    We also note that the depth of the simulation is insufficient to reach the saturation value of the average entanglement for the smaller noise rates. 
    This is visible in Fig.~\ref{fig:numerics_random_circuits} (a) for the curves corresponding to $\gamma = 5 \cdot 10^{-4}$ as opposed to those for $\gamma = 2 \cdot 10^{-3}$ or $\gamma = 10^{-2}$ which saturated (for all except the {mixed-unitary} unraveling). 
    For that reason, panels (b-d) should be seen rather as an indication of the behavior. 
    Still, our results indicate that the trend of separations should be conserved.

\subsection{Simulability thresholds}
\label{sec:thresholds}

    In simulations of noisy random 1D circuits based on unraveling into trajectories, the average growth of entanglement with system size can undergo a noise rate-dependent transition from a volume law to an area law.
    Based on the relation between entanglement scaling and representability using MPS, the entanglement phase transition is interpreted as a threshold between non-simulable and \emph{efficiently} simulable noise regimes (on average over random circuits). 
    The exact value of the threshold depends, however, on the choice of unraveling, and is thus used as a criterion for the quality of an unraveling choice. 
    This directly connects to the study of \emph{measurement-induced phase transitions} (MIPT)~\cite{skinnerMeasurementInducedPhaseTransitions2019, zabaloCriticalPropertiesMeasurementinduced2020, baoTheoryPhaseTransition2020, jianMeasurementinducedCriticalityRandom2020a, suzukiQuantumComplexityPhase2023}.
    Indeed, in the setting of noisy random quantum circuits, 
    MIPT can be understood as one 
    particular choice of unraveling (the Projective unraveling) of the dephasing noise channel into projectors onto the computational basis (see the example in Section~\ref{sec:preliminaries}). 
    
    To convey a sense of the scale of noise rate sufficient for efficient simulability using our algorithm, we discuss known noise thresholds
from the trajectory sampling and MIPT literature. Since we expect and numerically observe a fairly robust performance improvement from our unraveling, we expect our unraveling to achieve noise thresholds matching or below the MIPT and {Haar Optimal} unraveling thresholds discussed below.
    
Mixed-Unitary unravelings admit no noise threshold, as randomly inserted single-qubit unitaries do not affect the average entanglement dynamics of random circuits. 
    In contrast, {Projective} unraveling of the dephasing channel was shown to have a threshold around $0.2$~\cite{skinnerMeasurementInducedPhaseTransitions2019, zabaloCriticalPropertiesMeasurementinduced2020, kolodrubetzOptimalityLindbladUnfolding2023}. 
    The value of the thresholds was reduced further using the {Haar Optimal} unraveling $0.09$~\cite{chengEfficientSamplingNoisy2023, chenOptimizedTrajectoryUnraveling2024}. 
    The threshold is expected to be improvable further using multi-qubit unravelings.

    In the random circuit setting, Ref.~\cite{chenOptimizedTrajectoryUnraveling2024} numerically compared their fixed Haar Optimal unraveling to their adaptive numerical optimization unraveling (a heuristic version of our analytical unraveling method). Somewhat surprisingly at the time, they estimated similar thresholds for the two unravelings.
    As a consequence of Section~\ref{sec:random_circuits}, we now understand why the adaptive methods should reduce to the fixed {Haar Optimal} unraveling in the random circuits setting.
    It is, however, to be expected that our strategy would lead to an improvement of thresholds for other systems compared to the fixed Haar Optimal unraveling.
    Because our method provides the exact analytical solution to the local entanglement minimization problem approximated by the heuristic in Ref.~\cite{chenOptimizedTrajectoryUnraveling2024}, we expect the resulting simulability thresholds to be identical in the limit where the heuristic successfully finds the global optimum. Our contribution is thus unlikely to result in a significant shift in the physical threshold itself, but instead furnishes the computation of this optimum deterministically with $O(1)$ overhead, avoiding the convergence uncertainties and computational cost of variational optimization.

\subsection{Comparison to other simulation methods}
\label{sec:comparison_other_methods}

\subsubsection{Tensor network methods}

When trying to generalize tensor network representations for one-dimensional quantum systems from MPS to mixed states, one is given several options.
    The first approach is to use \emph{matrix product operators} (MPOs), in which one adds an open leg to each tensor for the dual space~\cite{pirvuMatrixProductOperator2010, guthjarkovskyEfficientDescriptionManyBody2020, weimerSimulationMethodsOpen2021}. 
    MPOs represent the density matrix as a train of rank-4 tensors. In addition to the potential computational overhead associated with the quadratic increase in dimension, 
    this approach is known to suffer from technical issues limiting the effectiveness of SVD based truncation.
    Indeed, truncations based on the efficiently implementable SVD\footnote{For $1$-D tensor trains, SVDs are efficiently implementable in the dimension of the local tensors.} minimize the $2$-norm error of the post truncation state. For vector representations of states such as MPS, this is the relevant metric. However, for matrix representations of states such as MPOs, the relevant metric is the trace norm.
    Furthermore, truncation of MPOs may result in negative eigenvalues. 
    As checking whether an MPO has negative eigenvalues is computationally hard~\cite{klieschMatrixProductOperatorsStates2014}, one may also obtain non-physical results. 
    MPO representations are also known to suffer from an entanglement barrier when used to simulate open system evolutions~\cite{nohEfficientClassicalSimulation2020, wellnitzRiseFallSlow2022, chengEfficientSamplingNoisy2023}.
    It has been shown that, in certain settings, simulation based on unraveled trajectories can lead to an exponential reduction in simulation cost compared to MPO simulation and thereby avoid the entanglement barrier plaguing MPOs~\cite{vovkQuantumTrajectoryEntanglement2024}.
    Such results are, however, very sensitive to the specific system being simulated, the rate of decoherence, and the initial state, among others, leading to mixed affirmations on the comparison of the efficiency of both methods~\cite{bonnesSuperoperatorsVsTrajectories2014, wolffNumericalEvaluationTwotime2020, preisserComparingBipartiteEntropy2023, vannieuwenburgDynamicsManybodyLocalization2017}. 
    However, to the best of our knowledge, none of these comparisons have been performed with the improved unraveling strategies including the Haar Optimal unraveling from Refs.~\cite{chengEfficientSamplingNoisy2023, chenOptimizedTrajectoryUnraveling2024}.
    
    An alternative representation of mixed states, which is guaranteed to preserve positivity and which can be truncated like an MPS, is obtained by (locally) purifying the density matrix.
    This gives rise to \emph{locally purified density operators} (LPDOs)~\cite{verstraeteMatrixProductDensity2004, zwolakMixedStateDynamicsOneDimensional2004, wernerPositiveTensorNetwork2016, weimerSimulationMethodsOpen2021, chengSimulatingNoisyQuantum2021, guoLocallyPurifiedDensity2023}\footnote{Note that the naming of this tensor network construction is highly inconsistent in the literature. Some authors refer to them as matrix product density operators (which others use for MPOs with trace $1$, independently of positivity), locally purified density operators, or purified matrix product states.}. 
    These, however, result in looped tensor network structures (when tracing out the auxiliary systems) which can not be optimized locally and for which there is no entanglement optimal form~\cite{cuevasPurificationsMultipartiteStates2013, delascuevasFundamentalLimitationsPurifications2016}.
    Additionally, the required bond dimensions may be unboundedly larger than for the MPO counterpart~\cite{glasserExpressivePowerTensornetwork2019} (although, albeit only for exact representations~\cite{delascuevasApproximateTensorDecompositions2021}).
    In this language, trajectory sampling methods can be interpreted as measuring out the auxiliary systems, thus destroying the purifying bond of the LPDO.
    LPDOs are, however, strictly harder to optimize than trajectory methods, as one needs to minimize the local entanglement both within the state and between the state and the environment.

\subsubsection{Pauli propagation methods}

    Another increasingly prominent method for the simulation of noisy quantum circuits is the Pauli propagation method (also known as the Pauli path method
    or sparse Pauli dynamics
    )~\cite{rallSimulationQubitQuantum2019a, huangClassicalSimulationQuantum2022, aharonovPolynomialtimeClassicalAlgorithm2022, begusicFastClassicalSimulation2023, begusicSimulatingQuantumCircuit2023, fontanaClassicalSimulationsNoisy2023, rudolphClassicalSurrogateSimulation2023, ermakovUnifiedFrameworkEfficiently2024, schusterPolynomialtimeClassicalAlgorithm2024, shaoSimulatingNoisyVariational2024, garca2024pauli, angrisaniClassicallyEstimatingObservables2024, cirstoiuFourierAnalysisFramework2024, lerchEfficientQuantumenhancedClassical2024, angrisaniSimulatingQuantumCircuits2025, martinezEfficientSimulationParametrized2025, rudolphPauliPropagationComputational2025}.
    Pauli propagation and trajectory sampling methods are expected to perform best in different settings.
    While MPS-based trajectory methods are tailored to systems with low entanglement, Pauli propagation is rather related to the amount of magic, as they can represent evolutions with (approximately) sparse decompositions in the Pauli basis. Both methods can be used to perform both Observable Estimation and Output Distribution Sampling (see Sec.~\ref{sec:reductions}). However, efficient Pauli path based Output Distribution Sampling algorithms assume anti-concentration of the circuit output distribution~\cite{schusterPolynomialtimeClassicalAlgorithm2024} with some special exceptions \cite{nelson2026classicalsimulationnoisy, Rajakumar2025NoisyIQP}.

    The Pauli path method has proven to be a great tool for proving the existence of efficient classical Observable Estimation algorithms for typical noisy random quantum circuits. A key strength of this method is that it comes with a priori performance guarantees. However, this method is more limited in scope, applying to specific types of noise in specific location and to random circuits. However, progress continues on this front~\cite{garca2024pauli, schusterPolynomialtimeClassicalAlgorithm2024}. For example, Ref.~\cite{schusterPolynomialtimeClassicalAlgorithm2024} relaxed the location restrictions on the depolarizing noise model they consider and concurrently improved the random circuit setting to a random computational basis state initialization. Additionally, in the Observable Estimation setting, Ref.~\cite{garca2024pauli} removed the circuit randomization entirely for 2D Clifford plus T circuits and depolarizing noise.  
        
    Several other subtle differences differentiate the applicability of each method. 
    While Pauli propagation methods consider observables that are at most combinations of polynomially many Pauli strings,
    MPS trajectories allow for the computation of expectation values for any observable that can be represented as a low bond dimension MPO.
    On the other hand, MPS trajectory simulations require initial states that can be represented as MPS states of small bond dimension, whereas Pauli propagation only requires efficient expectation value calculation for single Pauli strings.
    However, most commonly, the classical simulation literature focuses on simulating circuits with simple input states, like the all-zero state.

\section{Limitations}
\label{sec:limitations}

\subsection{Local vs global optimization}

    \begin{figure}
        \centering
        \includegraphics[scale=1]{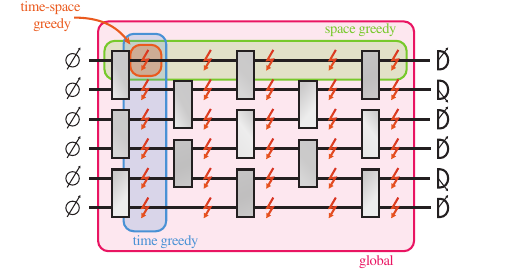}
        \caption{Various possible choices of joint optimization of single-qubit noise channels.}
        \label{fig:greediness}
    \end{figure}
    
    The most desirable setting would be to have the ability to jointly optimize the unitary freedom over all channels, in both time and space.
    Indeed, given two single-qubit channels, optimizing over the unitary freedom of the joint two-qubit channel can lead to better results than optimizing them individually. 
    Such global optimization, however, is not achievable in practice, as it involves, in the most general case, an optimization over a $4^{2nL}$-dimensional unitary.
    Our approach optimizes each channel individually and can be seen as a time and space-greedy optimization method.
    In between, one can consider time- or space-only greedy optimization, but it is unclear if such optimizations are feasible. 
    Take the case of time-only greedy optimization. 
    Finding the entanglement-optimal decomposition of the joint $n$-qubit channel is at least as hard as computing the entanglement of formation of an $n$-qubit mixed state, which is known to be prohibitively hard if one can not leverage some symmetry~\cite{divincenzoOptimalDecompositionsBarely2000, terhalEntanglementFormationIsotropic2000, vollbrechtWhyTwoQubits2000a, vollbrechtEntanglementMeasuresSymmetry2001, gharibianStrongNPhardnessQuantum2010}.
    See Fig.~\ref{fig:greediness} for a graphical representation of the different types of ``greediness'' referred to above.
    
    Our unraveling will, in general, perform worse than joint optimizations would. 
    There is also no guarantee that such a greedy optimization will work at all, as one can construct instances where an ``optimal'' choice in an early step results in disadvantageous initial conditions in a later step.
    We found instances where using a fixed unraveling choice at all steps (in time and space) outperformed our greedy optimization.
    It would be interesting to understand in which settings such greediness performs well and in which it is disadvantageous.  
    For instance, in the setting of time-independent Lindbladian evolutions, one may expect the greedy unraveling to fluctuate less and less as the evolution thermalizes.
    An example where a seemingly irrelevant change in basis significantly affects the performance of the {Haar Optimal} unraveling is presented in Appendix~\ref{sec:greediness_discussion}.

\subsection{Von Neumann entanglement entropy as the choice of cost function}

    It is often stated that MPS offer efficient and good approximations of states with low entanglement, while they fail at faithfully representing highly entangled states.
    This statement, however, is not formally correct. 
    Let us first recall that entanglement can be measured by the more general class of \emph{R\'enyi entropies}
    \begin{equation}
        S_\alpha (\rho) = \frac{1}{1-\alpha} \log \Tr{}{\rho^\alpha} \, ,
        \quad 
        0 \leq \alpha \leq \infty \, .
    \end{equation}
    They generalize the von Neumann entanglement entropy used in Eq.~\eqref{def:entanglement_entropy}, which is recovered in the limit of $\alpha$ tending to $1$.
    It has been shown that classes of states which have entanglement growing logarithmically with system size as measured by some R\'enyi entropy with $\alpha < 1$ can be represented efficiently with MPS~\cite{verstraeteMatrixProductStates2006} while those with R\'enyi $\alpha > 1$ entropies that grow super-logarithmically with system size ($\sim n^\kappa$ even for $\kappa <1$) cannot~\cite{schuchEntropyScalingSimulability2008}. 
    For von Neumann entropies, the non-approximability result holds for linearly scaling entropy, but the remaining guarantees are lost.
    For constant and logarithmically scaling von Neumann entropies, there are even examples of both classes of states, representable states and nonrepresentable states.
        
    Our algorithm is devised to minimize the ensemble-averaged (von Neumann) entanglement entropy in the hope of reducing simulation costs.
    Since it minimizes the von Neumann entanglement entropy, there will be instances where these differences in MPS representability emerge.
    Although one expects the different R\'enyi entropies to have similar behavior in typical settings, there will be examples where optimization of the von Neumann entanglement entropy, some R\'enyi entropy with $\alpha < 1$, or the truncation error will give widely different optimal unravelings.
    In such a setting, optimizing the von Neumann entanglement entropy could lead to poor performance.

\section{Conclusion and discussion}
\label{sec:conclusion}

    In this work, we have aimed to substantially contribute to the body of work on the classical simulation of noisy quantum circuits with trajectory sampling by putting such studies onto a significantly
    more quantitative and rigorous underpinning.
    Concretely, we improve existing methods 
    in two ways. 
    We derive rigorous error bounds on the performance of such algorithms based on MPS, particularly for the relevant tasks of expectation value estimation and measurement outcome distribution sampling. 
    We also propose new ways of computing good decompositions of noise channels. 
    Our first choice of unraveling is provably optimal in a state-dependent, greedy fashion by maximally disentangling the target noisy qubit. 
    At the cost of computational effort similar to that of an SVD of a $4 \times 4$ matrix, we can minimize the average entanglement between the noisy qubit and the rest of the state.
    This unraveling is not tailored to random circuits but is applicable to any circuit. 
    It is also not limited to specific noise channels, as we enable the deterministic treatment of systems where each qubit is affected by a potentially different noise process.

    By readily lending itself to analytic investigation, the same tools allow for the construction of an even simpler, state-independent unraveling for random circuits. In this setting, it provides a new and conceptually significant perspective on what is being optimized. It also reproduces the previous results on fixed unravelings for random circuits and extends these to arbitrary noise models.
    
    Our unraveling strategy has repercussions of a practical and conceptual nature. 
    It translates into improved performance, and perhaps most importantly, an extended scope by permitting locally-optimal unravelings with respect to arbitrary 1D local circuits and single-qubit noise. We believe the latter is an important extension for future progress relating to quantum advantage and the noise-driven decay of quantum computational resources. 
    
    While our results comprehensively answer questions about the optimal (greedy) single-qubit unraveling, many questions remain open, and we raise new ones. 
    The problem of analytically motivated unravelings for multi-qubit noise processes remains untackled. 
    Although in the single-qubit setting, one can focus on disentangling one qubit, which can only be done using non-unitary operators, the multi-qubit case is more nuanced. 
    There, one could minimize the entanglement between the target qubits and the rest of the state or within the target qubits themselves. 
    The first approach bears some similarities with our approach, although there are barely any results on the computation of the entanglement of formation for more than two qubits. 
    The second opens the door for entanglement reduction through unitary Kraus operators. 
    Furthermore, it would be beneficial to lift the limitations of our approach. Developing unravelings based on R\'enyi entropies or even the truncation error would lead to stronger analytical performance guarantees. 
    Some works also focused on the scaling factors of the computational cost and, in particular, on how to reduce the variance in the trajectories. It would be interesting to study our unraveling from this perspective. 

    With the flexibility and generality of our approach, we make progress towards a coherent theory of quantum noise and its impact on quantum advantage. 
    Over the years, in the quest for finding \emph{quantum advantages}, 
    it has been an enormously fruitful endeavor to
run quantum platforms against classical simulations of the same task,
usually by resorting to tensor network models -- actually improving
both research fields at the same time
\cite{Trotzky,IBM_exp2023,Tindall_2024,Zhou_2020,PhysRevResearch.6.013326}. For this reason, having strong predictive power and knowing what the 
precise limitations are for the classical simulation is important, and hence our work 
substantially contributes here.
Our work lends itself to more detailed numerical studies investigating specific classes of circuits and their behavior in the presence of single-qubit noise. Our work provides a posteriori guarantees on simulation accuracy and, since these can in principle take all circuit-level data into account, generally speaking, one might expect such guarantees to be tighter than coarse-grained a priori guarantees. However, a priori performance guarantees are very important for identifying regimes or circuit classes that admit efficient classical simulability. This motivates the development of new techniques that provide a priori performance guarantees on our classical simulation algorithm.

\paragraph*{Code availability:} the code for the simulations is openly accessible under~\cite{unraveling_code}.

    \begin{acknowledgments}
    We would like to thank Frederik vom Ende for useful discussions. Funded by the Deutsche Forschungsgemeinschaft (DFG, German Research Foundation) under Germany´s Excellence Strategy (The Berlin Mathematics Research Center MATH+ EXC-2046/1, project ID: 390685689, the ML4Q
    EXC 2004/1, project ID: 390534769)
    the CRC 183, SPP 2541, the BMFTR (MuniQC-Atoms, QuSol, QPIC-1), the Munich Quantum Valley, Berlin Quantum, the Quantum Flagship (PasQuans2, Millenion), and the European 
    Research Council (DebuQC).
    \end{acknowledgments}

\bibliography{bibliography.bib}

\onecolumngrid
\newpage
\appendix
\begin{center}
\textbf{\large Appendix for\\
``Classical simulation of noisy quantum circuits via locally-entanglement-optimal unravelings''}\\
\vspace{1ex}
\end{center}

\startcontents[app]
\printcontents[app]{l}{1}{}

\section{Unitary freedom in state and channel representations}
\label{sec:unitary_freedom}

    In this section, we recap some key concepts about representations of quantum channels and their implications for the representation of states after applying those channels. 
    First, we will recall the Kraus decomposition and the Stinespring dilation of channels and revise how these implicitly define one possible ensemble decomposition and purification of the output state. 
    We will then focus on the unitary freedom in these different representations and how they relate. 
    Finally, we will dive into how these insights can be used for more efficient classical simulation.

\subsection{Decompositions of states and channels}

\paragraph*{Representations of quantum channels.}
    Any completely positive map $\mathcal{N}$ 
    (hence also any quantum channel we consider) has a Kraus representation, that is a set of $r$ operators $\{K_i\}_{i=1}^{r}$ such that
    \begin{equation}
    \label{eq:def_kraus_2}
        \mathcal{N} : \rho \mapsto \sum_{i=1}^{r} K_i \rho K_i^\dagger.
    \end{equation}
    This map is trace-preserving if and only if $\sum_{i=1}^{r} K_i^\dagger K_i = \II$.
    Equivalently, any channel can be described by its Stinespring dilation. 
    Calling the Hilbert space in which the system evolves $\mathcal{S}$ for \emph{system} and defining an auxiliary Hilbert space $\mathcal{E}$ for \emph{environment}, the Stinespring dilation is given by an isometry $V_{\mathcal{S}\mathcal{E}|\mathcal{S}} : \mathcal{H}_{\mathcal{S}} \rightarrow \mathcal{H}_{\mathcal{SE}}$ from the system $\mathcal{S}$ to the joint system and environment $\mathcal{S}\mathcal{E}$ such that
    \begin{equation}
    \label{eq:def_stinespring}
        \mathcal{N} [\rho_{\mathcal{S}}] = \Tr{\mathcal{E}}{ V_{\mathcal{S}\mathcal{E}|\mathcal{S}}^{\vphantom{\dagger}} \rho_{\mathcal{S}} V_{\mathcal{S}\mathcal{E}|\mathcal{S}}^\dagger },
    \end{equation}
    where the state $\rho$ is only defined on the system Hilbert space.
    Such an isometry can be constructed from the Kraus 
    operators by choosing 
    \begin{equation}
        \label{eq:kraus_to_stinespring}
        V_{\mathcal{S}\mathcal{E}|\mathcal{S}} = \sum_{i=1}^{r} (K_i)_{\mathcal{S}} \otimes \ket{i}_{\mathcal{E}},
    \end{equation}
    which, when inserted into Eq.~\eqref{eq:def_stinespring}, 
    will yield the Kraus decompositions of Eq.~\eqref{eq:def_kraus_2}. 
    For completeness, we note that the Stinespring dilation can be extended to a unitary system-environment evolution by finding an adequate unitary $U_{\mathcal{S}\mathcal{E}}$ and environment state $\eta_{\mathcal{E}}$ such that 
    \begin{equation}
        \mathcal{N} [\rho_\mathcal{S}] = \Tr{\mathcal{E}}{ U_{\mathcal{S}\mathcal{E}} (\rho_{\mathcal{S}} \otimes \eta_{\mathcal{E}}) U_{\mathcal{S}\mathcal{E}}^\dagger }.
    \end{equation}

\paragraph*{Quantum state descriptions.}
    One way of describing a mixed quantum state is through an ensemble decomposition. 
    That is, a set of (pure) states $\{\ket{\phi_i}\}_{i=1}^{r}$ and a corresponding set of probabilities $\{p_i\}_{i=1}^{r}$ such that the mixed state can be interpreted as a probabilistic (convex) mixture of those states
    \begin{equation}
        \rho = \sum_{i=1}^{r} p_i \ketbra{\phi_i}{\phi_i}.
    \end{equation}
    Observe that the eigendecomposition is one such possible ensemble decomposition. 
    Another ensemble decomposition 
    of high interest to us is the one induced directly by the Kraus representation of a channel. 
    Say we apply a quantum channel $\mathcal{N}$ to some state vector $\ket{\psi}$. 
    Then, the outcome is given by 
    \begin{equation}
        \mathcal{N}[\ketbra{\psi}{\psi}] 
        = \sum_{i=1}^{r} K_i \ketbra{\psi}{\psi} K_i^\dagger
        = \sum_{i=1}^{r} p_i \ketbra{\phi_i}{\phi_i}
    \end{equation}
    with $p_i = \| K_i \ket{\psi} \|_2^2 = \bra{\psi} K_i^\dagger K_i \ket{\psi}$ and $\ket{\phi_i} = p_i^{-1/2}K_i \ket{\psi}$.
    Those are the so-called ``pure-state trajectories'' when unraveling the noisy evolution. 
    We will see that the state and the channel decompositions are not unique.

\subsection{Unitary freedom}

    We start with the unitary relation between ensemble decompositions of states and from there derive unitary relations between channels.
    
\paragraph*{Unitary freedom for states.}
    Ensemble decompositions are not unique: there can be many different sets of states and corresponding probability distributions such that their probabilistic mixture result in the same state. 
    This is particularly obvious for states that have degenerate eigenvalues. 
    The most extreme example is the maximally mixed state. 
    Take, for instance, the two-qubit case. 
    It can be decomposed into any orthogonal basis, for instance, the computational and the Bell basis
    \begin{align}
        \frac{\II}{4} &
        = \frac{1}{4} \left( \ketbra{0,0}{0,0} + \ketbra{0,1}{0,1} + \ketbra{1,0}{1,0} + \ketbra{1,1}{1,1} \right) \\ 
        \nonumber&
        = \frac{1}{4} \left( \ketbra{\Phi^+}{\Phi^+} + \ketbra{\Phi^-}{\Phi^-} + \ketbra{\Psi^+}{\Psi^+} + \ketbra{\Psi^-}{\Psi^-} \right) \, .
    \end{align}
    Note, and it will be crucial later, that all states are product states in the first decomposition, while they are maximally entangled in the other.
    
    These ensemble decompositions are all unitarily related. 
    Let $\{p_k, \ket{\phi_k}\}_{k=1}^{r}$ and $\{q_{k'}, \ket{\varphi_{k'}}\}_{k'=1}^{r'}$ be two decompositions of the same state $\rho$.
    Then, there exists an $\ell \times \ell$ dimensional unitary matrix $U$, with $\ell = \max (r, r')$, such that 
    \begin{equation}
    \label{eq:unitarity_ensembles}
        \sqrt{q_{k'}} \ket{\varphi_{k'}} = \sum_{k=1}^{r} U_{kk'} \sqrt{p_{k}} \ket{\phi_{k}} \quad \forall k' \in \{ 1, \dots, r' \} \, .
    \end{equation}
    This result can be obtained through steering, using projective measurements on equivalent, unitarily related purifications of the state. 
    More on this can be found in  Renes' book~\cite[Section 6.2.2]{renesQuantumInformationTheory2022a}.
    An alternative formulation in matrix-algebraic terms is given in Ref.~\cite{audenaertVariationalCharacterizationsSeparability2001}, although first proven by Huhgston, Josza and Wootters in Ref.~\cite{hughstonCompleteClassificationQuantum1993}
    \footnote{
        The order of the indices of the elements of the unitary matrix might be swapped depending on the source. 
        This is of minor importance as one could pick the transpose of the unitary as the new unitary freedom, resulting in the alternative formulation. 
        In Eq.~\eqref{eq:unitarity_kraus} we have used the equivalent formulation but used $U^\top$ instead of $U$, to align with the notation in Ref.~\cite{woottersEntanglementFormationArbitrary1998} and because it simplifies some notation.
    }.
    In graphical notation, the same result can be viewed as adding a unitary and its conjugate in the tensor network, which disappear upon contraction
    \begin{align}
        \rho & 
        = \sum_{k=1}^{r} p_k \ketbra{\phi_k}{\phi_k}
        = \left( \sum_{k=1}^{r} \ketbra{\phi_k}{k} \right) 
          \left( \sum_{k''=1}^{r} p_k \ketbra{k''}{k''} \right)
          \left( \sum_{k'=1}^{r} \ketbra{k'}{\phi_{k'}} \right) && 
        = \tikzineq{
            \Vertex[x=-.5, size=.3, shape=rectangle, color=pink2s, ]{phi}
            \Vertex[x=.5, size=.3, shape=rectangle, color=pink2s, label=$\ket{\phi}$, position=90]{phi'}
            \Vertex[x=0, size=.1, shape=rectangle, color=yellows, label=$p$, position=-135, 
                    style={rotate=45}]{lambda}
            \Edge (phi)(phi')
            \Edge (phi')(.9,0)
            \Edge (phi)(-.9,0)
        } \\
        \nonumber & 
        = \left( \sum_{k=1}^{r} \sqrt{p_k} \ketbra{\phi_k}{k} \right) 
          \II
          \left( \sum_{k'=1}^{r} \sqrt{p_{k'}} \ketbra{k'}{\phi_{k'}} \right) && 
        = \tikzineq{
            \Vertex[x=-1, size=.3, shape=rectangle, color=pink2s, ]{phi}
            \Vertex[x=1, size=.3, shape=rectangle, color=pink2s, label=$\ket{\phi}$, position=90]{phi'}
            \Vertex[x=-.5, size=.1, shape=rectangle, color=yellows, label=$\sqrt{p}$, position=-135, 
                    style={rotate=45}]{p}
            \Vertex[x=-0, size=.05, shape=circle, color=black, label=$\II$, position=90]{I}
            \Vertex[x=.5, size=.1, shape=rectangle, color=yellows, position=-135, 
                    style={rotate=45}]{p'}
            \Edge (phi)(phi')
            \Edge (phi')(-1.4,0)
            \Edge (phi)(1.4,0)
        } \\ 
        \nonumber
        &
        = \left( \sum_{k=1}^{r} \sqrt{p_k} \ketbra{\phi_k}{k} \right) 
          \left( \sum_{i,j=1}^{r'} U_{i,j} \ketbra{i}{j} \right)
          \left( \sum_{i',j'=1}^{r'} U^*_{i'j'} \ketbra{j'}{i'} \right)
          \left( \sum_{k'=1}^{r} \sqrt{p_{k'}} \ketbra{k'}{\phi_{k'}} \right) && 
        = \tikzineq{
            \Vertex[x=-1, size=.3, shape=rectangle, color=pink2s]{phi}
            \Vertex[x=-.5, size=.1, shape=rectangle, color=yellows, 
                    style={rotate=45}]{lambda1}
            \Vertex[x=0, size=.3, shape=rectangle, color=blue1s, label=$U^\dagger$, position=90]{U}
            \Vertex[x=.5, size=.3, shape=rectangle, color=blue1s, label=$U$, position=90]{U'}
            \Vertex[x=1, size=.1, shape=rectangle, color=yellows, 
                    style={rotate=45}, label=$\sqrt{p}$, position=-135]{lambda2}
            \Vertex[x=1.5, size=.3, shape=rectangle, color=pink2s, label=$\ket{\phi}$, position=90]{phi}
            \Edge (phi)(phi')
            \Edge (phi)(-1.4,0)
            \Edge (phi')(1.9,0)
        }  \\
        \nonumber&
        = \left( \sum_{k=1}^{r} \sum_{i,j=1}^{r'} U_{i,j} \sqrt{p_k} \inner{k}{i} \ketbra{\phi_k}{j} \right) 
          \left( \sum_{k'=1}^{r} \sum_{i',j'=1}^{r'} U^*_{i',j'} \sqrt{p_{k'}} \inner{i'}{k'} \ketbra{j'}{\phi_{k'}} \right) &&\\ &
        = \sum_{j=1}^{r'} \left( \sum_{k=1}^{r} U_{k,j} \sqrt{p_k} \ket{\phi_k} \right) \left( \sum_{k'=1}^{r} U_{k',j}^* \sqrt{p_{k'}} \bra{\phi_{k'}} \right) && \\
        \nonumber& 
        = \sum_{j=1}^{r'} \left( \sqrt{q_j} \ket{\varphi_j} \right) \left( \sqrt{q_j} \bra{\varphi_j} \right) && 
        = \tikzineq{
            \Vertex[x=-.5, size=.3, shape=rectangle, color=purple2s]{phi'}
            \Vertex[x=0, size=.1, shape=rectangle, color=green2s, 
                    style={rotate=45}]{lambda1}
            \Vertex[x=.5, size=.1, shape=rectangle, color=green2s, 
                    style={rotate=45}, label=$\sqrt{q}$, position=-135]{lambda2}
            \Vertex[x=1, size=.3, shape=rectangle, color=purple2s, label=$\ket{\varphi}$, position=90]{phi}
            \Edge (phi)(phi')
            \Edge (phi)(1.4,0)
            \Edge (phi')(-.9,0)
        }\;.
        \nonumber
    \end{align}
    Above, we assumed $r' > r$.
    This can be done without loss of generality, as one could prove it the other way around in the other case.
    Note that when increasing the number of elements in the decomposition (i.e., $r' \geq r$), any unitary $U$ of dimension $r' \times r'$ will result in a valid decomposition of the state and can thus be chosen without concerns.
    We have to be more careful when decreasing the number of elements.
    To avoid this, we can start from the decomposition with the minimal number of elements, i.e., the eigendecomposition 
    \begin{equation}
        \rho = \sum_{i=1}^{r} \lambda_i \ketbra{v_i}{v_i} \, .
    \end{equation}
    The rank of the state gives the minimal number of elements. 
    For any choice of $U$
    \begin{equation}
        \sqrt{p_{k}} \ket{\phi_{k}} = \sum_{i=1}^{} U_{ik} \sqrt{\lambda_i} \ket{v_i} 
        \quad \forall k \in \{ 1, \dots, r' \}, 
        \quad r' \geq r
    \end{equation}
    is a valid decomposition of $\rho$.

\paragraph*{Unitary freedom for channels.}
    Kraus representations of quantum channels can be obtained from the ensemble decompositions of their Choi operator~\footnote{
        For a channel $\mathcal{N}$ from a $d$-dimensional Hilbert space $\mathcal{H}$ to itself, 
        its Choi state is defined as ${\choi}(\mathcal{N}) = (\mathcal{N} \otimes \mathcal{I})[\ketbra{\omega}{\omega}]$, 
        where the maximally entangled state on $\mathcal{H} \otimes \mathcal{H}$ is $ \ket{\omega} = \frac{1}{\sqrt{d}} \sum_i \ket{i} \otimes \ket{i}$. 
        Its eigendecomposition is ${\choi} = \sum_j \lambda^{\choi}_j \ketbra{\lambda_j}{\lambda_j}$ (although any decomposition would be a valid choice), 
        and a Kraus decomposition can be obtained by finding a set of operators $\{ K_j \}_j$ such that $\sqrt{\lambda^{\choi}_j} \ket{\lambda_j} = (K_i \otimes \II) \ket{\omega}$.
    }~\cite[Section 6.4]{renesQuantumInformationTheory2022a},~\cite[Section 2.2]{watrousTheoryQuantumInformation2018}. 
    As a result, one finds that Kraus representations are not unique either, 
    and all Kraus representations must be unitarily related. 
    This means that, for any two Kraus representations $\{ K_i \}_{i=1}^{r}$ and $\{ K^{\circlearrowright}_{i} \}_{i=1}^{r'}$, there has to be a unitary $U$ such that 
    \begin{equation}
        K^{\circlearrowright}_{i} = \sum_{j=1}^{r} U_{i,j} K_{j} \quad \forall i \in \{ 1, \dots, r' \}.
    \end{equation}
    For a detailed discussion on Choi operators, ensemble decompositions, and unitary relations, we again recommend Renes' Book~\cite{renesQuantumInformationTheory2022a}, and in particular Sections 5.2, 5.3, and 5.4 together with Sections 6.2, 6.3, and 6.4.
    \begin{proof}[Proof sketch]
        We start by representing channels in the Stinespring picture, as in Eq.~\eqref{eq:def_stinespring},
        where we can confirm that applying a local unitary on the environment does not change the channel. 
        Defining the alternative isometry 
        $V^{\circlearrowright}_{\mathcal{S E}|\mathcal{S}} = (\II_\mathcal{S} \otimes U_\mathcal{E}) V_{\mathcal{S E}|\mathcal{S}}$ and using the circularity of the trace we obtain
        \begin{equation}
            \Tr{{\mathcal{E}}}{ V^{\circlearrowright}_{\mathcal{S}\mathcal{E}|\mathcal{S}} \rho_\mathcal{S} V^{\circlearrowright\dagger}_{\mathcal{S}\mathcal{E}|\mathcal{S}} }
            = \Tr{{\mathcal{E}}}{ ({\II}_{\mathcal{S}} \otimes U_\mathcal{E}) V_{\mathcal{S}\mathcal{E}|\mathcal{S}} \rho_\mathcal{S} V_{\mathcal{S}\mathcal{E}|\mathcal{S}}^\dagger ({\II}_{\mathcal{S}} \otimes U_\mathcal{E})^\dagger }
            = \Tr{{\mathcal{E}}}{ V_{\mathcal{S}\mathcal{E}|\mathcal{S}} \rho_\mathcal{S} V_{\mathcal{S}\mathcal{E}|\mathcal{S}}^\dagger },
        \end{equation}
        which is the original description of the channel.
        Using Eq.~\eqref{eq:kraus_to_stinespring} again, we get that 
        \begin{align}
            \Tr{\mathcal{E}}{ (\II_{\mathcal{S}} \otimes U_{\mathcal{E}}) V_{\mathcal{SE}|\mathcal{S}} \rho_{\mathcal{S}} V_{\mathcal{SE}|\mathcal{S}}^\dagger (\II_{\mathcal{S}} \otimes U_{\mathcal{E}})^\dagger } & 
            = \Tr{\mathcal{E}}{ (\II_{\mathcal{S}} \otimes U_{\mathcal{E}}) \left(\sum_{i} (K_i)_{\mathcal{S}} \otimes \ket{i}_{\mathcal{E}}\right) \rho_{\mathcal{S}} \left(\sum_{i'} (K_{i'})_{\mathcal{S}} \otimes \ket{i'}_{\mathcal{E}}\right)^\dagger (\II_{\mathcal{S}} \otimes U_{\mathcal{E}})^\dagger } \\
            \nonumber 
            & 
            = \Tr{\mathcal{E}}{ \sum_{i, i'} \left( (K_i)_{\mathcal{S}} \otimes U_{\mathcal{E}} \ket{i}_{\mathcal{E}}\right) \rho_{\mathcal{S}} \left( (K_{i'}^\dagger)_{\mathcal{S}} \otimes \bra{i'}_{\mathcal{E}} U_{\mathcal{E}}^\dagger \right) } \\ & 
            = \sum_{i, i', i''} \left( \bra{i''} U \ket{i}_{\mathcal{E}} K_i \right) \rho \left( K_{i'}^\dagger \bra{i'} U^\dagger \ket{i''}_{\mathcal{E}} \right) \\ & 
            = \sum_{i''} K^{\circlearrowright}_{i''} \rho K^{\circlearrowright\dagger}_{i''}
            \nonumber
        \end{align}
        with 
        $K^{\circlearrowright}_{i''} =  \sum_{i} \bra{i''} U \ket{i} K_i$.
    \end{proof}

\paragraph*{Correspondance between unitary freedoms of states and channels.}
    Both unitary freedoms presented above are intimately related.
    Let us formalize the relation between the unitary that transforms Kraus operators and the one that transforms the resulting ensemble decomposition.
    Starting from Eq.~\eqref{eq:Kraus_to_ensemble} and using Eq.~\eqref{eq:unitarity_kraus}, we derive the effect of the unitary freedom of the Kraus operators on the ensemble decomposition as
    \begin{align}
        \mathcal{N}[\ketbra{\psi}{\psi}] & 
        = \sum_{i'=1}^{r'} K^{(U)}_{i'}\ketbra{\psi}{\psi} K^{(U)\dagger}_{i'} \\ 
        \nonumber & 
        = \sum_{i'=1}^{r'} \left(\sum_{i=1}^{r} U_{i'i} K_{i} \right) \ketbra{\psi}{\psi} \left(\sum_{i''=1}^{r} U_{i'i''} K_{i''} \right)^\dagger \\
        \nonumber& 
        = \sum_{i'=1}^{r'} \left(\sum_{i=1}^{r} U_{i'i} K_{i} \ket{\psi} \right) \left(\sum_{i''=1}^{r} \bra{\psi} K_{i''}^\dagger (U^\dagger)_{i''i'} \right) \\
        \nonumber& 
        = \sum_{i'=1}^{r'} \left(\sum_{i=1}^{r} U_{i'i} \sqrt{p_i} \ket{\phi_i} \right) \left(\sum_{i''=1}^{r} \sqrt{p_i} \bra{\phi_i} (U^\dagger)_{i''i'} \right) \\ & 
        = \sum_{i'=1}^{r'} q_{i'} \ketbra{\varphi_{i'}}{\varphi_{i'}}
        \nonumber
        .
    \end{align}
    The first equality uses the alternative definition of the channel.
    The second equality directly applies Eq.~\eqref{eq:unitarity_kraus}. 
    The fourth uses the naming choice from Eq.~\eqref{eq:Kraus_to_ensemble}.
    The last equality is obtained by using Eq.~\eqref{eq:unitarity_ensembles}.
    The different ensembles are thus related by the same unitary relating the different Kraus representations. 

\subsection{Example: the Werner state}

    We already saw the example of the maximally mixed state above. 
    However, this example is unique due to eigenvalue degeneracies. 
    We will now consider a slightly more complicated example, revealing many interesting aspects that come into play in our algorithm.
    The Werner state~\cite{wernerQuantumStatesEinsteinPodolskyRosen1989, bennettMixedstateEntanglementQuantum1996} on two qubits can be written (among several different but equivalent descriptions) as the mixture of a maximally entangled state and a maximally mixed state
    \begin{equation}
    \label{eq:def-Werner-state}
        W(\lambda) = \lambda \ketbra{\Psi^-}{\Psi^-} + (1-\lambda) \frac{\II}{4} .
    \end{equation}
    We want to use this example to discuss the ensemble-averaged entanglement entropy, the unitary freedom, and the entanglement of formation. 

    Let us first decompose the identity part in the arguably most natural way, i.e., in the computational basis
    \begin{equation}
        \frac{\II}{4} = \frac{1}{4} \left( \ketbra{0,0}{0,0} + \ketbra{0,1}{0,1} + \ketbra{1,0}{1,0} + \ketbra{1,1}{1,1} \right) \, .
    \end{equation}
    This induces a decomposition of the Werner state as 
    $\left\{ \sqrt{\lambda} \ket{\Psi^-}, \frac{1}{2}\sqrt{1-\lambda} \ket{0,0}, \frac{1}{2}\sqrt{1-\lambda} \ket{0,1}, \frac{1}{2}\sqrt{1-\lambda} \ket{1,0}, \frac{1}{2}\sqrt{1-\lambda} \ket{1,1} \right\}$,
    which we will refer to as the ``computational'' decomposition.
    Each element of the decomposition of the identity does not contribute to the average entanglement entropy (Eq.~\eqref{def:ensemble-averaged_entanglement}) since $E(\ket{xy})=0$ for all $x,y \in \{ 0,1 \}$.
    Then,
    \begin{equation}
        E_{\mathrm{av}} ^{(\text{comp})} (W(\lambda)) 
        = \lambda E(\ket{\Psi^-}) + \frac{1-\lambda}{4}(0+0+0+0) 
        = \lambda \, .
    \end{equation}
    With this decomposition, the ensemble-averaged entanglement entropy scales linearly in (is even equal to) the parameter $\lambda$ controlling the probability of the Bell state in the decomposition, going from the minimum of zero to the maximum of one.
    However, we will see that it does not describe the entanglement faithfully.
    Taking the decomposition of the maximally mixed state in the Bell basis
    \begin{equation}
        \frac{\II}{4} = \frac{1}{4} \left( \ketbra{\Phi^+}{\Phi^+} + \ketbra{\Phi^-}{\Phi^-} + \ketbra{\Psi^+}{\Psi^+} + \ketbra{\Psi^-}{\Psi^-}\right) \, , 
    \end{equation}
    results in the ``Bell'' decomposition
    $\left\{ \frac{1}{2}\sqrt{3\lambda+1} \ket{\Psi^-}, \frac{1}{2}\sqrt{1-\lambda} \ket{\Phi^+}, \frac{1}{2}\sqrt{1-\lambda} \ket{\Phi^-}, \frac{1}{2}\sqrt{1-\lambda} \ket{\Psi^+} \right\}$
    that has an average entanglement of
    \begin{equation}
        E_{\mathrm{av}}^{(\text{Bell})} (W(\lambda)) 
        = \frac{(3\lambda+1)}{4} E(\ket{\Psi^-}) + \frac{(1-\lambda)}{4} \left( E(\ket{\Phi^+}) + E(\ket{\Phi^-}) + E(\ket{\Psi^+}) \right)
        = \frac{(3\lambda+1)}{4} \cdot 1 + \frac{(1-\lambda)}{4} \left( 1 + 1 + 1 \right)
        = 1 .
    \end{equation}
    In this decomposition, the average entropy is independent of the parameter $\lambda$ and maximal.
    
    However, it is known that two-qubit Werner states are separable for $\lambda \leq 1/3$. 
    Thus, there has to exist an ensemble decomposition for which $E_{\mathrm{av}}^{(\text{opt})} = 0$ if $\lambda \leq 1/3$.
    An example is the decomposition $\left\{ \frac{1}{2} \ket{\phi_1}, \frac{1}{2} \ket{\phi_2}, \frac{1}{2} \ket{\phi_3}, \frac{1}{2} \ket{\phi_4} \right\}$ where the $\ket{\phi_i}$ are  defined by
    \begin{align}
         \ket{\phi_1} & = \frac{1}{2}\sqrt{3\lambda+1} \ket{\Psi^-} + \frac{i}{2}\sqrt{1-\lambda} \ket{\Phi^+} + \frac{1}{2}\sqrt{1-\lambda} \ket{\Phi^-} + \frac{1}{2}\sqrt{1-\lambda} \ket{\Psi^+} ,\\
         \ket{\phi_2} & = \frac{1}{2}\sqrt{3\lambda+1} \ket{\Psi^-} + \frac{i}{2}\sqrt{1-\lambda} \ket{\Phi^+} - \frac{1}{2}\sqrt{1-\lambda} \ket{\Phi^-} - \frac{1}{2}\sqrt{1-\lambda} \ket{\Psi^+}, \\
        \ket{\phi_3} & = \frac{1}{2}\sqrt{3\lambda+1} \ket{\Psi^-} - \frac{i}{2}\sqrt{1-\lambda} \ket{\Phi^+} + \frac{1}{2}\sqrt{1-\lambda} \ket{\Phi^-} - \frac{1}{2}\sqrt{1-\lambda} \ket{\Psi^+} ,\\
        \ket{\phi_4} & = \frac{1}{2}\sqrt{3\lambda+1} \ket{\Psi^-} - \frac{i}{2}\sqrt{1-\lambda} \ket{\Phi^+} - \frac{1}{2}\sqrt{1-\lambda} \ket{\Phi^-} + \frac{1}{2}\sqrt{1-\lambda} \ket{\Psi^+}\;,
     \end{align}
    which has an ensemble-averaged entanglement entropy equal to the entanglement of formation for $\lambda \geq 1/3$, achieving zero at $\lambda = 1/3$.
    This unraveling corresponds to applying Eq.~\eqref{eq:unitarity_ensembles}, that is, mixing the four elements of the decomposition in the Bell basis with the parameters from the unitary 
    \begin{equation}
        U = \frac{1}{2} \begin{pmatrix}
            1 & i & 1 & 1 \\
            1 & i & -1 & -1 \\
            1 & -i & 1 & -1 \\
            1 & -i & -1 & 1 
        \end{pmatrix}
        = \underbrace{
            \frac{1}{2} \begin{pmatrix}
                1 & 1 & 1 & 1 \\
                1 & 1 & -1 & -1 \\
                1 & -1 & 1 & -1 \\
                1 & -1 & -1 & 1 
            \end{pmatrix}}_{
            4 \times 4 \text{ Hadamard}
        }\underbrace{
            \begin{pmatrix}
                1 & 0 & 0 & 0 \\
                0 & i & 0 & 0 \\
                0 & 0 & 1 & 0 \\
                0 & 0 & 0 & 1 
            \end{pmatrix}}_{
            \text{phase on } \ket{\Phi^+}
        } \, .
    \end{equation}
    This specific and some alternative unravelings with their respective ensemble-averaged entanglement entropies are shown in Fig.~\ref{fig:Werner_unraveling}.
    More on optimal decompositions for the separable Werner states can be found in works by Wootters~\cite{woottersEntanglementFormationArbitrary1998} or Azuma and Ban~\cite{azumaAnotherConvexCombination2006}.
    
    \begin{figure}
        \centering
        \includegraphics{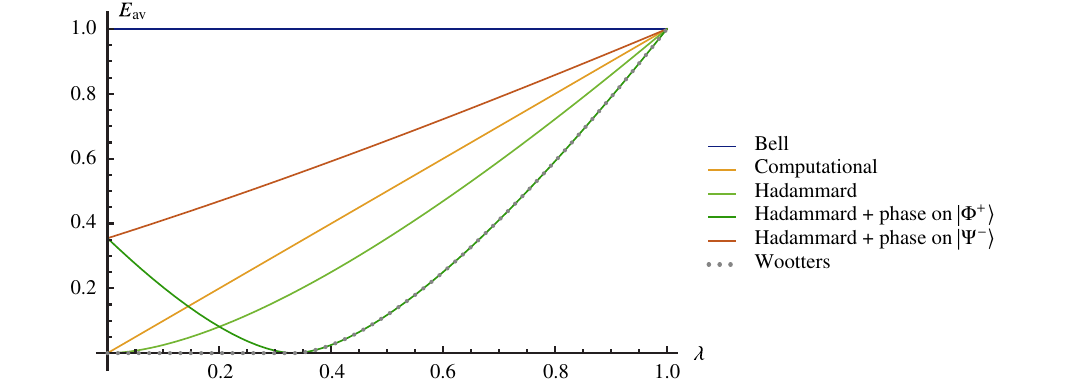}
        \caption{Ensemble-averaged entanglement entropy for some unravelings of the Werner state.
        }
        \label{fig:Werner_unraveling}
    \end{figure}

\subsection{Unitary degree of freedom and the choice of unraveling}
\label{sec:unitary_unraveling_choice}

    We have seen above that the freedom in decomposing states can result in widely varying values of the ensemble-averaged entanglement entropy. 
    We want to use this observation to discuss how to choose the Kraus representation of quantum channels to minimize the average entanglement of the induced decomposition.

\paragraph*{The Werner state as a depolarized Bell state.}
    The Werner state can be seen as the application of the (single qubit) depolarizing channel with noise parameter $p=1-\lambda$
    \begin{equation}
        \mathcal{N}_{DP}[\rho] = (1-p) \rho + p \frac{\II}{2}
    \end{equation}
    on the first qubit of the Bell state vector 
    $\ket{\Psi^-} = \frac{1}{\sqrt{2}}(\ket{0,1} - \ket{1,0})$, i.e.,  
    \begin{equation}
         (\mathcal{N}_{DP} \otimes \mathcal{I})[\ketbra{\Psi^-}{\Psi^-}]
         = (1-p) \ketbra{\Psi^-}{\Psi^-} + p \frac{\II}{2} \otimes \Tr{1}{\ketbra{\Psi^-}{\Psi^-}}
         = (1-p) \ketbra{\Psi^-}{\Psi^-} + p \frac{\II}{2} \otimes \frac{\II}{2}.
    \end{equation}
    Due to simple habit, we will now substitute $\ket{\Psi^-}$ with 
    $\ket{\Phi^+} = \frac{1}{\sqrt{2}}(\ket{0,0} + \ket{1,1})$.
    This has no consequence in the following.

    Now let us find a (unitarily equivalent) Kraus decomposition of the depolarizing channel. 
    From the definition above, we can convince ourselves that the channel can be written as
    \begin{equation}
        \mathcal{N}_{DP}[\rho] = (1-p) \II \rho \II + \frac{p}{2} (\ketbra{0}{0} \rho \ketbra{0}{0} + \ketbra{0}{1} \rho \ketbra{1}{0} + \ketbra{1}{0} \rho \ketbra{0}{1} + \ketbra{1}{1} \rho \ketbra{1}{1}).
    \end{equation}
    Consequently, one valid choice of Kraus operators, that we call ``{Projective}'', is given by 
    \begin{equation}
        \left\{ K_i^{\mathrm{proj}}\right\}_{i=1}^5 = \left\{ \sqrt{1-p} \II, \sqrt{\frac{p}{2}} \ketbra{0}{0}, \sqrt{\frac{p}{2}} \ketbra{0}{1}, \sqrt{\frac{p}{2}} \ketbra{1}{0}, \sqrt{\frac{p}{2}} \ketbra{1}{1} \right\}.
    \end{equation}
    A different, well-known decomposition using the Pauli matrices is
    \begin{equation}
        \mathcal{N}_{DP}[\rho] = \left( 1-\frac{3p}{4} \right) \II \rho \II + \frac{p}{4} (X \rho X + Y \rho Y + Z \rho Z)
    \end{equation}
    using that for all states $\rho$ it holds that
    $\II \rho \II + X \rho X + Y \rho Y + Z \rho Z = \II$.
    Thus, our second valid set of Kraus operators, which we call ``unitary'', is given by
    \begin{equation}
        \left\{ K_i^{\mathrm{unit}}\right\}_{i=1}^4 = \left\{ \sqrt{1-\frac{3p}{4}} \II, \frac{\sqrt{p}}{2} X, \frac{\sqrt{p}}{2} Y, \frac{\sqrt{p}}{2} Z \right\}.
    \end{equation}
    Although we do not prove it here, the rank of the Choi isomorphism of the depolarizing channel is four, so the number of Kraus operators in any decomposition has to be at least four. 
    Thus, the unitary decomposition is minimal. 

    Applying the {Projective} decomposition of the depolarizing channel to the Bell state results in the computational decomposition of the Werner state, whereas the unitary decomposition induces the Bell decomposition. 
    As discussed previously, one can also mix the Kraus operators by choosing any unitary matrix. 
    We show a few examples of the effect of the unitary freedom on the ensemble-averaged entanglement entropy in Fig.~\ref{fig:Bell-depolarizing}, where, to keep it simple, we consider the $2$-parameter, $4\times 4$ unitary matrix
    \begin{equation}
    \label{eq:U_phi_theta}
        U(\theta, \phi) = \begin{pmatrix}
            \cos \theta \cos \phi & \cos \theta \sin \phi & \sin \theta \cos \phi & \sin \theta \sin \phi \\
            -\cos \theta \sin \phi & \cos \theta \cos \phi & -\sin \theta \sin \phi & \sin \theta \cos \phi \\
            -\sin \theta \cos \phi & -\sin \theta \sin \phi & \cos \theta \cos \phi & \cos \theta \sin \phi \\
            \sin \theta \sin \phi & -\sin \theta \cos \phi & -\cos \theta \sin \phi & \cos \theta \cos \phi
        \end{pmatrix} \, .
    \end{equation}
    
    Similarly, decompositions of the noise channel exist that induce ensembles of states that minimize the average entanglement. 
    Since simulating states with a lot of entanglement using MPS is expensive, reducing the ensemble-averaged entanglement entropy can lead to lower simulation costs.
    
    \begin{figure}
        \centering
        \includegraphics{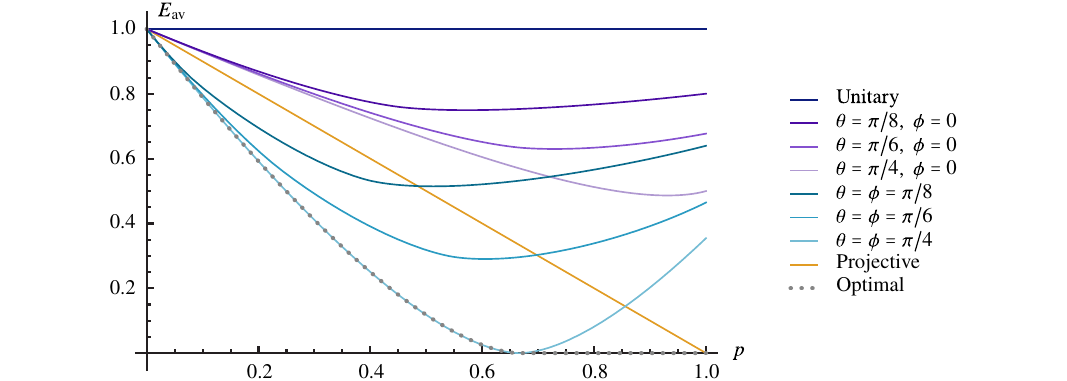}
        \caption{Average entanglement of a depolarized Bell state for different choices of unraveling.
        The unravelings labeled by $\theta$ and $\phi$ refer to rotations of the unitary decomposition with the unitary from Eq.~\ref{eq:U_phi_theta}.
        }
        \label{fig:Bell-depolarizing}
    \end{figure}

\paragraph*{Using the unitary freedom for improved simulation performance.}

    The differences in average entanglement entropy could potentially be used in our algorithm. 
    Different decompositions of a channel result in different Schmidt spectra and thus different computational costs and errors. 
    First of all, some Kraus representations have more elements than others. 
    That leads to a decision tree for sampling pure states with more leaves, but it does not change the general functioning of the algorithm as long as one can efficiently sample from them.
    It might affect the sampling overhead, but not the asymptotic scaling of the simulation.
    Furthermore, as seen above, the entanglement properties of the resulting pure state ensemble decomposition of the mixed state differ depending on the choice of Kraus decomposition. 
    Since entanglement is an expensive resource for simulation with matrix product state methods and leads to errors in truncation, choosing a Kraus decomposition that leads to low-entangled states can improve the algorithm's performance, all without changing the simulated physical process.
    
    As stated above, the error committed by the algorithm is related to the average entanglement entropy of each pure state in the ensemble decomposition. 
    Our goal is then to minimize this quantity, that is, finding
    \begin{equation}
        \inf_{\{p_i, \ket{\phi_i}\}} \sum_{i} p_i E(\ket{\phi_i})
        \quad \text{or equivalently} \quad 
        \inf_{\{K_i\}} \sum_{i} \bra{\psi} K^\dagger_i K_i \ket{\psi} E \left( \frac{K_i \ket{\psi}}{\sqrt{\bra{\psi} K^\dagger_i K_i \ket{\psi}}} \right)
    \end{equation}
    with $E(\ket{\psi})$ being the entanglement entropy of the state $\ket{\psi}$ over some bipartition (Eq.~\eqref{def:entanglement_entropy}).
    This quantity is already known in the literature as the \emph{entanglement of formation}~\cite{bennettConcentratingPartialEntanglement1996, bennettMixedstateEntanglementQuantum1996, horodeckiQuantumEntanglement2009}.
    One possible choice of partition is to separate the first $q$ qubits (where qubit $q$ could be the one affected by the noise channel) from the remaining $n-q$ qubits
    \begin{equation}
        E(\ket{\psi}_{1\dots q : q+1 \dots n}) = -\Tr{}{\rho_{1\dots q} \log \rho_{1\dots q}} 
        \, .
    \end{equation}
    This choice of bipartition would naturally translate to properties of the bond in an MPS between qubits $q$ and $q+1$.
    However, finding the entanglement of formation of some arbitrary mixed state is a hard problem in the worst case~\cite{gharibianStrongNPhardnessQuantum2010, zhuImprovedLowerUpper2012, kimEntanglementFormationMonogamy2021}, so one might have to resign to finding good approximations or reduce the problem to simpler instances which have known solutions.
    To do so, we will consider the partition between the affected target qubit and the rest of the state. 
    This reduces the problem to the computation of entropies of single-qubit states and represents the minimization of the entanglement between the target qubit and all other qubits in the state.

\subsection{Some common unravelings}
\label{sec:common_unravelings}

    In this section, we will describe a few unravelings that have been used before for the depolarizing, dephasing, and amplitude-damping noise channels. 
    
\subsubsection*{Dephasing}

    We start with the \emph{mixed-unitary} (or {orthogonal}) representation of the dephasing channel
    \begin{equation}
        \mathcal{N}[\rho] = \left(1-\frac{p}{2}\right) \rho + \frac{p}{2}Z \rho Z 
        \quad \Rightarrow \quad
        \{K_i\}_{i=1}^{2} = \left\{ \sqrt{1-\frac{p}{2}} \II, \sqrt{\frac{p}{2}} Z \right\} \, .
    \end{equation}
    In this representation, each Kraus operator is unitary, up to a rescaling related to the respective probabilities.
    Next, we have the \emph{Projective} decomposition into the identity (up to constant factor) and rank-$1$ projectors on the computational basis
    \begin{equation}
        \mathcal{N}[\rho] = \left(1-p\right) \rho + p \ketbra{0}{0} \rho \ketbra{0}{0} + p \ketbra{1}{1} \rho \ketbra{1}{1}
        \quad \Rightarrow \quad
        \{K_i\}_{i=1}^{3} = \left\{ \sqrt{1-p} \II, \sqrt{p} \ketbra{0}{0}, \sqrt{p} \ketbra{1}{1} \right\} \, .
    \end{equation}
    Finally, we also consider the \emph{Haar Optimal} decomposition, 
    obtained from applying the Hadamard gate as unitary freedom on the {mixed-unitary} decomposition as proposed in Refs.~\cite{chengEfficientSamplingNoisy2023, chenOptimizedTrajectoryUnraveling2024}
    \begin{align}
        \{K_i\}_{i=1}^{2} & 
        = \left\{ \sqrt{\frac{1-p/2}{2}} \II + \frac{\sqrt{p}}{2} Z, \sqrt{\frac{1-p/2}{2}} \II - \frac{\sqrt{p}}{2} Z \right\} \\ 
         \nonumber&
        = \left\{ 
            \frac{1}{\sqrt{2}} \begin{pmatrix}
                \sqrt{1-p/2} + \sqrt{p/2} & 0 \\ 0 & \sqrt{1-p/2} + \sqrt{p/2} 
            \end{pmatrix}, 
            \frac{1}{\sqrt{2}} \begin{pmatrix}
                \sqrt{1-p/2} + \sqrt{p/2} & \\ 0 & \sqrt{1-p/2} - \sqrt{p/2}
            \end{pmatrix}
        \right\}
        \, .
    \end{align}

\subsubsection*{Depolarizing}

    Just like dephasing, we also introduce a {mixed-unitary} ({Orthogonal}) unraveling, what we label {Projective} unraveling and a {Haar Optimal} unraveling of the depolarizing noise channel.
    The {mixed-unitary} unraveling is a decomposition into the four rescaled Pauli operators
    \begin{equation}
        \mathcal{N}[\rho] = \left(1-\frac{3p}{4}\right) \rho + \frac{p}{4} (X \rho X + Y \rho Y + Z \rho Z) 
        \quad \Rightarrow \quad
        \{K_i\}_{i=1}^{4} = \left\{ \sqrt{1-\frac{3p}{4}} \II, \frac{\sqrt{p}}{2} X, \frac{\sqrt{p}}{2} Y, \frac{\sqrt{p}}{2} Z \right\} \, .
    \end{equation} 
    Our so-called {Projective} unraveling is again formed of the rescaled identity operator and some rank-$1$ operators
    \begin{multline}
        \mathcal{N}[\rho] = \left(1-p\right) \rho + \frac{p}{2} (\ketbra{0}{0} \rho \ketbra{0}{0} + \ketbra{0}{1} \rho \ketbra{1}{0} + \ketbra{1}{0} \rho \ketbra{0}{1} + \ketbra{1}{1} \rho \ketbra{1}{1}) \\
        \Rightarrow \quad
        \{K_i\}_{i=1}^{5} = \left\{ \sqrt{1-p} \II, \sqrt{\frac{p}{2}} \ketbra{0}{0}, \sqrt{\frac{p}{2}} \ketbra{0}{1}, \sqrt{\frac{p}{2}} \ketbra{1}{0}, \sqrt{\frac{p}{2}} \ketbra{1}{1} \right\} \, .
    \end{multline} 
    Although it is not properly a decomposition into projectors, it is reminiscent of the {Projective} unraveling of the dephasing channel as it acts trivially with some probability, or measures the state into the computational basis, but this time followed by a random preparation of one of the two computational basis states.
    The {Haar Optimal} unraveling is given by the unitary freedom (Eq.~\eqref{eq:unitarity_kraus}) on the {mixed-unitary} unraveling with
    \begin{equation}
        U = \frac{1}{2} \begin{pmatrix}
            1 & 1 & 1 & 1 \\
            1 & 1 & -1 & -1 \\
            1 & -1 & 1 & -1 \\
            1 & -1 & -1 & 1 
        \end{pmatrix} 
        \quad \Rightarrow \quad
        \{K_i\}_{i=1}^{4} = \begin{Bmatrix}
            \frac{1}{2} \sqrt{1-\frac{3p}{4}} \II + \frac{\sqrt{p}}{4}(X + Y + Z), \\
            \frac{1}{2} \sqrt{1-\frac{3p}{4}} \II + \frac{\sqrt{p}}{4}(X - Y - Z), \\
            \frac{1}{2} \sqrt{1-\frac{3p}{4}} \II + \frac{\sqrt{p}}{4}(-X + Y - Z), \\
            \frac{1}{2} \sqrt{1-\frac{3p}{4}} \II + \frac{\sqrt{p}}{4}(-X - Y + Z)
        \end{Bmatrix}
        \, .
    \end{equation}

\subsubsection*{Amplitude damping}

    The amplitude damping channel does not admit a mixed-unitary decomposition and is usually described in terms of what we call the \emph{Orthogonal} decomposition
    \begin{equation}
        \{K_i\}_{i=1}^{2} = \left\{
            \begin{pmatrix}
                1 & 0 \\ 0 & \sqrt{1-\gamma}
            \end{pmatrix}, 
            \begin{pmatrix}
                0 & \sqrt{\gamma} \\ 0 & 0
            \end{pmatrix}
        \right\} \, .
    \end{equation} 
    Like for the two channels above, we also consider the {Haar Optimal} version thereof, given by applying the Hadamard matrix
    \begin{equation}
        H = \frac{1}{\sqrt{2}} \begin{pmatrix} 1 & 1 \\ 1 & -1 \end{pmatrix} 
        \quad \Rightarrow \quad
        \{K_i\}_{i=1}^{2} = \left\{ \sqrt{\frac{p}{2}} \begin{pmatrix} 1 & \sqrt{\gamma} \\ 0 & \sqrt{1-\gamma} \end{pmatrix}, \sqrt{\frac{p}{2}} \begin{pmatrix} 1 & -\sqrt{\gamma} \\ 0 & \sqrt{1-\gamma} \end{pmatrix} \right\}
    \end{equation} 
     as unitary freedom.

\section{Proofs of the error bounds}
\label{sec:upper_bound_proofs}

    This section shows how to compute a posteriori error bounds for trajectory sampling techniques. 
    In particular, after establishing some notation (Section~\ref{sec:notation}), we prove Theorem~\ref{result:sampled_trajectories} (Section~\ref{proof:sampled_trajectories}) and show how the bound can be computed in practice when using MPS as an underlying data structure (Section~\ref{sec:MPS-based_errors}).
    Note that the simulation error depends on the choice of decomposition of the density matrices, which is not unique.
    Nonetheless, the analysis we will do now holds for any choice of unraveling.
    The a posteriori error bound, computed along the way, will then differ depending on the unraveling choice.

\subsection{Clarifying some notation}
\label{sec:notation}

    Let us first introduce some notation choices we will use to derive results on the performance of the quantum trajectory unraveling algorithm. 
    Those are also summarized in Table~\ref{tab:notation_memo}.
    We consider an initial state vector $\ket{\psi^{(0)}}$, which can be any state efficiently representable as an MPS (for instance the all-zero state $\ket{\psi^{(0)}} = \ket{0}^{\otimes n}$). 
    The corresponding density matrix is $\rho^{(0)} = \ketbra{\psi^{(0)}}{\psi^{(0)}}$. 
    We also consider a sequence of $L$ quantum channels $\{\Lambda^{(\ell)}\}_{\ell=1}^L$ and call $\rho^{(L)}$ the state after the full process
    \begin{equation}
        \label{eq:exact_noisy_state}
        \rho^{(L)} = \left( \bigcirc_{\ell=1}^L \Lambda^{(\ell)} \right) \left[ \ketbra{\psi^{(0)}}{\psi^{(0)}} \right] \, .
    \end{equation}
    This state is decomposed into pure states through the trajectory unraveling procedure.
    Let us see how this decomposition is constructed. 
    Applying the first channel to the initial state, we obtain some ensemble decomposition of the mixed state induced by the choice of Kraus decomposition, as shown in Eq.~\eqref{eq:Kraus_to_ensemble}
    \begin{equation}
        \Lambda^{(1)} \left[ \ketbra{\psi^{(0)}}{\psi^{(0)}} \right]
        = \sum_s K^{(1)}_s \ketbra{\psi^{(0)}}{\psi^{(0)}} K^{(1) \dagger}_s
        = \sum_s p({s}) \ketbra{{\psi}_{{s}}}{{\psi}_{{s}}} 
        \quad \text{for} \quad
        \Lambda^{(1)} \left[ \cdot \right]
        = \sum_s K^{(1)}_s (\cdot) K^{(1)\dagger}_s
        \, .
    \end{equation}
    Here, the label ${s}$ runs over all the Kraus operators of $\Lambda^{(1)}$.
    Then comes the application of the next channel, $\Lambda^{(2)}$.
    For each element of the decomposition of the state after the first layer, we are free to choose a different Kraus decomposition of the second layer
    $\Lambda^{(2)} \left[ \cdot \right]
    = \sum_{s'} K^{(2)}_{s'} (\cdot) K^{(2)\dagger}_{s'}$
    where the decomposition $\{K^{(2)}_{s'}\}_{s'}$ is a function of the input state vector.
    Still, we can enumerate the Kraus operators and use an integer label $s'$
    \begin{equation}
        \Lambda^{(2)} \circ \Lambda^{(1)} \left[ \ketbra{\psi^{(0)}}{\psi^{(0)}} \right]
        = \sum_{{s}} p({s}) \sum_{{s'}}  p({s}'|{s})\ketbra{{\psi}_{{s, s'}}}{{\psi}_{{s, s'}}}
        = \sum_{{s, s'}} p(s,s') \ketbra{{\psi}_{{s, s'}}}{{\psi}_{{s, s'}}} \, .
    \end{equation}
    Repeating at every layer, a tree of possibilities is generated where each branch is one possible trajectory from the process, and each leaf is one of the possible sampled final states. 
    From now on, we will use arrays of integers denoted $\boldsymbol{s}^{(L)}$ to label the branches. 
    Each label $\boldsymbol{s}^{(L)}$ is composed of $L$ integers, representing which Kraus operator has been sampled at each layer. 
    That way, the labels $\boldsymbol{s}^{(L)}$ store the whole history of sampling of one trajectory $\ket{{\psi}_{\boldsymbol{s}^{(L)}}}$.
    Then, for any $\ell < L$, the label of the intermediate state $\ket{{\psi}_{\boldsymbol{s}^{(\ell)}}}$ is simply the array of the first $\ell$ elements of $\boldsymbol{s}^{(L)}$.
    Then we have that, after $L$ layers, the state can be written as
    \begin{equation}
        \label{eq:exact_noisy_state_decomposition}
        \rho^{(L)}
        = \sum_{\boldsymbol{s}^{(L)}} p(\boldsymbol{s}^{(L)}) \ketbra{{\psi}_{\boldsymbol{s}^{(L)}}}{{\psi}_{\boldsymbol{s}^{(L)}}} \, . 
    \end{equation}
    A small example of such a tree of trajectories is shown in Fig.~\ref{fig:unraveling_tree}.

    \begin{figure}
        \centering
        \includegraphics[width=\linewidth]{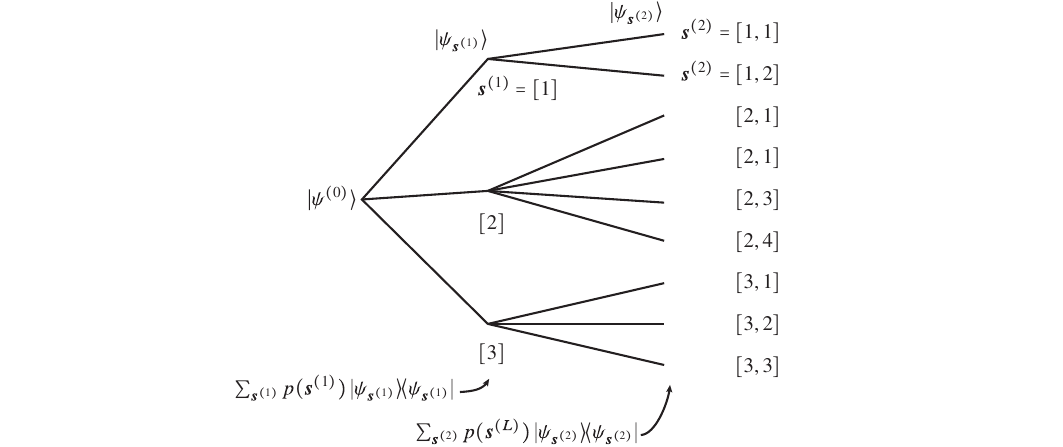}
        \caption{Representation of the tree of possible trajectories in a quantum trajectory unraveling procedure.}
        \label{fig:unraveling_tree}
    \end{figure}

    \begin{table}[]
        \centering
        \begin{tabular}{lll}
            Exact state & 
            $\rho^{(L)} 
            = \left( \bigcirc_{\ell=1}^L \Lambda^{(\ell)} \right) \left[ \ketbra{\psi^{(0)}}{\psi^{(0)}} \right] $ & 
            $= \sum_{\boldsymbol{s}} {p}_{\boldsymbol{s}} \ketbra{{\psi}_{\boldsymbol{s}}}{{\psi}_{\boldsymbol{s}}}$ \\
            Approximate representation & 
            $\approxrep{\rho}^{(L)} 
            = \left( \bigcirc_{\ell=1}^L \operatorname{Approx} \circ \Lambda^{(\ell)} \right) \left[ \ketbra{\psi^{(0)}}{\psi^{(0)}} \right]$ & 
            $= \sum_{\boldsymbol{s}} \approxrep{p}_{\boldsymbol{s}} \ketbra{\approxrep{\psi}_{\boldsymbol{s}}}{\approxrep{\psi}_{\boldsymbol{s}}}$ \\
            Before the last approximation & 
            $\preapproxrep{\rho}^{(L)} 
            = \Lambda^{(L)} [\approxrep{\rho}^{(L-1)}]$ & 
            $= \sum_{\boldsymbol{s}} \approxrep{p}_{\boldsymbol{s}} \ketbra{\preapproxrep{\psi}_{\boldsymbol{s}}}{\preapproxrep{\psi}_{\boldsymbol{s}}}$ \\
            Estimator state & 
            \multicolumn{2}{l}{
            $\hat{\rho} 
            = \frac{1}{N} \sum_{i=1}^N  \ketbra{\hat{\psi}_{i}}{\hat{\psi}_{i}}$
            \quad with  
            $ \ket{\hat{\psi}_{i}}$ sampled from $\approxrep{\rho}^{(L)}$
            }\\
        \end{tabular}
        \caption{Overview of notations and definitions used for the different states appearing in the proofs of the approximation errors.}
        \label{tab:notation_memo}
    \end{table}

    This would be the representation if we could always store each element of the decomposition exactly.
    However, to ensure computational efficiency, the state at each step has to be approximated by another state that can be treated efficiently. 
    In the case of MPS, this corresponds to truncating to a fixed bond dimension. 
    We will use $\operatorname{Approx}[\rho]$ to denote this procedure of approximating the state, leaving it generic for now.
    Any approximation procedure we consider is an approximation that acts on the state vector level, so for any state $\rho$ and any decomposition of it $\rho = \sum_i p_i \ketbra{\psi_i}{\psi_i}$
    the approximation procedure approximates each element individually as
    \begin{equation}
    \label{eq:approx_on_decomposition}
        \operatorname{Approx}[\rho] = \sum_i p_i \ketbra{\operatorname{Approx}[\psi_i]}{\operatorname{Approx}[\psi_i]} \, .
    \end{equation}
    To refer to the approximate representation of $\rho^{(L)}$, where approximations are performed at each layer, we will use
    \begin{equation}
        \label{eq:approx_noisy_state}
        \approxrep{\rho}^{(L)} = \left( \bigcirc_{\ell=1}^L \operatorname{Approx} \circ \Lambda^{(\ell)} \right) \left[ \ketbra{\psi^{(0)}}{\psi^{(0)}} \right] \, .
    \end{equation}
    Analogously to the exact case, Eq.~\eqref{eq:exact_noisy_state_decomposition}, each channel is decomposed into a Kraus representation (depending on the input state vector), resulting in a pure-state decomposition of the output state, where each element is approximated. 
    At the same time, the probability distributions differ from the exact case since, after performing a truncation at one layer, the states at the next layer and their corresponding probabilities are affected.
    Then, at any depth, there is a decomposition of the approximated state, obtained from the repeated application of Kraus decompositions and approximation
    \begin{equation}
        \label{eq:approx_noisy_state_decomposition}
        \approxrep{\rho}^{(\ell)} 
        = \sum_{\boldsymbol{s}^{(\ell)}} \approxrep{p}(\boldsymbol{s}^{(\ell)})   \ketbra{\approxrep{\psi}_{\boldsymbol{s}^{(\ell)}}}{\approxrep{\psi}_{\boldsymbol{s}^{(\ell)}}} \, .
    \end{equation}
    To make things explicit, let us also define the state resulting from applying the channel $\Lambda^{(\ell)}$ to the approximate state from the previous layer $\approxrep{\rho}^{(\ell-1)}$, but before performing the approximation, 
    \begin{equation}
    \label{eq:pre-aprox_noisy_state}
        \preapproxrep{\rho}^{(\ell)} 
        = \Lambda^{(\ell)} \left[ \approxrep{\rho}^{(\ell-1)} \right] 
        = \sum_{\boldsymbol{s}^{(\ell)}} \approxrep{p}(\boldsymbol{s}^{(\ell)}) \ketbra{\preapproxrep{\psi}_{\boldsymbol{s}^{(\ell)}}}{\preapproxrep{\psi}_{\boldsymbol{s}^{(\ell)}}} \, .
    \end{equation}
    Then we obtain Eq.~\eqref{eq:approx_noisy_state_decomposition} through the approximation of each state vector $\ket{\approxrep{\psi}_{\boldsymbol{s}^{(\ell)}}} = \ket{\operatorname{Approx}[\preapproxrep{\psi}_{\boldsymbol{s}^{(\ell)}}]}$ and $\approxrep{\rho}^{(\ell)} = \operatorname{Approx} [\preapproxrep{\rho}^{(\ell)}]$.
    Let us call $\varepsilon(\approxrep{\psi}_{\boldsymbol{s}^{(\ell)}})$ the error in approximating the state $\ket{\preapproxrep{\psi}_{\boldsymbol{s}^{(\ell)}}}$
    \begin{equation}
    \label{eq:single_approx_error}
        \varepsilon(\approxrep{\psi}_{\boldsymbol{s}^{(\ell)}})
        = \left\| \ketbra{\preapproxrep{\psi}_{\boldsymbol{s}^{(\ell)}}}{\preapproxrep{\psi}_{\boldsymbol{s}^{(\ell)}}} - \ketbra{\approxrep{\psi}_{\boldsymbol{s}^{(\ell)}}}{\approxrep{\psi}_{\boldsymbol{s}^{(\ell)}}} \right\|_{\mathrm{Tr}} \, . 
    \end{equation}
    If the approximation is given by the truncation of the MPS, then the error above can be bounded using Eq.~\eqref{eq:MPS_truncation_error}, as we will see in Section~\ref{sec:MPS-based_errors}.
    We will also use $\varepsilon_{\mathrm{tot}}(\approxrep{\psi}_{\boldsymbol{s}^{(L)}})$ for the accumulation of the errors up to layer $L$
    \begin{equation}
    \label{eq:cumulated_trajectory_error}
        \varepsilon_{\mathrm{tot}}(\approxrep{\psi}_{\boldsymbol{s}^{(L)}})
        = \sum_{\ell = 1}^L 
        \varepsilon(\approxrep{\psi}_{\boldsymbol{s}^{(\ell)}}) \, . 
    \end{equation}
    This quantity does not translate directly to a trace distance error between some exact and approximate states, but it will appear in our analysis.
    These notation choices can be found summarized in Table~\ref{tab:notation_memo}.

    With this, we will make a few observations on the tree of trajectories, the probability distributions, and the accumulation of errors.
    First, due to Bayes' rule, the probabilities of the leaves can easily be written in terms of the probabilities at each branching
    \begin{equation}
        \approxrep{p} (\boldsymbol{s}^{(L)})
        = \approxrep{p} (\boldsymbol{s}^{(L)} | \boldsymbol{s}^{(L-1)}) \approxrep{p} (\boldsymbol{s}^{(L-1)})
        = \dots 
        = \left( \prod_{\ell = 1}^L \approxrep{p} (\boldsymbol{s}^{(\ell)} | \boldsymbol{s}^{(\ell-1)}) \right) \approxrep{p} (\boldsymbol{s}^{(0)})
    \end{equation}
    where $\boldsymbol{s}^{(0)}$ is the empty array and $\approxrep{p} (\boldsymbol{s}^{(0)}) = 1$. 
    Based on that, at any step and for any input pure state, the distribution of outputs forms a valid probability distribution.
    That means that at each node of the tree, conditioned on being at that node, the sum of probabilities at any depth to the right has to sum up to one, i.e.,
    \begin{equation}
    \label{eq:conditional_tree_distributions}
         \sum_{\boldsymbol{s}^{(\ell')}} \approxrep{p} (\boldsymbol{s}^{(\ell')} | \boldsymbol{s}^{(\ell)}) = 1
         \quad \forall \ell < \ell' \, .
    \end{equation}

\subsection{Proof of Theorem~\ref{result:sampled_trajectories}}
\label{proof:sampled_trajectories}

\trajectoryerror*

    Before proceeding with the proof, let us reconcile the (simplified) notation from the main text with the notation presented in the section above. 
    To reduce notational overhead, in the statement of Theorem~\ref{result:sampled_trajectories} we have used $\rho$ for $\rho^{(L)}$ and $\sigma$ for $\approxrep{\rho}^{(L)}$. 
    Correspondingly, the label $i$ stands for the label of each final trajectory $\boldsymbol{s}^{(L)}$ with respective probabilities $ q_i =  \approxrep{p}_{\boldsymbol{s}^{(L)}} =  \approxrep{p}(\boldsymbol{s}^{(L)})$ and states $\ket{\psi_i} = \ket{\approxrep{\psi}_{\boldsymbol{s}^{(L)}}}$.
    Throughout this section, we will use the nomenclature presented in Section~\ref{sec:notation}, in particular the notation for the exact state, the approximate state and the state resulting from sampling trajectories.

\begin{proof}
    To prove the statement of Theorem~\ref{result:sampled_trajectories}, we will first prove that the approximation error from the unraveling procedure is bounded by the sum of error bounds from each branch of the tree of possible trajectories. 
    Each of them can be bounded by the sum of errors of each layer in one trajectory (Proposition~\ref{result:systematic_error})
    or by a constant (Proposition~\ref{result:path_error_bound}). 
    Then, in Proposition~\ref{result:systematic_error_bound}, we show that although the quantity from Proposition~\ref{result:systematic_error} is not known, we compute an estimator that is probably approximately correct.

    \begin{proposition}
        \label{result:systematic_error}
        Given
        \begin{itemize}
            \item An initial state vector ($\ket{\psi^{(0)}}$) as an MPS (eg. $\ket{\psi^{(0)}} = \ket{0}^{\otimes n}$),
            \item A sequence of quantum channels $\{\Lambda^{(\ell)}\}_{\ell=1}^L$,
            \item An approximation procedure $\operatorname{Approx}[\cdot]$ which acts as in Eq.~\eqref{eq:approx_on_decomposition},
        \end{itemize}
        then the error in the approximate representation of the evolution $\approxrep{\rho}^{(L)}$ as defined in Eq.~\eqref{eq:approx_noisy_state_decomposition} can be bounded
        \begin{equation}
            \| \rho^{(L)} - \approxrep{\rho}^{(L)} \|_{\mathrm{Tr}} = \varepsilon_{\mathrm{sys}}
            \quad \text{with} \quad 
            \rho^{(L)} \coloneqq \Lambda^{(L)} \circ \cdots \circ \Lambda^{(1)} [\rho^{(0)}]
        \end{equation}
        where 
        \begin{equation}
        \label{eq:systematic_error}
            \varepsilon_{\mathrm{sys}} \leq \sum_{\boldsymbol{s}^{(L)}} \approxrep{p}_{\boldsymbol{s}^{(L)}} \varepsilon_{\mathrm{tot}}(\approxrep{\psi}_{\boldsymbol{s}^{(L)}}) 
            \quad \text{with} \quad 
             \varepsilon_{\mathrm{tot}} (\cdot)
             \text{ as defined in Eq.~\eqref{eq:cumulated_trajectory_error}} \, .
        \end{equation}
    \end{proposition}

    \begin{proof}
    We will show this in three steps. 
    First, we show that for such a layered process, the total approximation error is bounded by the sum of approximation errors in each layer (Proposition~\ref{result:error_per_layer}). 
    Second, we show that for each layer, the approximation error is bounded by the sum of approximation errors of each element of the obtained pure state decomposition (Proposition~\ref{result:error_per_pure_state}).
    Third, we show that the resulting sum over layers, of the sum over elements is equivalent to a sum over all final trajectories, of the sum at each layer (Proposition~\ref{result:error_sums_equivalence}).
    One could also study the error committed individually in the samples given by each full trajectory (all layers) and average them according to the probability distribution.
    However, the interplay of the truncation errors and the resulting error in the sampling probabilities makes the analysis cumbersome. 
    
    We begin by showing that given some exactly representable state vector $\ket{\psi^{(0)}}$ (e.g.,
    the state $\ketbra{0^{\otimes n}}{0^{\otimes n}}$) and an evolution given by $L$ layers of quantum channels $ \left\{ \Lambda^{(\ell)} \right\}_{\ell=1}^L$ (e.g. $L$ layers of unitary operations composed with noise channels), the total approximation error is upper-bounded by the sum of approximation errors at each layer.

    \begin{proposition}[Cumulated approximation error]
    \label{result:error_per_layer}
        Let $\rho^{(L)}$ the state after the process (Eq.~\eqref{eq:exact_noisy_state})
        \begin{equation}
            \rho^{(L)} = \left( \bigcirc_{\ell=1}^L \Lambda^{(\ell)} \right) \left[ \ketbra{\psi^{(0)}}{\psi^{(0)}} \right] \, .
        \end{equation}
        Let $\operatorname{Approx}[\cdot]$ denote the procedure that returns some approximate input representation.
        Consequently, the approximate representation of $\rho^{(L)}$, where approximations are performed at each layer, is (Eq.~\eqref{eq:approx_noisy_state})
        \begin{equation}
            \approxrep{\rho}^{(L)} = \left( \bigcirc_{\ell=1}^L \operatorname{Approx} \circ \Lambda^{(\ell)} \right) \left[ \ketbra{\psi^{(0)}}{\psi^{(0)}} \right] \, .
        \end{equation}
        Then, the cumulated approximation error after $L$ layers is upper-bounded by the sum of approximation errors at each layer 
        \begin{equation}
        \label{eq:error_per_layer}
            \left\| \rho^{(L)} - \approxrep{\rho}^{(L)} \right\|_{\mathrm{Tr}} 
            \leq \sum_{\ell = 1}^L \Biggl\| \underbrace{\Lambda^{(\ell)} \left[ \approxrep{\rho}^{(\ell-1)} \right]}_{\preapproxrep{\rho}^{(\ell)}} 
            - \underbrace{\operatorname{Approx} \circ \Lambda^{(\ell)} \left[ \approxrep{\rho}^{(\ell-1)} \right]}_{\approxrep{\rho}^{(\ell)}} \Biggr\|_{\mathrm{Tr}}
            \, .
        \end{equation}
    \end{proposition}

    \begin{proof}
        Let $\operatorname{Approx}[\rho]$ denote some approximate representation of $\rho$, inducing some error. 
        Starting from some exactly representable state, the noisy circuit is composed of $L$ layers of channels $ \left\{ \Lambda^{(\ell)} \right\}_{\ell=1}^L$.
        The error committed at layer $\ell + 1$ fulfills
        \begin{align}
            \left\| \rho^{(\ell+1)} - \approxrep{\rho}^{(\ell+1)} \right\|_{\mathrm{Tr}} &
            = \left\| \Lambda^{(\ell+1)} \left[ \rho^{(\ell)} \right] - \operatorname{Approx} \circ \Lambda^{(\ell+1)} \left[ \approxrep{\rho}^{(\ell)} \right] \right\|_{\mathrm{Tr}} \nonumber \\ 
            \nonumber & 
            = \Big\| \Lambda^{(\ell+1)} \left[\rho^{(\ell)}\right] \underbrace{- \Lambda^{(\ell+1)} \left[ \approxrep{\rho}^{(\ell)} \right] + \Lambda^{(\ell+1)} \left[ \approxrep{\rho}^{(\ell)} \right]}_{=0} - \operatorname{Approx} \circ \Lambda^{(\ell+1)} \left[ \approxrep{\rho}^{(\ell)} \right] \Big\|_{\mathrm{Tr}} \nonumber \\ & 
            \leq \underbrace{ \left\| \Lambda^{(\ell+1)} \left[\rho^{(\ell)}\right] - \Lambda^{(\ell+1)} \left[ \approxrep{\rho}^{(\ell)} \right] \right\|_{\mathrm{Tr}}}_{ \leq \| \rho^{(\ell)} - \approxrep{\rho}^{(\ell)} \|_{\mathrm{Tr}}} 
            + \underbrace{ \left\| \Lambda^{(\ell+1)} \left[ \approxrep{\rho}^{(\ell)} \right] - \operatorname{Approx} \circ \Lambda^{(\ell+1)} \left[ \approxrep{\rho}^{(\ell)} \right] \right\|_{\mathrm{Tr}}}_{\mathrm{approximation error}}.
            \nonumber
        \end{align}
        The definition of a layer in our model gives the first equality. 
        The second equality results from adding and subtracting the same element, thus adding zero.
        The third line is obtained using the triangle inequality and the final inequality on the left summand is a consequence of the fact that quantum channels reduce the distinguishability of two states.
        Thus, by induction, the error committed on the expected state after the $L$ layers is upper-bounded by the sum of errors committed by each truncation step,
        \begin{equation}
            \left\| \rho^{(L)} - \approxrep{\rho}^{(L)} \right\|_{\mathrm{Tr}} 
            \leq \sum_{\ell = 1}^L \left\| \Lambda^{(\ell)} \left[ \approxrep{\rho}^{(\ell-1)} \right] - \operatorname{Approx} \circ \Lambda^{(\ell)} \left[ \approxrep{\rho}^{(\ell-1)} \right] \right\|_{\mathrm{Tr}}
        \end{equation}
        with $\approxrep{\rho}^{(0)} = \rho^{(0)} = \ketbra{\psi^{(0)}}{\psi^{(0)}}$.
    \end{proof}

    Next, we can focus on the approximation error at each step.
    To this end, we use the notation defined in Section~\ref{sec:notation}.
    The exact noisy state $\rho^{(L)}$ can be decomposed into (subnormalized) state vectors $\{\sqrt{{p}_{\boldsymbol{s}}} \ket{\psi_{\boldsymbol{s}}} \}$ (Eq.~\eqref{eq:exact_noisy_state_decomposition}).
    The approximate representation of it, $\approxrep{\rho}^{(L)}$, is decomposed into pure states as $\{\sqrt{\approxrep{p}_{\boldsymbol{s}}} \ket{\approxrep{\psi}_{\boldsymbol{s}}} \}$ (Eq.~\eqref{eq:approx_noisy_state_decomposition}).
    Finally, the state resulting from applying the channel $\Lambda^{(L)}$ to the approximate state from the previous layer, but before performing the approximation, $\preapproxrep{\rho}^{(L)}$ and its decomposition $\{\sqrt{\approxrep{p}_{\boldsymbol{s}}} \ket{\preapproxrep{\psi}_{\boldsymbol{s}}} \}$ as in Eq.~\eqref{eq:pre-aprox_noisy_state}.

    \begin{proposition}
    \label{result:error_per_pure_state}
        Given a layered quantum process described by $L$ quantum channels
        $ \left \{ \Lambda^{(\ell)} \right \}_{\ell = 1}^L$, 
        resulting in density matrices decomposed into pure states
        \begin{equation}
            \preapproxrep{\rho}^{(\ell)} 
            = \Lambda^{(\ell)} \left[ \approxrep{\rho}^{(\ell-1)} \right] 
            = \sum_{\boldsymbol{s}^{(\ell)}} \approxrep{p}_{\boldsymbol{s}^{(\ell)}} \ketbra{\preapproxrep{\psi}_{\boldsymbol{s}^{(\ell)}}}{\preapproxrep{\psi}_{\boldsymbol{s}^{(\ell)}}}
        \end{equation}
        and some approximation procedure on the state vectors
        \begin{equation}
            \ket{\approxrep{\psi}_{\boldsymbol{s}}} = \ket{\operatorname{Approx}[\preapproxrep{\psi}_{\boldsymbol{s}}]}
        \end{equation}
        acting on mixed states as in Eq.~\eqref{eq:approx_on_decomposition}, then the approximation error of the exactly evolved state $\rho^{(L)}$ and the one resulting from approximating at each layer $\approxrep{\rho}^{(L)}$ is bounded by the sum of approximation errors of each element of the ensemble at each layer
        \begin{equation}
        \label{eq:simulation_error_per_sample}
            \left\| \rho^{(L)} - \approxrep{\rho}^{(L)} \right\|_{\mathrm{Tr}} 
            \leq \sum_{\ell = 1}^L \sum_{\boldsymbol{s}^{(\ell)}} \approxrep{p}_{\boldsymbol{s}^{(\ell)}} \underbrace{ \left\| \ketbra{\preapproxrep{\psi}_{\boldsymbol{s}^{(\ell)}}}{\preapproxrep{\psi}_{\boldsymbol{s}^{(\ell)}}} - \ketbra{\approxrep{\psi}_{\boldsymbol{s}^{(\ell)}}}{\approxrep{\psi}_{\boldsymbol{s}^{(\ell)}}} \right\|_{\mathrm{Tr}} }_{\varepsilon(\approxrep{\psi}_{\boldsymbol{s}^{(\ell)}})} \, .
        \end{equation}
    \end{proposition}

    \begin{proof}
        First, we observe that for any state $\rho$, for any decomposition of it
        \begin{equation}
            \rho = \sum_{\boldsymbol{s}} p_{\boldsymbol{s}} \ketbra{\psi_{\boldsymbol{s}}}{\psi_{\boldsymbol{s}}}
        \end{equation}
        and an approximation procedure that takes a pure state decomposition of the mixed state and approximates each element individually as in Eq.~\eqref{eq:approx_on_decomposition},
        then the approximation error on the full state after approximating each pure state individually can be bounded by the sum of the approximation errors on each pure state,
        \begin{align}
            \left\| \rho - \operatorname{Approx}[\rho] \right\|_{\mathrm{Tr}} & 
            = \left\| \sum_{\boldsymbol{s}} {p}_{\boldsymbol{s}}  \ketbra{\psi_{\boldsymbol{s}}}{\psi_{\boldsymbol{s}}} - \ketbra{\operatorname{Approx}[\psi_{\boldsymbol{s}}]}{\operatorname{Approx}[\psi_{\boldsymbol{s}}]} \right\|_{\mathrm{Tr}} \\ 
             \nonumber
            & \leq \sum_{\boldsymbol{s}} {p}_{\boldsymbol{s}}  \left\| \ketbra{\psi_{\boldsymbol{s}}}{\psi_{\boldsymbol{s}}} - \ketbra{\operatorname{Approx}[\psi_{\boldsymbol{s}}]}{\operatorname{Approx}[\psi_{\boldsymbol{s}}]} \right\|_{\mathrm{Tr}} \, .
        \end{align}
        That is a direct result of the triangle inequality.
        Applying this to a single layer of our process and adopting the notation from Table~\ref{tab:notation_memo}, 
        where we replace $\rho$ with the output state from each channel layer $\Lambda^{(\ell)} \left[ \approxrep{\rho}^{(\ell-1)} \right]$, 
        results in 
        \begin{align}
            \left\| \preapproxrep{\rho}^{(\ell)} - \approxrep{\rho}^{(\ell)} \right\|_{\mathrm{Tr}} &
            = \left\| \Lambda^{(\ell)} \left[ \approxrep{\rho}^{(\ell-1)} \right] - \operatorname{Approx} \circ \Lambda^{(\ell)} \left[ \approxrep{\rho}^{(\ell-1)} \right] \right\|_{\mathrm{Tr}} \\
            \nonumber &
            = \left\| \sum_{\boldsymbol{s}^{(\ell)}} \approxrep{p}_{\boldsymbol{s}^{(\ell)}} \left( \ketbra{\preapproxrep{\psi}_{\boldsymbol{s}^{(\ell)}}}{\preapproxrep{\psi}_{\boldsymbol{s}^{(\ell)}}} - \ketbra{\approxrep{\psi}_{\boldsymbol{s}^{(\ell)}}}{\approxrep{\psi}_{\boldsymbol{s}^{(\ell)}}} \right) \right\|_{\mathrm{Tr}} \\ 
            \nonumber&
            \leq \sum_{\boldsymbol{s}^{(\ell)}} \approxrep{p}_{\boldsymbol{s}^{(\ell)}}  \left\| \ketbra{\preapproxrep{\psi}_{\boldsymbol{s}^{(\ell)}}}{\preapproxrep{\psi}_{\boldsymbol{s}^{(\ell)}}} - \ketbra{\approxrep{\psi}_{\boldsymbol{s}^{(\ell)}}}{\approxrep{\psi}_{\boldsymbol{s}^{(\ell)}}} \right\|_{\mathrm{Tr}} \, .
        \end{align}
        Putting things together, we get a bound on the approximation error for the full evolution.
        Considering again the full layered quantum process described by $L$ quantum channels
        $ \left \{ \Lambda^{(\ell)} \right \}_{\ell = 1}^L$
        and inserting the result from Proposition~\ref{result:error_per_layer} results in the desired expression
        \begin{equation}
            \left\| \rho^{(L)} - \approxrep{\rho}^{(L)} \right\|_{\mathrm{Tr}} 
            \leq \sum_{\ell = 1}^L \sum_{\boldsymbol{s}^{(\ell)}} \approxrep{p}_{\boldsymbol{s}^{(\ell)}}  \underbrace{\left\| \ketbra{\preapproxrep{\psi}_{\boldsymbol{s}^{(\ell)}}}{\preapproxrep{\psi}_{\boldsymbol{s}^{(\ell)}}} - \ketbra{\approxrep{\psi}_{\boldsymbol{s}^{(\ell)}}}{\approxrep{\psi}_{\boldsymbol{s}^{(\ell)}}} \right\|_{\mathrm{Tr}}}_{\varepsilon(\approxrep{\psi}_{\boldsymbol{s}^{(\ell)}})} \, ,
        \end{equation}
        which ends the proof.
    \end{proof}

        As a last intermediate step, let us show that it is equivalent to summing the errors of the elements of the decomposition at each layer, weighted by the respective probability, or summing the cumulated errors of each final trajectory, weighted by the final probability. 
    
        \begin{proposition}
        \label{result:error_sums_equivalence}
            Given a state $\approxrep{\rho}^{(L)}$ obtained from consecutive applications of quantum channels and approximations as in Eq.~\eqref{eq:approx_noisy_state_decomposition},
            the probabilities $\approxrep{p}_{\boldsymbol{s}^{(\ell)}}$ of each element $\ket{\approxrep{\psi}_{\boldsymbol{s}^{(\ell)}}}$ of the decomposition at depth $\ell$ and
            the corresponding approximation errors $\varepsilon(\approxrep{\psi}_{\boldsymbol{s}^{(\ell)}})$,
            then the sum of errors layer by layer of the weighted errors of each element is equivalent to the weighted sum of the cumulated error of each full trajectory
            \begin{equation}
            \label{eq:layer_to_trajectory_errors}
                \sum_{\ell = 1}^L \sum_{\boldsymbol{s}^{(\ell)}} \approxrep{p}_{\boldsymbol{s}^{(\ell)}} \varepsilon(\approxrep{\psi}_{\boldsymbol{s}^{(\ell)}})
                = \sum_{\boldsymbol{s}^{(L)}} \approxrep{p}_{\boldsymbol{s}^{(L)}} \varepsilon_{\mathrm{tot}}(\approxrep{\psi}_{\boldsymbol{s}^{(L)}}) \, ,
            \end{equation}
            where $\varepsilon_{\mathrm{tot}}(\cdot)$ is the cumulated errors of one trajectory as defined in Eq.~\eqref{eq:cumulated_trajectory_error}.
        \end{proposition}
        \begin{proof}
            First, we show that ``splittings of the tree at deeper layers do not matter'' when dealing with errors at layer $\ell$
            \begin{equation}
            \label{eq:error_sums_equivalence}
                \sum_{\boldsymbol{s}^{(L)}} \approxrep{p} (\boldsymbol{s}^{(L)}) \varepsilon(\approxrep{\psi}_{\boldsymbol{s}^{(\ell)}})
                =
                \sum_{\boldsymbol{s}^{(\ell)}} \approxrep{p} (\boldsymbol{s}^{(\ell)}) \varepsilon(\approxrep{\psi}_{\boldsymbol{s}^{(\ell)}}) \, .
            \end{equation}
            For this, we divide the sum to separate the splittings of the tree before and after layer $\ell$. 
            Then we use the fact that everything that happens to the right does not matter since it sums up to a factor of 1. 
            Lastly, the remaining sum of probabilities reduces to the probability of the node of interest in the tree,
            \begin{align}
                \sum_{\boldsymbol{s}^{(L)}} \approxrep{p} (\boldsymbol{s}^{(L)}) \varepsilon(\approxrep{\psi}_{\boldsymbol{s}^{(\ell)}}) & 
                = \sum_{\boldsymbol{s}^{(\ell)}} \sum_{\boldsymbol{s}^{(L)} | \boldsymbol{s}^{(\ell)}} \approxrep{p} (\boldsymbol{s}^{(L)}) \varepsilon(\approxrep{\psi}_{\boldsymbol{s}^{(\ell)}}) \\
                \nonumber & 
                = \sum_{\boldsymbol{s}^{(\ell)}} \sum_{\boldsymbol{s}^{(L)} | \boldsymbol{s}^{(\ell)}} \approxrep{p} (\boldsymbol{s}^{(\ell)}) \approxrep{p} (\boldsymbol{s}^{(L)} | \boldsymbol{s}^{(\ell)}) \varepsilon(\approxrep{\psi}_{\boldsymbol{s}^{(\ell)}}) \\ 
                 \nonumber& 
                = \sum_{\boldsymbol{s}^{(\ell)}} \approxrep{p} (\boldsymbol{s}^{(\ell)}) \varepsilon(\approxrep{\psi}_{\boldsymbol{s}^{(\ell)}}) \underbrace{\sum_{\boldsymbol{s}^{(L)} | \boldsymbol{s}^{(\ell)}}  \approxrep{p} (\boldsymbol{s}^{(L)} | \boldsymbol{s}^{(\ell)})}_{=1} \\ 
                 \nonumber& 
                = \sum_{\boldsymbol{s}^{(\ell)}} \approxrep{p} (\boldsymbol{s}^{(\ell)}) \varepsilon(\approxrep{\psi}_{\boldsymbol{s}^{(\ell)}}). 
                 \nonumber
            \end{align}
            The first equality results from rewriting the sum over all leaves as the sum over all intermediate nodes at depth $\ell$ and the sum over all leaves conditioned on the intermediate state at depth $\ell$.
            The second is again the application of Bayes' rule.
            The third is commuting the sum over the labels after depth $\ell$ with $\approxrep{p} (\boldsymbol{s}^{(\ell)})$ and $\varepsilon(\approxrep{\psi}_{\boldsymbol{s}^{(\ell)}})$, since they do not depend on the last $L-\ell$ elements of the label. 
            The last equation uses Eq.~\eqref{eq:conditional_tree_distributions}.
            Now, to the equation we want to prove, starting from the right-hand side
            \begin{align}
                \sum_{\boldsymbol{s}^{(L)}} \approxrep{p} (\boldsymbol{s}^{(L)}) \varepsilon_{\mathrm{tot}}(\approxrep{\psi}_{\boldsymbol{s}^{(L)}}) & 
                = \sum_{\boldsymbol{s}^{(L)}} \approxrep{p} (\boldsymbol{s}^{(L)}) \left( \sum_{\ell = 1}^L \varepsilon(\approxrep{\psi}_{\boldsymbol{s}^{(\ell)}}) \right) \\
                 \nonumber& 
                = \sum_{\ell = 1}^L \sum_{\boldsymbol{s}^{(L)}} \approxrep{p} (\boldsymbol{s}^{(L)}) \varepsilon(\approxrep{\psi}_{\boldsymbol{s}^{(\ell)}}) \\ &
                = \sum_{\ell = 1}^L \sum_{\boldsymbol{s}^{(\ell)}} \approxrep{p} (\boldsymbol{s}^{(\ell)}) \varepsilon(\approxrep{\psi}_{\boldsymbol{s}^{(\ell)}}). 
                 \nonumber
            \end{align}
            The first equality is given by the definition of the cumulated error in one full trajectory, Eq.~\eqref{eq:cumulated_trajectory_error}.
            Then we commuted the order of the sums. 
            Lastly, the final identity is given by Eq.~\eqref{eq:error_sums_equivalence}.
        \end{proof}
        Using Proposition~\ref{result:error_sums_equivalence} on the outcome of Proposition~\ref{result:error_per_pure_state} gives us the result we aimed for
        \begin{equation}
            \left\| \rho^{(L)} - \approxrep{\rho}^{(L)} \right\|_{\mathrm{Tr}} 
            \leq  \sum_{\boldsymbol{s}^{(L)}} \approxrep{p}_{\boldsymbol{s}^{(L)}} \varepsilon_{\mathrm{tot}}(\approxrep{\psi}_{\boldsymbol{s}^{(L)}}) \, ,
        \end{equation}
    which again ends the proof.
    \end{proof}

    Proposition~\ref{result:systematic_error} states that the systematic error can be bounded by the weighted sum over all paths of the sums of the approximation errors (or upper bounds thereof) at each layer. 
    This error per trajectory could become arbitrarily large (for infinitely deep circuits).
    We will show that it can be bounded by a constant independent of system size and circuit depth.
    
    \begin{proposition}[Path error bound]
        \label{result:path_error_bound}
        Given the setting from Proposition~\ref{result:systematic_error}, the contribution of each path to the total systematic error can be bounded by any user-defined threshold $2 \leq \varepsilon_{\mathrm{max}} \leq 2L$
        \begin{equation}
        \label{eq:systematic_error_with_constant_bound}
            \varepsilon_{\mathrm{sys}} \leq \sum_{\boldsymbol{s}^{(L)}} \approxrep{p}_{\boldsymbol{s}^{(L)}} \varepsilon_{\mathrm{bound}}(\approxrep{\psi}_{\boldsymbol{s}^{(L)}}) 
            \quad \text{with} \quad 
            \varepsilon_{\mathrm{bound}} (\cdot) = 
            \begin{cases}
                \varepsilon_{\mathrm{tot}}(\cdot) & \text{if } \varepsilon_{\mathrm{tot}}(\cdot) \leq \varepsilon_{\mathrm{max}} -2 \\ 
                \varepsilon_{\mathrm{max}} & \text{otherwise} \, .
            \end{cases}
        \end{equation}
    \end{proposition}

    \begin{proof}
        From Proposition~\ref{result:error_per_layer}, we know that the total error of the simulation is bounded by the sum of the errors per layer. 
        Let us consider some layer $\zeta$ at which some path reaches a cumulated error per path (as per Eq.~\eqref{eq:cumulated_trajectory_error}) which exceeds the user-defined threshold.
        Let us call the specific path up to depth $\zeta$ of the unraveling $\boldsymbol{s}^{\lightning(\zeta)}$, 
        such that $\varepsilon_{\mathrm{tot}} (\approxrep{\psi}_{\boldsymbol{s}^{\lightning(\zeta)}}) > \varepsilon_{\mathrm{max}} - 2$ but $\varepsilon_{\mathrm{tot}} (\approxrep{\psi}_{\boldsymbol{s}^{\lightning(\zeta-1)}}) \leq \varepsilon_{\mathrm{max}} - 2$.
        One could chose for instance $\varepsilon_{\mathrm{max}} - 2 = 2$, as in general worst case bounds on distances between states give $\norm{\rho - \sigma}_{\mathrm{Tr}} \leq 2$ for any $\rho$ and $\sigma$.
        The state at this point of the simulation would be $\ket{\approxrep{\psi}_{\boldsymbol{s}^{\lightning(\zeta)}}}$.
        Let us now redefine the approximation procedure from Proposition~\ref{result:error_per_pure_state} at layer $\ell$ such that 
        \begin{equation}
            \operatorname{Approx}^\lightning [\ket{\preapproxrep{\psi}_{\boldsymbol{s}^{(\zeta)}}}]
            = \begin{cases}
                \ket{\approxrep{\psi}_{\boldsymbol{s}^{(\zeta)}}} & 
                \text{if } \boldsymbol{s}^{(\zeta)} \neq \boldsymbol{s}^{\lightning(\zeta)} \\
                0 & 
                \text{if } \boldsymbol{s}^{(\zeta)} = \boldsymbol{s}^{\lightning(\zeta)} \, .
            \end{cases}
        \end{equation}
        Then all trajectories are evolved as before, except the trajectories such that $\boldsymbol{s}^{(L)} [1:\zeta] = \boldsymbol{s}^{\lightning(\zeta)}$ trajectory, which are mapped to the zero amplitude state. 
        The error committed in this layer for these trajectories (Eq.~\eqref{eq:simulation_error_per_sample}) is then 
        \begin{equation}
            \left\| \ketbra{\preapproxrep{\psi}_{\boldsymbol{s}^{\lightning(\zeta)}}}{\preapproxrep{\psi}_{\boldsymbol{s}^{\lightning(\zeta)}}} - \ketbra{\approxrep{\psi}_{\boldsymbol{s}^{\lightning(\zeta)}}}{\approxrep{\psi}_{\boldsymbol{s}^{\lightning(\zeta)}}} \right\|_{\mathrm{Tr}}
            = \left\| \ketbra{{\psi}_{\boldsymbol{s}^{\lightning(\zeta)}}}{{\psi}_{\boldsymbol{s}^{\lightning(\zeta)}}} - 0 \right\|_{\mathrm{Tr}}
            = 1 \, .
        \end{equation}
        Since any quantum channel leaves the zero-amplitude state unaffected, $\Lambda[0] = 0$ $\forall \Lambda$, these trajectories are effectively removed from the consideration of the computation of error. 
        The state at later depths (Eq.~\eqref{eq:approx_noisy_state_decomposition}) can then be written
        \begin{equation}
            \approxrep{\rho}^{(\ell)} 
            = \sum_{\substack{\boldsymbol{s}^{(\ell)} : \\ \mathclap{\boldsymbol{s}^{(\ell)}[1:\zeta] \neq \boldsymbol{s}^{\lightning(\zeta)}}}} 
            \approxrep{p}_{\boldsymbol{s}^{(\ell)}} \ketbra{\approxrep{\psi}_{\boldsymbol{s}^{(\ell)}}}{\approxrep{\psi}_{\boldsymbol{s}^{(\ell)}}}
            + 0 \sum_{\substack{\boldsymbol{s}^{(\ell)} : \\ \mathclap{\boldsymbol{s}^{(\ell)}[1:\zeta] = \boldsymbol{s}^{\lightning(\zeta)}}}} 
            \approxrep{p}_{\boldsymbol{s}^{(\ell)}} 
        \end{equation}
        and the zero trajectories do not contribute to the error (neither in the layer-wise perspective from the left-hand side of Eq.~\eqref{eq:layer_to_trajectory_errors}, nor from the path-wise approach from the right-hand side).
        The convex combination of trajectories results in a sub-normalized state (since $\Tr{}{\approxrep{\rho}^{(\ell)}} \leq 1$) but one can complement the probabilities to form the distribution defined by the evolution where the trajectories where not mapped to $0$.
        Before applying the last layer, the errors of each trajectory are then 
        \begin{equation}
            \varepsilon_{\mathrm{tot}}(\approxrep{\psi}_{\boldsymbol{s}^{(L-1)}})
            = \begin{cases}
                \sum_{\ell = 1}^{L-1} \varepsilon(\approxrep{\psi}_{\boldsymbol{s}^{(\ell)}}) & 
                \text{if } \boldsymbol{s}^{(L-1)}[1:\zeta] \neq \boldsymbol{s}^{\lightning(\zeta)} \\
                \sum_{\ell = 1}^{\zeta-1} \varepsilon(\approxrep{\psi}_{\boldsymbol{s}^{(\ell)}}) + 1 & 
                \text{if } \boldsymbol{s}^{(L-1)}[1:\zeta] = \boldsymbol{s}^{\lightning(\zeta)} \, .
            \end{cases}
        \end{equation}
        We then also redefine the approximation procedure of the last layer such that it performs the normal approximation on states such that $\boldsymbol{s}^{(L)}[1:\zeta] \neq \boldsymbol{s}^{\lightning(\zeta)}$ while it maps the zero-amplitude state to the mixture of states that would be obtained from not having mapped to zero and perform the normal unraveling
        \begin{equation}
            \operatorname{Approx}^\lightning [0] 
            = \sum_{\substack{\boldsymbol{s}^{(L)} : \\ \mathclap{\boldsymbol{s}^{(L)}[1:\zeta] = \boldsymbol{s}^{\lightning(\zeta)}}}} 
            \approxrep{p}_{\boldsymbol{s}^{(L)}} \ketbra{\approxrep{\psi}_{\boldsymbol{s}^{(L)}}}{\approxrep{\psi}_{\boldsymbol{s}^{(L)}}} \, .
        \end{equation}
        The contribution to the error of the last layer (in Eq.~\eqref{eq:error_per_layer}) from this procedure is also bounded as 
        \begin{equation}
            \norm{\sum_{\substack{\boldsymbol{s}^{(L)} : \\ \mathclap{\boldsymbol{s}^{(L)}[1:\zeta] = \boldsymbol{s}^{\lightning(\zeta)}}}} 
            \approxrep{p}_{\boldsymbol{s}^{(L)}} \ketbra{\approxrep{\psi}_{\boldsymbol{s}^{(L)}}}{\approxrep{\psi}_{\boldsymbol{s}^{(L)}}} - 0}_{\mathrm{Tr}}
            \leq \sum_{\substack{\boldsymbol{s}^{(L)} : \\ \mathclap{\boldsymbol{s}^{(L)}[1:\zeta] = \boldsymbol{s}^{\lightning(\zeta)}}}} 
            \approxrep{p}_{\boldsymbol{s}^{(L)}} \norm{\ketbra{\approxrep{\psi}_{\boldsymbol{s}^{(L)}}}{\approxrep{\psi}_{\boldsymbol{s}^{(L)}}} - 0}_{\mathrm{Tr}}
            = \sum_{\substack{\boldsymbol{s}^{(L)} : \\ \mathclap{\boldsymbol{s}^{(L)}[1:\zeta] = \boldsymbol{s}^{\lightning(\zeta)}}}} 
            \approxrep{p}_{\boldsymbol{s}^{(L)}}
            \, .
        \end{equation}
        Then the sum of the errors of the trajectories in each layer from Eq.~\eqref{eq:simulation_error_per_sample} becomes
        \begin{equation}
            \left\| \rho^{(L)} - \approxrep{\rho}^{(L)} \right\|_{\mathrm{Tr}} 
            \leq \sum_{\ell = 1}^{\zeta-1} \sum_{\boldsymbol{s}^{(\ell)}} \approxrep{p}_{\boldsymbol{s}^{(\ell)}} \varepsilon(\approxrep{\psi}_{\boldsymbol{s}^{(\ell)}}) 
            + \sum_{\ell = \zeta}^L \left( \sum_{\substack{\boldsymbol{s}^{(\ell)} : \\ {\boldsymbol{s}^{(\ell)}[1:\zeta] \neq \boldsymbol{s}^{\lightning(\zeta)}}}} \approxrep{p}_{\boldsymbol{s}^{(\ell)}} \varepsilon(\approxrep{\psi}_{\boldsymbol{s}^{(\ell)}})
            + 2 \sum_{\substack{\boldsymbol{s}^{(\ell)} : \\ \mathclap{\boldsymbol{s}^{(\ell)}[1:\zeta] = \boldsymbol{s}^{\lightning(\zeta)}}}} \approxrep{p}_{\boldsymbol{s}^{(\ell)}}
            \right)
            \, .
        \end{equation}
        Now, adopting the trajectory-wise perspective as in Proposition~\ref{result:error_sums_equivalence} allows to rewrite
        \begin{equation}
            \left\| \rho^{(L)} - \approxrep{\rho}^{(L)} \right\|_{\mathrm{Tr}} 
            \leq \sum_{\substack{\boldsymbol{s}^{(L)} : \\ \mathclap{\boldsymbol{s}^{(L)}[1:\zeta] \neq \boldsymbol{s}^{\lightning(\zeta)}}}} \approxrep{p}_{\boldsymbol{s}^{(L)}}
            \sum_{\ell = 1}^L   \varepsilon(\approxrep{\psi}_{\boldsymbol{s}^{(\ell)}})
            + \sum_{\substack{\boldsymbol{s}^{(L)} : \\ \mathclap{\boldsymbol{s}^{(L)}[1:\zeta] = \boldsymbol{s}^{\lightning(\zeta)}}}} \approxrep{p}_{\boldsymbol{s}^{(\ell)}} \left( 2+
            \sum_{\ell = 1}^{\zeta-1}  \varepsilon(\approxrep{\psi}_{\boldsymbol{s}^{(\ell)}}) \right)
            \, .
        \end{equation}
        
        We have only defined one path with large error $\boldsymbol{s}^{\lightning(\zeta)}$ and assumed that all other paths had small error. 
        This procedure can, however, be repeated for all paths crossing the error threshold, even if it happens at different depths. 
        We can now rewrite the error bound in the desired form, using the fact that the sum in the parenthesis is bounded by $\varepsilon_{\mathrm{max}} - 2$ (since it is how we defined $\zeta$ above)
        \begin{equation}
            \left\| \rho^{(L)} - \approxrep{\rho}^{(L)} \right\|_{\mathrm{Tr}} \leq \sum_{\boldsymbol{s}^{(L)}} \approxrep{p}_{\boldsymbol{s}^{(L)}} \varepsilon_{\mathrm{bound}}(\approxrep{\psi}_{\boldsymbol{s}^{(L)}}) 
            \quad \text{with} \quad 
            \varepsilon_{\mathrm{bound}} (\cdot) = 
            \begin{cases}
                \varepsilon_{\mathrm{tot}}(\cdot) & \text{if } \varepsilon_{\mathrm{tot}}(\cdot) \leq \varepsilon_{\mathrm{max}} -2 \\ 
                \varepsilon_{\mathrm{max}} & \text{otherwise} \, .
            \end{cases}
        \end{equation}
    \end{proof}

    With this, we have derived a bound on the systematic error of our simulation. 
    However, its computation would require knowledge of all trajectories. 
    Computing the exact value of the systematic error would thus defeat the purpose of the algorithm. 
    We will next see how it can be estimated, leading to the result of Theorem~\ref{result:sampled_trajectories}.
        
    \begin{proposition}
    \label{result:systematic_error_bound}
        Given the setting from Proposition~\ref{result:path_error_bound} and
        $N$ pure state trajectories $\ket{\hat{\psi}_i} \in \{ \ket{\approxrep{\psi}_{\boldsymbol{s}^{(L)}}} \}_{\boldsymbol{s}^{(L)}}$ sampled according to the probability distribution $\{ \approxrep{p}_{\boldsymbol{s}^{(L)}} \}_{\boldsymbol{s}^{(L)}}$ from the decomposition of $\approxrep{\rho}^{(L)}$, 
        then the bound on the systematic error in Eq.~\eqref{eq:systematic_error} can be approximated with high probability by the average approximation error of the samples.
        \begin{equation}
            \sum_{\boldsymbol{s}^{(L)}} \approxrep{p}_{\boldsymbol{s}^{(L)}} \varepsilon_{\mathrm{bound}}(\approxrep{\psi}_{\boldsymbol{s}^{(L)}})
            \leq \hat{\varepsilon} + \sqrt{\frac{\varepsilon_{\mathrm{max}}^2}{2N} \log \left( \frac{1}{\delta}\right)}
            \quad \text{with probability} \quad 
            1-\delta
        \end{equation}
        where $\hat{\varepsilon} = \frac{1}{N} \sum_{i=1}^N \varepsilon_{\mathrm{bound}} (\hat{\psi}_i)$ and $\varepsilon_{\mathrm{bound}}(\cdot)$ is defined as in Eq.~\eqref{eq:systematic_error_with_constant_bound}.
    \end{proposition}

    \begin{proof}
        Let us show that the bounds on the systematic error can be approximated with high probability by the average of the approximation errors of the sampled trajectories.
        We define the estimator for the systematic error
        \begin{equation}
            \hat{\varepsilon} = \frac{1}{N} \sum_{i=1}^N \varepsilon_{\mathrm{bound}} (\hat{\psi}_i) \, .
        \end{equation}
        We will now bound the difference between the estimator and the actual upper bound for the error (Eq.~\eqref{eq:systematic_error}) using Hoeffding's inequality~\cite{hoeffdingProbabilityInequalitiesSums1963}.
        Given a set of independent random variables $\{ X_i \}_{i=1}^N$ with $a \leq X_i \leq b$ and considering the sum of these random variables $S_N = \sum_{i=1}^N X_i$, then 
        \begin{equation}
            \mathds{P} \left( S_N - \mathds{E}[S_N] \geq \epsilon \right) \leq \ee^{-\frac{2 \epsilon^2}{N(b-a)^2}} \, .
        \end{equation}
        Choosing the set of random variables $X_i = \frac{1}{N} \varepsilon_{\mathrm{bound}} (\hat{\psi}_i)$, 
        they are bounded $0 \leq X_i \leq \frac{\varepsilon_{\mathrm{max}}}{N}$
        and their expectation value is given by
        \begin{equation}
            \mathds{E} \left[ X_i \right]
            = \frac{1}{N} \sum_{\boldsymbol{s}^{(L)}} \approxrep{p}_{\boldsymbol{s}^{(L)}} \varepsilon_{\mathrm{bound}} \left( \approxrep{\psi}_{\boldsymbol{s}^{(L)}} \right) \, .
        \end{equation}
        Taking the sum over $N$ elements, the expectation value of $S_N$ is precisely what we have from Proposition~\ref{eq:systematic_error_with_constant_bound}. 
        Then
        \begin{equation}
            \mathds{P} \left( \hat{\varepsilon} -\sum_{\boldsymbol{s}^{(L)}} \approxrep{p}_{\boldsymbol{s}^{(L)}} \varepsilon_{\mathrm{bound}}(\approxrep{\psi}_{\boldsymbol{s}^{(L)}}) \geq \epsilon \right) \leq \ee^{-\frac{2\epsilon^2 N}{\varepsilon_{\mathrm{max}}^2}} \, .
        \end{equation}
        Thus, for a given probability of failing at estimating the approximation error $\delta$, we have that with probability $1 - \delta$ the error on the systematic error is bounded 
        \begin{equation}
        \label{eq:systematic_error_bound}
            \varepsilon_{\mathrm{sys}} 
            = \left \| \rho^{(L)} - \approxrep{\rho}^{(L)} \right \|_{\mathrm{Tr}}
            \leq \hat{\varepsilon} + \sqrt{\frac{\varepsilon_{\mathrm{max}}^2}{2N} \log \left( \frac{1}{\delta}\right)}
            \, .
        \end{equation}

    \end{proof}
    Combining Proposition~\ref{result:systematic_error} and Proposition~\ref{result:systematic_error_bound} we obtain 
    \begin{equation}
        \varepsilon_{\mathrm{sys}} 
        = \left \| \rho^{(L)} - \approxrep{\rho}^{(L)} \right \|_{\mathrm{Tr}}
        \leq  \sum_{\boldsymbol{s}^{(L)}} \approxrep{p}_{\boldsymbol{s}^{(L)}} \varepsilon_{\mathrm{bound}}(\approxrep{\psi}_{\boldsymbol{s}^{(L)}})
        \leq \hat{\varepsilon} + \sqrt{\frac{\varepsilon_{\mathrm{max}}^2}{2N} \log \left( \frac{1}{\delta}\right)}
    \end{equation}
    where the second inequality is satisfied with probability $1-\delta$. 
    Thus, the first is fulfilled with at least the same probability, completing the proof.
\end{proof}

\subsection{Errors based on MPS truncations}
\label{sec:MPS-based_errors}

    We now want to restrict to the specific setting where the underlying data structure of the simulation are MPS, 
    and the approximation procedure is the truncation of the MPS to a certain maximum bond dimension, followed by a renormalization. 
    In particular, we want to show how to compute $\varepsilon_{\mathrm{sys}}$ or $\hat{\varepsilon}$. 
    Both are the result of the sum (over trajectories and over layers) of the single-layer approximation error in Eq.~\eqref{eq:single_approx_error}. 
    Each is the trace norm of the difference between two pure-state density matrices (rank-$1$ projectors).
    We can use Proposition~\ref{result:bounds_trace_to_2_norm} (1) to relate it to $2$-norm vector distances 
    \begin{equation}
        \norm{\ketbra{\psi}{\psi} - \ketbra{\phi}{\phi}}_{\mathrm{Tr}} 
        \leq 2 \Norm{\ket{\psi} - \ket{\phi}}_2 \, .
    \end{equation}
    On the other hand, from Eq.~\eqref{eq:MPS_truncation_error}, we have that the (vector) $2$-norm truncation error (before renormalization) at one layer is bounded by the sum over all bonds of the sum of the squared discarded singular values.
    Let us call this value $\varepsilon_{\mathrm{trunc}}(\psi)$
    \begin{equation}
        \| \ket{\psi} - \ket{\psi'} \|_2 
        = \varepsilon_{\mathrm{trunc}}(\psi)
        \leq \sqrt{2 \sum_{k=1}^{n-1} \epsilon_k(d_k)} 
        \quad \text{with} \quad 
        \epsilon_k(d_k) = \sum_{i=d_k+1}^{d_{k,{\max}}} s_i^2 \, .
    \end{equation}
    Due to the triangle inequality, 
    we have that $\norm{\ket{\psi'}}_2 \geq 1-\varepsilon_{\mathrm{trunc}}(\psi)$.
    Considering now the renormalized version of $\ket{\psi'}$, we have
    \begin{align}
        \Norm{\ket{\psi} - \frac{\ket{\psi'}}{\norm{\ket{\psi'}}_2}}_2 & 
        \leq \Norm{\ket{\psi} - \ket{\psi'}} + \Norm{\ket{\psi'} - \frac{\ket{\psi'}}{\norm{\ket{\psi'}}_2}}_2 \\ 
        \nonumber & 
        = \varepsilon_{\mathrm{trunc}}(\psi) - \Norm{\ket{\psi'}}_2 \left|1 - \frac{1}{\norm{\ket{\psi'}}_2} \right| \\ 
        \nonumber
        & 
        = \varepsilon_{\mathrm{trunc}}(\psi) - \Norm{\ket{\psi'}}_2 + 1 \\ & 
        \leq \varepsilon_{\mathrm{trunc}}(\psi) - (1-\varepsilon_{\mathrm{trunc}}(\psi)) + 1 \\ & 
        = 2 \varepsilon_{\mathrm{trunc}}(\psi) \, .
        \nonumber
    \end{align}
    Then, putting all steps  
    together, we have that 
    \begin{align}
        \varepsilon(\approxrep{\psi}_{\boldsymbol{s}^{(L)}}, \ell)
        = \varepsilon(\approxrep{\psi}_{\boldsymbol{s}^{(\ell)}}) & 
        = \left\| \ketbra{\preapproxrep{\psi}_{\boldsymbol{s}^{(\ell)}}}{\preapproxrep{\psi}_{\boldsymbol{s}^{(\ell)}}} - \ketbra{\approxrep{\psi}_{\boldsymbol{s}^{(\ell)}}}{\approxrep{\psi}_{\boldsymbol{s}^{(\ell)}}} \right\|_{\mathrm{Tr}} \\ 
        \nonumber& 
        \leq 2 \left\| \ket{\preapproxrep{\psi}_{\boldsymbol{s}^{(\ell)}}} - \ket{\approxrep{\psi}_{\boldsymbol{s}^{(\ell)}}} \right\|_{\mathrm{2}} \\ 
        \nonumber& 
        \leq 4  \varepsilon_{\mathrm{trunc}}(\approxrep{\psi}_{\boldsymbol{s}^{(\ell)}})
        \nonumber
        \, .
    \end{align}
    This results in a computable estimator
    \begin{equation}
        \hat{\varepsilon} \leq \frac{4}{N} \sum_{i=1}^N \left( \sum_{\ell=1}^L \varepsilon_{\mathrm{trunc}}(\hat{\psi}_i, \ell) \right)
    \end{equation}
    where $\varepsilon_{\mathrm{trunc}}(\hat{\psi}_i, \ell)$ is the square root of the sum of discarded singular values on all bonds as defined above.

\subsection{Improvements of the error bounds due to concentration}
\label{sec:concentration_bounds}

    In proof of Theorem~\ref{result:sampled_trajectories}, we use the inequality
    \begin{align}
         \left\| \Lambda^{(\ell+1)} \left[\rho^{(\ell)}\right] - \Lambda^{(\ell+1)} \left[ \approxrep{\rho}^{(\ell)} \right] \right\|_{\mathrm{Tr}} \leq \| \rho^{(\ell)} - \approxrep{\rho}^{(\ell)} \|_{\mathrm{Tr}}=\epsilon_\ell .
    \end{align}
    This inequality relies on the fact that the application of a quantum channel cannot increase the trace distance of our approximated state and the true state. 
    However, in general, this is not tight for non-unitary channels and guaranteed not to be if $\Lambda^{\ell+1}$ is a single fixed-point channel. 
    As such, we can write 
    \begin{align}
        \left\| \Lambda^{(\ell+1)} \left[\rho^{(\ell)}\right] - \Lambda^{(\ell+1)} \left[ \approxrep{\rho}^{(\ell)} \right] \right\|_{\mathrm{Tr}} \leq \alpha \| \rho^{(\ell)} - \approxrep{\rho}^{(\ell)} \|_{\mathrm{Tr}}=\alpha \epsilon_\ell 
    \end{align}
    with $\alpha < 1$. 
    In this case, we obtain a new bound
    \begin{align}
        \hat{\varepsilon} = \frac{1}{N} \sum_{i=1}^N \left( \sum_{\ell=1}^L \alpha^{L-\ell}\varepsilon_{\mathrm{trunc}}(\hat{\psi}_i, \ell) \right)
    \end{align}
    giving an exponential suppression of old errors.
    In the case of local depolarizing noise with noise rate $p$, where it is known that very strong convergence to the maximally mixed state occurs, we can leverage the strong data processing inequality to obtain
    \begin{align}
        \hat{\varepsilon} = \frac{1}{N} \sum_{i=1}^N \left( \sum_{\ell=1}^L \min \left( \sqrt{8n}(1-p)^{L-\ell},1 \right) \varepsilon_{\mathrm{trunc}}(\hat{\psi}_i, \ell) \right)
    \end{align}
    For general noise channels, the worst-case guarantees are not that strong, meaning $\alpha$ is very close to $1$ in the worst case. As we consider one particular evolution and set of observables, the effect can be expected to be much stronger for practical implementations. Similar effects have been found for random circuits for many more noise models.

\section{Proofs of the reductions}
\label{sec:reduction_proofs}

\subsection{Proof of Lemma~\ref{result:sampling_reduction}: Output distribution sampling from trajectory sampling}

\samplingreduc*

\begin{proof}
    First, we explain how performing the trajectory sampling and measuring each trajectory follows the same statistics as measuring the full approximate state.
    The output of an $\epsilon \textsc{-tts}(\mathcal{D}, \chi)$ algorithm can be thought of as a matrix-valued random variable. 
    An observation of this random variable is a state vector $\ketbra{\hat{\psi}_i}{\hat{\psi}_i}$ where $\ket{\hat{\psi}_i}$ is an MPS with bond dimension at most $\chi$. 
    Hence one can classically sample from the distribution $P(x|\hat{\psi}_i)$ in runtime $\poly(n,\chi)$~\cite{ferrisPerfectSamplingUnitary2012}. 
    Conditioned on the success of the $\epsilon \textsc{-tts}(\mathcal{D},\chi)$ algorithm (the probability of success is $\geq 1-\delta$), the solution can be viewed as a matrix-valued random variable that is the mixed state $\sigma = \sum_i q_i \ketbra{\psi_i}{\psi_i}$. 
    Indeed each sampled trajectory $\ket{\hat{\psi}_i}$ is sampled from the set $\{ \ket{\psi} \}_i$ with respective probabilities $\{ q_i \}_i$. 
    Then
    \begin{equation}
        \mathds{E}[\ketbra{\hat{\psi}_i}{\hat{\psi}_i}] 
        = \sum_i q_i \ketbra{\psi_i}{\psi_i}
        = \sigma \, .
    \end{equation}
    With this, it also follows that the output probabilities of measuring an instance of the trajectory sampling are the same as when measuring the $\sigma$
    \begin{equation}
        p_{\textsc{tts}}(x) 
        = \sum_i p(\psi_i) p(x| \psi_i) 
        = \sum_i q_i \inner{x}{\psi_i} \! \inner{\psi_i}{x}
        = \bra{x} \sigma \ket{x} 
        = p(x | \sigma)
        \, .
    \end{equation}

    Next, we show that the resulting distribution is close to the target distribution.
    As per the definition of $\epsilon \textsc{-tts}(\mathcal{D}, \chi)$, 
    $\sigma$ is within the $\epsilon$-ball of the target state $\rho$. 
    \begin{equation}
        \norm{\rho - \sigma}_{\mathrm{Tr}} \leq \epsilon \, .
    \end{equation}
    As a result, the error from sampling the measurement output of trajectories with respect to the target distribution (measurement outcomes from the exact state $\rho$) can also be bounded by $\epsilon$.
    This is due to the fact that measuring (as any quantum channel) reduces distances.
    Then 
    \begin{equation}
        TV(\mathcal{P}_{\rho}, \mathcal{P}_{\sigma}) 
        = \frac{1}{2} \sum_x |p(x|\rho) - p(x|\sigma)| 
        \leq D(\rho, \sigma) 
        = \frac{1}{2} \| \rho - \sigma \|_{\mathrm{Tr}} 
        \, .
    \end{equation}
    As a result, the random variable $X$ generated by sampling from the distribution $P(x|\hat{\psi})$ where $\hat{\psi}$ is a solution to $\epsilon \textsc{-tts}(\mathcal{D}, \chi)$ has at most $\frac{1}{2} \epsilon$ total variation distance from the target distribution conditioned on the success of the $\epsilon\textsc{-tts}$ procedure.
\end{proof}

\subsection{Proof of Lemma~\ref{result:expectation_reduction}: Expectation value estimation from trajectory sampling}

\observablereduc*

\begin{proof}
    The task is to estimate the expectation value $\braket{O} = \Tr{}{O \rho}$, where $\rho$ is the exact noisy state, while 
    the trajectory sampling protocol returns samples $\ket{\hat{\psi}_i}$ drawn from $\sigma = \sum_i q_i \ketbra{\psi_i}{\psi_i}$, with 
    \begin{equation}
        \| \rho - \sigma \|_{\mathrm{Tr}} \leq \epsilon
    \end{equation}
    with probability $1-\delta$.
    Let us consider a set of $N$ such trajectories sampled independently, 
    \begin{equation}
        \left\{ \ket{\hat{\psi}_j} \ : \ \ket{\hat{\psi}_j} = \ket{\psi_i} \text{ with probability } q_i \right\}_j \, .
    \end{equation}
     We will construct the estimator of the expectation value as the average of the expectation value of $O$ with respect to the sampled states
    \begin{equation}
        \hat{O} 
        = \frac{1}{N} \sum_{i=1}^N \hat{O}_i
        \quad \text{with} \quad 
        \hat{O}_i = \bra{\hat{\psi}_i} O \ket{\hat{\psi}_i}
         \, .
    \end{equation}
    For observables that can be written as a $\poly(\chi,n)$ bond dimension MPO, $\bra{\psi_i}O\ket{\psi_i}$ can be classically computed in $\poly(\chi,n)$ time.
    The error incurred by this choice of estimator can be divided into two components by application of the triangle inequality
    \begin{equation}
    \label{eq:estimation_error_systematic_statistical}
        \left| \braket{O} - \hat{O} \right| 
        \leq \underbrace{ \left| \Tr{}{\rho O} - \Tr{}{\sigma O} \right| }_{\Delta O_{\mathrm{sys}}} 
        + \underbrace{ \left| \Tr{}{\sigma O} - \hat{O} \right| }_{\Delta O_{\mathrm{stat}}}
    \end{equation}
    which we will call the systematic and the statistical errors.

    The first approach is to interpret the systematic error as given directly by the accuracy parameter of the trajectory sampling simulation $\epsilon$ by using that $\Delta O_{\mathrm{sys}} = |\Tr{}{O \rho} - \Tr{}{O \sigma} \leq \norm{O}_{\mathrm{op}} \norm{\rho - \sigma}_{\mathrm{Tr}}$ (see Section~\ref{sec:norms_toolbox} for more on matrix norms definitions and inequalities).
    The statistical error is bound using Hoeffding's inequality~\cite{hoeffdingProbabilityInequalitiesSums1963}.
    Given the set of random variables $\{ \frac{1}{N}\hat{O}_i \}_{i=1}^N$ with $-\norm{O}_{\mathrm{op}} \leq \hat{O}_i \leq \norm{O}_{\mathrm{op}}$, 
    the resulting estimator $\hat{O}$ is unbiased
    \begin{equation}
        \mathds{E} \left[ \hat{O} \right]
        = \mathds{E} \left[ \hat{O}_i \right]
        = \sum_i q_i \bra{\psi_i} O \ket{\psi_i}
        = \Tr{}{O \sigma} \, .
    \end{equation}
    Then, Hoeffding's inequality gives
    \begin{equation}
        \mathds{P} \left( \left| \hat{O} - \Tr{}{O \sigma} \right| \geq \eta \right) \leq 2 \ee^{-\frac{2 \eta^2 N}{2 \norm{O}^2_{\mathrm{op}}}} \, .
    \end{equation}
    This means that 
    \begin{equation}
        \left| \hat{O} - \Tr{}{O \sigma} \right| \leq \sqrt{\frac{2\norm{O}^2_{\mathrm{op}}}{N} \log\left( \frac{2}{\delta'} \right)}
        \quad \text{with probability } 1 - \delta' 
    \end{equation}
    or alternatively but equivalently that one requires $N = \frac{2\norm{O}^2_{\mathrm{op}}}{\eta^2} \log\left( \frac{2}{\delta'} \right)$ such that $\left| \hat{O} - \Tr{}{O \sigma} \right| \leq \eta$ with probability $1 - \delta'$.
    Eq.~\eqref{eq:estimation_error_systematic_statistical} results in the nesting of the two probably-approximately-correct statements.
    By the union bound, the probability of either bound failing is smaller than $\delta + \delta'$, thus 
    \begin{align}
        \left| \braket{O} - \hat{O} \right| & 
        \leq \left| \Tr{}{\rho O} - \Tr{}{\sigma O} \right| + \left| \Tr{}{\sigma O} - \hat{O} \right| \\ &
        \leq \norm{O}_{\mathrm{op}} \epsilon + \eta
    \end{align}
    with probability at least $1 - \delta - \delta'$. 
    This gives the second formulation of the Lemma.

    Alternatively, one can include the statistical uncertainty of the trajectory sampling into the systematic error of the observable estimation.
    The output of the trajectory sampling algorithm can be thought of as a matrix-valued random variable $\sigma$, which is a convex combination of $\sigma_{succ}$ and some unknown quantum state $\sigma_{fail}$ that is produced conditioned on failure of $\epsilon \textsc{-tts}(\mathcal{D},\chi)$, 
    \begin{equation}
        \sigma = (1-\delta) \sigma_{succ} + \delta \sigma_{fail} \, .
    \end{equation}
    Using Eq.~\eqref{eq:trace_distance_UB}, and the fact that $\| \rho-\sigma_{fail} \|_{\mathrm{Tr}} \leq 2$, it is easy to show that $\| \rho-\sigma \|_{\mathrm{Tr}} \leq (1-\delta) \epsilon + 2\delta$.
    Consequently, one can estimate the expectation value $\Tr{}{O \rho}$ using $\Tr{}{O \sigma}$, as the two quantities deviate by at most an additive error 
    \begin{equation}
        \Delta O_{\mathrm{sys}} 
        = | \Tr{}{O \rho} - \Tr{}{O \sigma} | 
        \leq [(1-\delta) \epsilon + 2\delta] \|O\|_{\mathrm{op}} \, .
    \end{equation}
    To produce an estimate of $\Tr{}{O \sigma}$ accurate up to additive error $\Delta O_{\mathrm{stat}} = \eta$ with probability $\geq 1-\delta'$, a total of $N = 2 \|O\|_{\mathrm{op}}^2 \eta^{-2} \log (2\delta'^{-1})$ samples are sufficient. 
    This results, through the inequality from Eq.~\eqref{eq:estimation_error_systematic_statistical}, in an overall additive error upper bound of 
    \begin{equation}
        \left| \braket{O} - \hat{O} \right|
        \leq [\epsilon + 2 \delta'] \|O\|_{\mathrm{op}} + \eta 
        \quad \text{with probability } 1 - \delta' \, .
    \end{equation}
    This gives the first formulation of Lemma~\ref{result:expectation_reduction}.

\end{proof}

\section{Operator norms toolbox}
\label{sec:norms_toolbox}

\subsection{Definitions}

    Since several different norms are used in this work, here is a summary of their definitions and some useful properties. 
    First, we remind ourselves of the definition of vector $p$-norms. 
    Given a vector $\Vec{x}$ with $d$ elements, 
    \begin{equation}
        \| \Vec{x} \|_p = \left( \sum_{i=1}^d |x_i|^p \right)^{\frac{1}{p}} \, .
    \end{equation}
    The most famous ones are the $1$, $2$ and $\infty$-norms
    \begin{align}
        \| \Vec{x} \|_1 & = \sum_{i=1}^d |x_i|, &
        \| \Vec{x} \|_2 & = \sqrt{\sum_{i=1}^d |x_i|^2} = \sqrt{ \Vec{x}^\dagger \Vec{x}}, &
        \| \Vec{x} \|_\infty & = \max_i{|x_i|}.
    \end{align}
    For matrices, several ways of defining norms exist.
    Most statements in this work are in terms of Schatten norms of operators. 
    They are generally written in terms of the singular values of the matrix
    \begin{equation}
        \|A\|_p = \left( \sum_{i=1}^d \sigma_i(A)^p \right)^{\frac{1}{p}} 
        \quad \text{where } \sigma_i(A) \text{ are the singular values of } A
        \, .
    \end{equation}
    One has to be wary, as they share the same notation as other operator norms but are not equivalent.
    For instance, the vector-norm-induced norms, as given away by the name, are based on the norms of vectors.
    It is the maximum $p$-norm of the vector obtained by applying the matrix to a normalized vector
    \begin{equation}
         \|A\|_p = \sup_{x \neq 0} \dfrac{\|Ax\|_p}{\|x\|_p} \, .
    \end{equation}
    Notably, the vector-norm induced norm for $p=2$ equals the Schatten norm with $p=\infty$.
    For this reason we use the terminology \emph{trace norm} $\norm{A}_{\mathrm{Tr}}$, \emph{Frobenius norm} $\norm{A}_{\mathrm{F}}$ and \emph{operator norm} $\norm{A}_{\mathrm{op}}$ respectively for the Schatten norms with $p=1,2,\infty$.
    Note also the existence of ``entry-wise'' matrix norms, which are not defined here. 
    In particular, it is notable that the Frobenius norm happens to be one as
    \begin{equation}
        \|A\|_{\mathrm{F}} = \sqrt{\sum_{i,j} |a_{ij}|^2} \, .
    \end{equation}
    A few such alternative definitions and properties are summarized in Table~\ref{tab:norms}.

    \begin{table}
    		\centering
    		\begin{tabular}{cllll}
    			& \multicolumn{2}{c}{$\boldsymbol{p}$\textbf{-norm-induced}} & \multicolumn{2}{c}{\textbf{Schatten}} \\
    			\vspace{0.15cm}
    			$\boldsymbol{p}$ & 
                \multicolumn{2}{c}{$\displaystyle \|A\|_p = \sup_{x \neq 0} \dfrac{\|Ax\|_p}{\|x\|_p}$} & 
                \multicolumn{2}{c}{$\displaystyle \|A\|_p = \left( \sum_{i=1}^d \sigma_i(A)^p \right)^{\frac{1}{p}} $} \\
            $\boldsymbol{1}$ & 
                \makecell[l]{Maximum col-\\umn $1$-norm} & $\displaystyle \|A\|_1 = \max_j \sum_i |a_{ij}| $ & 
                \makecell[l]{Nuclear/Trace \\norm} & $\displaystyle \|A\|_{\mathrm{Tr}} = \sum_{i=1}^d \sigma_i(A) = \Tr{}{\sqrt{A^\dagger A}}$ \\
            $\boldsymbol{2}$ & 
                \makecell[l]{Maximum sin-\\gular value} & $\displaystyle \|A\|_2 = \max_i \sigma_i(A) $ & 
                \makecell[l]{Frobenius/Hilbert-\\Schmidt norm} & $\displaystyle \|A\|_{\mathrm{F}} = \left( \sum_{i=1}^d \sigma_i(A)^2 \right)^{\frac{1}{2}} = \sqrt{\Tr{}{A^\dagger A}}$ \\
            $\boldsymbol{\infty}$ & 
                \makecell[l]{Maximum row \\1-norm} & $\displaystyle \|A\|_\infty = \max_i \sum_j |a_{ij}| $ & 
                \makecell[l]{Spectral/Operator \\norm} & $\displaystyle \|A\|_{\mathrm{op}} = \max_i \sigma_i(A) $ \\
    		\end{tabular}
    		\caption{Overview of a few different matrix norms, alternative definitions, and properties.}
    		\label{tab:norms}
    \end{table}

\subsection{Some inequalities}

    Some useful known relationships between the Schatten norms are 
    \begin{equation}
        \|A\|_{\mathrm{op}} 
        \leq \|A\|_{\mathrm{F}} 
        \leq \|A\|_{\mathrm{Tr}} 
        \leq \sqrt{\text{Rank}(A)} \|A\|_{\mathrm{F}} 
        \leq \sqrt{d} \|A\|_{\mathrm{F}}
        \leq d \|A\|_{\mathrm{op}} 
        \, .
    \end{equation}
    Also, given the alternative definition of the trace norm
    \begin{equation}
        \norm{A}_{\mathrm{Tr}} 
        = \sup_{\norm{B}_{\mathrm{op}}=1} |\Tr{}{AB}|
        = \sup_{B} \frac{|\Tr{}{AB}|}{\norm{B}_{\mathrm{op}}} \, ,
    \end{equation}
    it follows directly that
    \begin{equation}
        \norm{A}_{\mathrm{Tr}} \geq \frac{|\Tr{}{AB}|}{\norm{B}_{\mathrm{op}}} \quad \forall B
        \quad \Rightarrow \quad
        |\Tr{}{AB}| \leq \norm{B}_{\mathrm{op}} \norm{A}_{\mathrm{Tr}} \, .
    \end{equation}
    In particular, given an observable $O$ and two states $\rho$ and $\sigma$ we have
    \begin{equation}
        |\Tr{}{O(\rho - \sigma)} \leq \norm{O}_{\mathrm{op}} \norm{\rho - \sigma}_{\mathrm{Tr}} \, .
    \end{equation}

    We are also interested in relating the distance between rank-$1$ projectors with the distance of their corresponding generating vectors. 
    For this, we derive a few bounds relating the (operator) trace norm and the (vector) $2$-norm of states.

    \begin{proposition}[Bounds on the trace distance of two pure states]
    \label{result:bounds_trace_to_2_norm}
        Given two  state vectors $\ket{\psi}, \ket{\phi} \in \mathcal{H}$, 
        the trace distance between the two can be bounded or computed as 
        \begin{alignat}{3}
            (1) & \quad & 
            \Norm{\ketbra{\psi}{\psi} - \ketbra{\phi}{\phi}}_{\mathrm{Tr}} & 
            \leq 2 \Norm{\ket{\psi} - \ket{\phi}}_2 \, , \\
            (2) & \quad & 
            \Norm{\ketbra{\psi}{\psi} - \ketbra{\phi}{\phi}}_{\mathrm{Tr}} & 
            \leq \sqrt{2}\sqrt{\Norm{\ket{\psi}}_2^4+\Norm{\ket{\phi}}_2^4-2\Abs{\inner{\psi}{\phi}}^2} \, , \\
            (3) & \quad & 
            \Norm{\ketbra{\psi}{\psi} - \ketbra{\phi}{\phi}}_{\mathrm{Tr}} & 
            = \sqrt{\left(\Norm{\ket{\psi}}_2^2+\Norm{\ket{\phi}}_2^2\right)^2-4\Abs{\inner{\psi}{\phi}}^2} \, .
        \end{alignat}
    \end{proposition}

    \begin{proof}
        To obtain the first bound, we add and subtract the outer product between both states to obtain
        \begin{align}
             \Norm{\ketbra{\psi}{\psi} - \ketbra{\phi}{\phi}}_{\mathrm{Tr}} & 
            = \Norm{\ketbra{\psi}{\psi} - \ketbra{\phi}{\psi} + \ketbra{\phi}{\psi} - \ketbra{\phi}{\phi} }_{\mathrm{Tr}} \\ 
             \nonumber& 
            = \Norm{\left( \ket{\psi} - \ket{\phi} \right) \bra{\psi} + \ket{\phi} \left( \bra{\psi} - \bra{\phi} \right) }_{\mathrm{Tr}} \\
             \nonumber&
            \leq \Norm{\left( \ket{\psi} - \ket{\phi} \right) \bra{\psi} }_{\mathrm{Tr}} + \Norm{ \ket{\phi} \left( \bra{\psi} - \bra{\phi} \right) }_{\mathrm{Tr}} \\
            \nonumber
            &
            = \Norm{\ket{\psi} - \ket{\phi}}_2 \Norm{\ket{\psi}}_2 + \Norm{ \ket{\phi} }_2 \Norm{ \ket{\psi} - \ket{\phi} }_2 \\
             \nonumber
            &= \left[\Norm{\ket{\psi}}_2+\Norm{\ket{\phi}}_2\right]\Norm{\ket{\psi} - \ket{\phi}}_2
            \\
             \nonumber
            &
            \leq 2 \Norm{\ket{\psi} - \ket{\phi}}_2 \, .
             \nonumber
        \end{align}
        The second bound is obtained through a detour via the Frobenius norm,
        \begin{align}
            \Norm{\proj{\psi}-\proj{\phi}}_\text{Tr}&\leq\sqrt{2}\Norm{\proj{\psi}-\proj{\phi}}_\text{F}\\
             \nonumber
            &=\sqrt{2}\sqrt{\Tr{}{ \left(\proj{\psi}-\proj{\phi}\right)^2}}\\
             \nonumber
             &=\sqrt{2}\sqrt{\Norm{\ket{\psi}}_2^4+\Norm{\ket{\phi}}_2^4-2\Abs{\inner{\psi}{\phi}}^2},
        \end{align}
        Finally, to derive the identity, we can express $\ket{\phi}$ in terms of $\ket{\psi}$ and its orthogonal component $\ket{\psi^\perp}$ with equal norm, i.e.,
        \begin{equation}
            \ket{\phi} = \Norm{\ket{\phi}}\left(\cos\theta \frac{\ket{\psi}}{\Norm{\ket{\psi}}_2}+\sin\theta\frac{\ket{\psi^\perp}}{\Norm{\ket{\psi}}_2}\right)
        \end{equation}
        and therefore express $\proj{\psi}-\proj{\phi}$ with the basis vectors $\ket{\psi}/\Norm{\ket{\psi}}_2$ and $\ket{\psi^\perp}/\Norm{\ket{\psi}}_2$ as
        \begin{equation}
            \proj{\psi}-\proj{\phi}=\begin{pmatrix}
                \Norm{\ket{\psi}}_2^2-\Norm{\ket{\phi}}_2^2\cos^2\theta & -\Norm{\ket{\phi}}_2^2\cos\theta\sin\theta \\
                -\Norm{\ket{\phi}}_2^2\cos\theta\sin\theta &  -\Norm{\ket{\phi}}_2^2\sin^2\theta
            \end{pmatrix},
        \end{equation}
        whose eigenvalues are then given by
        \begin{equation}
            \lambda_{\pm}=\frac{1}{2}\left(\Norm{\ket{\psi}}_2^2-\Norm{\ket{\phi}}_2^2\pm\sqrt{\Norm{\ket{\psi}}_2^4+\Norm{\ket{\phi}}_2^4-2\Norm{\ket{\psi}}_2^2\Norm{\ket{\phi}}_2^2\cos2\theta}\right).
        \end{equation}
        Since one eigenvalue is positive and the other negative and using $\cos2\theta=\cos^2\theta-1$, we find
        \begin{align}
            \Norm{\proj{\psi}-\proj{\phi}}_{\mathrm{Tr}}^2=\left(\abs{\lambda_+}+\abs{\lambda_-}\right)^2 & = \Norm{\ket{\psi}}_2^4+\Norm{\ket{\phi}}_2^4-2\Norm{\ket{\psi}}_2^2\Norm{\ket{\phi}}_2^2\cos2\theta\\
            & = \Norm{\ket{\psi}}_2^4+\Norm{\ket{\phi}}_2^4+2\Norm{\ket{\psi}}_2^2\Norm{\ket{\phi}}_2^2-\underbrace{4\Norm{\ket{\psi}}_2^2\Norm{\ket{\phi}}_2^2\cos^2\theta}_{\Abs{\inner{\psi}{\phi}}^2}
        \end{align}
        and consequently
        \begin{equation}
            \Norm{\proj{\psi}-\proj{\phi}}_\text{Tr}=\sqrt{\left(\Norm{\ket{\psi}}_2^2+\Norm{\ket{\phi}}_2^2\right)^2-4\Abs{\inner{\psi}{\phi}}^2}.
        \end{equation}
    \end{proof}

\section{Ensemble decomposition matrix formulation}
\label{sec:matrix_formulation}

    Let us briefly introduce a matrix notation for the ensemble decompositions of states. 
    This will simplify the treatment of the unitary freedom, particularly in the derivation of the practical version of Wootters' method.

    Given the ensemble decomposition of a state $\rho$ like in Eq.~\eqref{eq:def_ensemble_decompositions}, we define an ``ensemble decomposition matrix''
    \begin{equation}
        \phi = \sum_{i=1}^r \sqrt{p_i} \ketbra{\phi_i}{i} \, . 
    \end{equation}
    In this formalism, the density matrix is simply 
    $\rho = \phi \phi^\dagger$ 
    and the unitary freedom corresponding to Eq.~\eqref{eq:unitarity_ensembles} becomes 
    \begin{equation}
    \label{eq:matrix_unitary_freedom}
        \varphi = \phi U^\dagger \, .
    \end{equation}
    This is easier to see starting from the right-hand side
    \begin{align}
        \phi U^\dagger & 
        = \left( \sum_{i'} \sqrt{p_{i'}} \ketbra{\phi_{i'}}{i'} \right) \left( \sum_{i,j} U^*_{j,i} \ketbra{i}{j} \right) \\ &  
        = \sum_{i,i',j} U^*_{j,i} \sqrt{p_{i'}} \inner{i'}{i} \ketbra{\phi_{i'}}{j} \\ & 
        = \sum_{i,j} \sqrt{p_i} U^*_{j,i} \ketbra{\phi_i}{j} 
        = \sum_{j} \left( \sum_{i} \sqrt{p_i} U^*_{j,i} \ket{\phi_i} \right) \bra{j} \\ & 
        = \sum_{j} \ketbra{\varphi_j}{j} 
        = \varphi \, .
    \end{align}

\section{Wootters' optimal decomposition}
\label{sec:wootters_optimal_decomp}

    In Ref.~\cite{hillEntanglementPairQuantum1997}, Hill and Wootters provide a closed-form formula for the entanglement of formation of two-qubit states.
    They show that the entanglement of formation of two-qubit states can be computed as
    \begin{equation}
        E_{\mathrm{oF}} (\rho) = h\left( \frac{1 + \sqrt{1-C^2}}{2} \right) \, ,
    \end{equation}
    where $h(x) = -x \log_2(x) - (1-x) \log_2(1-x)$ is the binary entropy, and $C$ is the concurrence of the state.
    This concurrence is defined as
    \begin{equation}
        C = \max \{0, \ \lambda_1 - \lambda_2 - \lambda_3 - \lambda_4\}
    \end{equation}
    where $\{ \lambda_i \}_i$ are the eigenvalues of the matrix $R = \sqrt{ \sqrt{ \rho } \tilde{\rho} \sqrt{ \rho } }$ in decreasing order, and the ``spin flip'' operation 
    \begin{equation}
    \tilde{\rho} = (\sigma_y \otimes \sigma_y) \rho^* (\sigma_y \otimes \sigma_y) \, .
    \end{equation}

    Later, in Ref.~\cite{woottersEntanglementFormationArbitrary1998}, Wootters introduces a method for computing an optimal decomposition achieving that value of ensemble-averaged entanglement entropy.
    The optimal decomposition is obtained through three transformations starting from the eigendecomposition of the density matrix. 
    These decompositions reach the entanglement of formation, i.e., result in a decomposition that minimizes ensemble-averaged entanglement entropy. 
    We will first recap these results following the same structure as Ref.~\cite{woottersEntanglementFormationArbitrary1998}, followed by a constructive (and numerically friendly) way of computing the decomposition. 

\subsection{Optimal decomposition for two-qubit states}
\label{sec:recap_Wootters}

    Wootters' method gives a decomposition for arbitrary two-qubit mixed states. 
    Thus, we always work with $\rho \in \operatorname{End}(\CC^{4 \times 4})$ from now on.
    They begin by relating the entanglement of formation to the concurrence. 
    To compute the latter, we need the so-called \emph{spin flip}, defined for pure and mixed states as
    \begin{align}
        \ket{\tilde{\psi}} & = (Y \otimes Y) \ket{\psi^*} \, , 
        \label{eq:spin-flip-vector} \\
        \tilde{\rho} & = (Y \otimes Y) \rho^* (Y \otimes Y) \, , 
        \nonumber
    \end{align}
    where the conjugation is done in the standard basis, and $Y$ is defined in the basis in which the conjugation is taken. 
    An alternative formulation of the same operation is to ``complex conjugate in the magic basis'', where the magic basis is defined as $ \{ \ket{\Phi^+}, i \ket{\Phi^-}, i \ket{\Psi^+}, \ket{\Psi^-}\}$.
    Then, the concurrence is defined as
    \begin{align}
        C(\ket{\psi}) & = |\inner{\psi}{\tilde{\psi}}| \, , \\
        C(\rho) & = \max \left\{0, \lambda^{(R)}_1 - \lambda^{(R)}_2 - \lambda^{(R)}_3 - \lambda^{(R)}_4 \right\} 
    \end{align}
    where $\lambda^{(R)}_i$ are the eigenvalues of $R = \sqrt{\sqrt{\rho}\tilde{\rho}\sqrt{\rho}}$, sorted in decreasing order.
    From that, we can compute the entanglement of formation as 
    \begin{equation}
        E_{oF}(\rho) = h\left( \frac{1}{2} \left( 1 + \sqrt{1 - C(\rho)^2} \right) \right)
    \end{equation}
    with the binary entropy $h(x) = x \log_2 x + (1-x) \log_2 (1-x)$.
    The state is entangled if the concurrence (and thus the entanglement of formation) is positive, and is separable if it is zero.

    The task solved in Ref.~\cite{woottersEntanglementFormationArbitrary1998} is finding an ensemble of (sub-normalized) states $\{ \ket{w_i} \}$ such that 
    \begin{equation}
        \label{eq:wootters_target_decomp}
        \sum_i \ketbra{w_i}{w_i} = \rho
        \quad \text{and} \quad
        E \left( \frac{\ket{w_i}}{\| \ket{w_i} \|_2} \right) = E_{oF}(\rho) \, \forall i \, .
    \end{equation}
    The procedure is different for entangled and separable states. 
    We first describe the entangled case ($C(\rho) > 0$) and then the separable one ($C(\rho) = 0$).
    The procedure for entangled states requires four steps and three decompositions, building upon one another.

\subsubsection*{Step 0: Eigendecomposition and arbitrary decompositions.}

    The starting point is the eigendecomposition of the state 
    \begin{equation}
        \rho = \sum_{i=1}^4 \lambda_i^{(\rho)} \ketbra{\lambda_i^{(\rho)}}{\lambda_i^{(\rho)}} \, , 
        \quad 
        \lambda_i^{(\rho)} \geq 0 \, \forall i \, , 
        \quad 
        \inner{\lambda_i^{(\rho)}}{\lambda_j^{(\rho)}} = \delta_{i,j} \, .
    \end{equation}
    It will turn out to be more convenient to work with sub-normalized vectors. The sub-normalized eigenstates are relabelled as $\ket{v_i} = \sqrt{ \lambda_i^{(\rho)}} \ket{\lambda_i^{(\rho)}} $, resulting in $\rho = \sum_{i=1}^{4} \ketbra{v_i}{v_i}$.
    From this, any ensemble decomposition can be obtained via 
    \begin{equation}
    \label{eq:subnormalized_ensembles}
        \rho = \sum_{i=1}^{m} \ketbra{w_i}{w_i}
        \quad \text{with} \quad 
        \ket{w_i} = \sum_{j=1}^4 U^*_{i,j} \ket{v_j} 
        \quad \text{for } i \in \{1, \dots, m\} \, ,
    \end{equation}
    where the conjugation of the unitary was added for later convenience.
    This is equivalent to Eq.~\eqref{eq:unitarity_ensembles} but includes the probabilities in the sub-normalized states.

\subsubsection*{Step 1: The spin-flipped-orthogonal decomposition}

    The first decomposition $\{ \ket{x_i} \}$ is constructed such that 
    \begin{equation}
    \label{eq:x_condition}
        \inner{x_i}{\tilde{x}_j} = \delta_{i,j} \lambda_i^{(R)} \, .
    \end{equation}
    This can be done by computing the matrix $\tau$ where $\tau_{i,j} = \inner{v_i}{\tilde{v}_j}$ and finding a unitary $U$ such that $U \tau U^\top = D$ (note the transpose, not Hermitian conjugate).
    Here $D$ is a diagonal matrix with the eigenvalues of $R$ on the diagonal, i.e., $D=\mathrm{diag}(\{\lambda_i^{(R)}\})$.
    Then, constructing 
    \begin{equation}
        \ket{x_i} = \sum_{j=1}^4 U^*_{i,j} \ket{v_j}
    \end{equation}
    fulfills the condition in Eq.~\eqref{eq:x_condition}.
    
\subsubsection*{Step 2: The complexified decomposition}

    Wootters' second decomposition consists of adding a phase to all elements except the first, i.e.,
    \begin{align}
        \ket{y_1} & = \ket{x_1} \\
        \ket{y_i} & = i\ket{x_i} \quad \text{for } i=2, 3, 4 \, .
    \end{align}
    
    This results in an added minus sign when computing $\inner{y_i}{\tilde{y}_j} = - \delta_{i,j} \lambda_i^{(R)}$ for $i=2, 3, 4$, while maintaining the orthogonality of the four elements, ensuring
    \begin{equation}
        \sum_{i=1}^4 \inner{y_i}{\tilde{y}_i} 
        = \lambda^{(R)}_1 - \lambda^{(R)}_2 - \lambda^{(R)}_3 - \lambda^{(R)}_4 \, ,
    \end{equation}
    which is the concurrence $C(\rho)$ of the initial state.
    
\subsubsection*{Step 3: The preconcurrence-averaged decomposition}

    The final step in Wootters' method is to find a decomposition $\{ \ket{z_i} \}$ for which 
    \begin{equation}
        \frac{\inner{z_i}{\tilde{z}_i}}{\inner{z_i}{z_i}} = C(\rho) \, .
    \end{equation}
    This can be achieved using an orthogonal (real unitary) transformation. 
    To find such a transformation, they propose an iterative procedure of mixing two elements at a time, one with a large preconcurrence and one with a small one, until all elements have the correct preconcurrence.

    Once reaching this final condition, one has solved the problem as described in Eq.~\eqref{eq:wootters_target_decomp} since 
    \begin{equation}
        E \left( \frac{\ket{z_i}}{\|\ket{z_i}\|_2} \right) = h\left( \frac{1}{2} \left( 1 + \sqrt{1 - C(\rho)^2} \right) \right) = E_{oF}(\rho) \, .
    \end{equation}
    This value of the entanglement entropy of each element guarantees that the average is minimal
    \begin{equation}
        E_{av}\left( \left\{  \frac{\ket{z_i}}{\|\ket{z_i}\|_2} \right\} \right)
        = \sum_i \|\ket{z_i}\|_2 ^2 E \left( \frac{\ket{z_i}}{\|\ket{z_i}\|_2} \right)
        = E_{oF}(\rho) \sum_i p_i
        = E_{oF}(\rho) \, ,
    \end{equation}
    concluding that the decomposition $ \left\{ \ket{z_i} \right\}$ is optimal.

\subsubsection*{The separable case}

    The above construction only holds for the non-separable case.
    For the separable case, where $\lambda_1^{(R)} - \lambda_2^{(R)}  - \lambda_3^{(R)}  - \lambda_4^{(R)} \leq 0$, Wootters provides a different decomposition. 
    Starting from the $\{ \ket{x_i} \}$ decomposition of Step 1, constructing
    \begin{align}
        \ket{z_1} & = \left( \ee^{\ii\theta_1} \ket{x_1} + \ee^{\ii\theta_2} \ket{x_2} + \ee^{\ii\theta_3} \ket{x_3} + \ee^{\ii\theta_4} \ket{x_4} \right) \, , \\
        \ket{z_2} & = \left( \ee^{\ii\theta_1} \ket{x_1} + \ee^{\ii\theta_2} \ket{x_2} - \ee^{\ii\theta_3} \ket{x_3} - \ee^{\ii\theta_4} \ket{x_4} \right) \, , \\
        \ket{z_3} & = \left( \ee^{\ii\theta_1} \ket{x_1} - \ee^{\ii\theta_2} \ket{x_2} + \ee^{\ii\theta_3} \ket{x_3} - \ee^{\ii\theta_4} \ket{x_4} \right) \, , \\
        \ket{z_4} & = \left( \ee^{\ii\theta_1} \ket{x_1} - \ee^{\ii\theta_2} \ket{x_2} - \ee^{\ii\theta_3} \ket{x_3} + \ee^{\ii\theta_4} \ket{x_4} \right) \, ,
    \end{align}
    where the phase factors $\theta_i$ are chosen such that 
    \begin{equation}
    \label{eq:separating_angles_condition}
        \sum_{j=1}^4\ee^{2\ii\theta_j}\lambda_j=0 \, ,
    \end{equation}
    results in a decomposition with zero average entanglement, thus being optimal.
    These results show that one can always find an optimal decomposition for mixed two-qubit states. 
    Next, we will show how to construct the respective decompositions and explicitly obtain the required unitaries.

\subsection{Practical/closed-form version}
    
    In the following, we will present an explicit blueprint for using Wootters' method practically in a pseudo-algorithmic fashion.
    We will use the matrix notation presented in Section~\ref{sec:matrix_formulation}, and show that the steps of Wootters' method reduce to known matrix decompositions on such objects.

    As described above in Eq.~\eqref{eq:subnormalized_ensembles}, any decomposition of a density matrix can be written as an ensemble of sub-normalized  state vectors $\{ \ket{w_i} \}_i$ with the corresponding decomposition matrix
    \begin{equation}
        w = \sum_{i=1}^m \ketbra{w_i}{i} 
        \quad \text{e.g. the eigendecomposition} \quad 
        v = \sum_{i=1}^4 \underbrace{\sqrt{\lambda_i^{(\rho)}} \ket{\lambda_i^{(\rho)}}}_{\ket{v_i}}  \!\! \bra{i} 
        \, . 
    \end{equation}
    Recall that in this formalism, the density matrix is simply 
    $\rho = v v^\dagger = w w^\dagger$ 
    and the unitary freedom corresponding to Eq.~\eqref{eq:subnormalized_ensembles} becomes 
    \begin{equation}
        w = v U^\dagger \, .
    \end{equation}
    We also extend the spin-flip operation on these objects.
    First, propagating the spin-flip through the unitary freedom in Eq.~\eqref{eq:subnormalized_ensembles} leads to
    \begin{align}
        \ket{\tilde{w}_i} 
        = (Y \otimes Y) \ket{w^*} 
        = (Y \otimes Y) \left( \sum_{j} U^*_{i,j} \ket{v_j} \right)^*
        = (Y \otimes Y) \sum_{j} U_{i,j} \ket{v^*_j}
        = \sum_{j} U_{i,j} (Y \otimes Y) \ket{v^*_j}
        = \sum_{j} U_{i,j}\ket{\tilde{v}_j} \, .
    \end{align}
    Comparing with the decomposition matrix formulation of the unitary freedom in Eq.~\eqref{eq:matrix_unitary_freedom} reveals that
    \begin{equation}
        \tilde{w} = \sum_i \ketbra{\tilde{w}_i}{i}
        = \tilde{v} U^\top \, .
    \end{equation}
    
    We will now construct the three subsequent decompositions following Wootters' method, resulting in an entanglement-minimizing final decomposition.

\subsubsection*{Step 1: The spin-flipped-orthogonal decomposition}

    The condition that the first decomposition has to fulfill, Eq.~\eqref{eq:x_condition}, can on the one hand be rewritten in terms of the corresponding decomposition matrix as
    \begin{equation}
        x^\dagger \tilde{x} 
        = \sum_{i,j} \inner{x_i}{\tilde{x}_j} \ketbra{i}{j}
        = \operatorname{diag} \left( \left\{ \lambda_i^{(R)} \right\} \right)
        \quad \text{for} \quad 
        x = v U_1^\dagger \, .
    \end{equation}
    On the other hand, the matrix $\tau$ is given by 
    \begin{equation}
        \tau = v^\dagger \tilde{v} = \sum_{i,j} \inner{v_i}{\tilde{v}_j} \ketbra{i}{j} \, .
    \end{equation}
    The diagonalization of $\tau$
    \begin{equation}
    \label{eq:tau_diag}
        U_1^{\vphantom{\top}} \tau U_1^\top = \operatorname{diag} \left( \left\{ \lambda_i^{(R)} \right\} \right)
    \end{equation}
    which returns the right unitary freedom $U_1$ 
    is precisely what the factorization by Takagi-Autonne decomposition achieves~\cite[Corollary 4.4.4(c)]{hornMatrixAnalysis2012}.
    The Takagi-Autonne factorization can be applied to any square, complex, symmetric matrix $M$ and can be easily computed from the singular value decomposition of 
    $M = U^{\vphantom{\dagger}}_M \Sigma^{\vphantom{\dagger}}_M V_M^\dagger$ 
    via $U_{\mathrm{Tak}} = U_M^{\vphantom{\top}} \sqrt{(U_M^\top V_M^{\vphantom{\top}})^*}$~\cite{houdeMatrixDecompositionsQuantum2024}~\footnote{
        For some special cases, there are also more numerically stable methods. 
        For instance, for real matrices with negative eigenvalues, 
        $M = U^{\vphantom{\top}}_M \Lambda^{\vphantom{\top}}_M U_M^\top$ 
        and $\lambda_i^{(M)} \leq 0$ for $i \geq i^-$, choosing 
        $U_{\mathrm{Tak}} = U_M \operatorname{diag}(\overbrace{1, \dots , 1}^{< i^-}, \overbrace{\ii, \dots, \ii}^{\geq i^-})$ 
        and $\Sigma_M = | \Lambda_M |$
        gives a valid Takagi-Autonne decomposition.
    }, 
    which will yield $M = U^{\vphantom{\top}}_{\mathrm{Tak}} \Sigma^{\vphantom{\top}}_M U_{\mathrm{Tak}}^\top$.
    With this decomposition, we 
    can now construct
    \begin{equation}
        x = v U_1^\dagger
    \end{equation}
    with $U_1 = U_{\mathrm{Tak}}^\dagger(\tau)$ (modulo some changes in indices to ensure that the $\lambda_i^{(R)}$ are ordered in decreasing order).
    To verify this, we can compute 
    \begin{align}
        x^\dagger \tilde{x} 
        = \left(v U_{\mathrm{Tak}} \right)^\dagger \widetilde{\left(v U_{\mathrm{Tak}} \right)} 
        = U_{\mathrm{Tak}}^\dagger v^\dagger \tilde{v} U_{\mathrm{Tak}}^* 
        = U_{\mathrm{Tak}}^\dagger \tau U_{\mathrm{Tak}}^*
        = \Sigma_\tau \, .
    \end{align}
    Let us now check that the values obtained in the diagonal matrix are the eingenvalues of $R$. 
    For this, we consider the matrix $R^2 = \sqrt{\rho}\tilde{\rho}\sqrt{\rho}$, whose eigenvalues are the eigenvalues of $R$ squared. 
    Then, from the decomposition $\rho = v v^\dagger$, we will make the eigenvalues explicit by normalizing the vectors in the decomposition matrix $v$, resulting in the typical eigendecomposition 
    \begin{equation}
        \rho = v v^\dagger
         = v_N^{\vphantom{\dagger}} \Lambda_\rho v_N^\dagger
    \end{equation}
    where $\Lambda_\rho$ is the diagonal matrix containing the eigenvalues of $\rho$ and $v_N = \sum_{i=1}^4 \ketbra{\lambda_i^{(\rho)}}{i}$. 
    Note that in the construction of $v_N$, each eigenstate is normalized (and orthogonal), so $v_N$ is a unitary matrix.
    Also, $v = v_N \sqrt{\Lambda_\rho}$. 
    With this, we can write
    \begin{equation}
        \sqrt{\rho} 
        = \lefteqn{\overbrace{\phantom{v_N^{\vphantom{\dagger}} \sqrt{\Lambda_\rho}}}^{v}}
        v_N^{\vphantom{\dagger}} \underbrace{\sqrt{\Lambda_\rho} v_N^\dagger}_{v^\dagger} 
        = v^{\vphantom{\dagger}}_N v^\dagger 
        = v v_N^\dagger
    \end{equation}
    Then, inputting this into $R^2$ gives
    \begin{equation}
        R^2 
        = 
        \lefteqn{
            \overbrace{
                \phantom{v_N v^{\vphantom{\dagger}}}
            }^{\sqrt{\rho}}
        } 
        v_N^{\vphantom{\dagger}}
        \underbrace{
            v^\dagger
            \lefteqn{
                \overbrace{
                    \phantom{\tilde{v} \tilde{v}^{{\dagger}}}
                }^{\tilde{\rho}}
            }
            \tilde{v}
        }_{\tau} 
        \underbrace{
            \tilde{v}^\dagger
            \lefteqn{
                \overbrace{
                    \phantom{v v^\dagger}
                }^{\sqrt{\rho}}
            }
            v
        }_{\tau^\dagger} 
        v_N^\dagger
        = v_N^{\vphantom{\dagger}} \tau \tau^\dagger v_N^\dagger
    \end{equation}
    Since $v_N$ is a unitary, $R^2$ shares the same eigenvalues as $\tau \tau^\dagger$. 
    Using the Takagi-Autonne decomposition above gives
    \begin{equation}
        \tau \tau^\dagger
        = U^{\vphantom{\top}}_{\mathrm{Tak}} \Sigma^{\vphantom{\top}}_{\tau} \underbrace{U_{\mathrm{Tak}}^\top U^{*\vphantom{\top}}_{\mathrm{Tak}}}_{\II} \Sigma^{\vphantom{\top}}_{\tau} U_{\mathrm{Tak}}^\dagger
        = U^{\vphantom{\top}}_{\mathrm{Tak}} \Sigma^{2\vphantom{\top}}_{\tau} U_{\mathrm{Tak}}^\dagger \, .
    \end{equation}
    With this, the diagonal matrix $\Sigma^{2}_{\tau}$ is also the matrix of eigenvalues of $R^2$, showing that Eq.~\eqref{eq:tau_diag} holds.
    
    This decomposition then also fulfills 
    $U_1 \tau \tau^* U_1^\dagger = \mathrm{diag}(\{(\lambda_i^{(R)})^2\})$, 
    although note that finding the eigendecomposition of $\tau \tau^*$ does not necessarily fulfill Eq.~\eqref{eq:tau_diag}
    \footnote{
        Take, for instance, the eigendecomposition $M = UDU^{-1}$. 
        Then replacing the unitary $U$ by any choice of the same with an added phase $\ee^{\ii\phi} U$ results in a valid eigendecomposition.
        However, in the Takagi decomposition $M = U \Sigma U^\top$, the added phase results in 
        $(\ee^{\ii\phi} U) \Sigma (\ee^{\ii\phi} U)^\top = \ee^{\ii 2 \phi} M \neq M$.
    }. 
    
\subsubsection*{Step 2: The complexified decomposition}

    The second decomposition can be rewritten in this formulation as
    \begin{equation}
        y = x U_2^\dagger
    \end{equation}
    with $U_2 = \operatorname{diag}(\{1, i, i, i\})$.
    Then the sum of flipped overlaps becomes 
    \begin{equation}
        \sum_{i=1}^4 \inner{y_i}{\tilde{y}_i} 
        = \Tr{}{y^\dagger \tilde{y}} 
        = \lambda_1^{(R)} - \lambda_2^{(R)}  - \lambda_3^{(R)}  - \lambda_4^{(R)} 
        = C(\rho) \, .
    \end{equation}

\subsubsection*{Step 3: The preconcurrence-averaged decomposition}

    In terms of the decomposition matrices, the condition for the final decomposition in Wootters' method is to find a unitary $U_3$ to construct $z = y U_3^\dagger$ such that 
    \begin{equation}
        \Tr{}{z^\dagger \tilde{z}}
        = \Tr{}{y^\dagger \tilde{y}}
        \quad \text{and} \quad
        \frac{\inner{z_i}{\tilde{z}_i}}{\inner{z_i}{z_i}} 
        = \frac{[z^\dagger \tilde{z}]_{ii}}{[z^\dagger z]_{ii}} 
        = C(\rho) \, .
    \end{equation}
    The first condition ensures that the average preconcurrence of the decomposition is preserved and thus stays minimal.
    Indeed, the average preconcurrence of the $y$ decomposition is the concurrence of the state $\rho$, and there can be no decomposition with lower concurrence.
    This condition can be rewritten by applying the unitary freedom as 
    \begin{equation}
        \Tr{}{U_3 y^\dagger \tilde{y} U_3^\top} 
        = \Tr{}{y^\dagger \tilde{y}}
    \end{equation}
    which is fulfilled by choosing the unitary $U_3$ to have real coefficients only (and thus to be an orthogonal matrix).
    The second condition is that the preconcurrence of each element of the decomposition is equal to the concurrence of the state, making the decomposition optimal.
    It is equivalent with 
    \begin{align}
        && \inner{z_i}{\tilde{z}_i} &
        = C(\rho) \inner{z_i}{z_i} && \\ &
        \Leftrightarrow & 
        \bra{i} z^\dagger \tilde{z} \ket{i} & 
        = C(\rho) \bra{i} z^\dagger z \ket{i} && \\ &
        \Leftrightarrow & 
        \bra{i} U_3 y^\dagger \tilde{y} U_3^\top \ket{i} & 
        = C(\rho) \bra{i} U_3 y^\dagger y U_3^\dagger \ket{i} && 
        \text{where } U_3^\dagger = U_3^\top \\ &
        \Leftrightarrow & 
        \bra{i} U_3 M U_3^\top \ket{i} & 
        = 0 && 
        \text{for } M = y^\dagger \tilde{y} - C(\rho) y^\dagger y
        \, .
    \end{align}
    The matrix $M$ is Hermitian as it is the sum of two hermitian matrices since
    \begin{align}
        y^\dagger \tilde{y} & 
        = \operatorname{diag}\left( \left\{ \lambda_1^{(R)}, -\lambda_2^{(R)}, -\lambda_3^{(R)}, -\lambda_4^{(R)} \right\} \right) && 
        \text{is diagonal and real,} \\
        (y^\dagger y)^\dagger & = y^\dagger y && 
        \text{is Hermitian by construction.} 
    \nonumber
    \end{align}
    Also, it has trace $0$ because
    \begin{align}
        \Tr{}{M} & 
        = \Tr{}{y^\dagger \tilde{y} - C(\rho) y^\dagger y}
        = \Tr{}{y^\dagger \tilde{y}} - C(\rho) \Tr{}{y^\dagger y} \\ &
        = C(\rho) - C(\rho) \Tr{}{\rho}
        = 0 \, .
    \end{align}
    We can find an orthogonal matrix $U_3$ fulfilling the two conditions above.

    If we decompose $M$ into its real and imaginary parts, we can write
    \begin{equation}
        M = M_R + i M_I
        \quad \text{with} \quad
        M_R, M_I \in \CC^{4 \times 4}, 
        \quad
        M_R = M_R^\top 
        \quad \text{and} \quad
        M_I = - M_I^\top \, .
    \end{equation}
    Since the imaginary part is antisymmetric, it has only zeros on the diagonal. Thus, the trace of $M$ is carried only by its real part.
    This implies that the real part also has a null trace, $\Tr{}{M_R} = \Tr{}{M} = 0$. 
    Diagonalizing the real part of $M$ gives 
    \begin{equation}
        M_R = Q_R D_R Q_R^\top \, .
    \end{equation}
    Since $M_R$ is symmetric, the diagonalizing matrix $Q_R$ is orthogonal and we have $D_R = Q_R^\top M_R Q_R$, where the diagonal matrix still fulfills $\Tr{}{D_R} = \Tr{}{M_R} = 0$.
    Next, we have to solve the simpler problem of finding an orthogonal transformation that takes a diagonal matrix with zero trace to some matrix where all diagonal elements are zero.
    This can be done by conjugating the diagonal matrix with the $4 \times 4$ Hadamard matrix,
    \begin{equation}A
        H_4 = \frac{1}{2} \begin{pmatrix}
            1 & 1 & 1 & 1 \\
            1 & 1 & -1 & -1 \\
            1 & -1 & 1 & -1 \\
            1 & -1 & -1 & 1 
        \end{pmatrix}.
    \end{equation}
    The diagonal elements of the resulting matrix are then given by $\frac{1}{4} \sum_i [D_R]_{ii} = \frac{1}{4} \Tr{}{D_R} = 0$.
    The product $H_4 Q_R^\top$ thus transforms the real part of $M$ to a matrix with null diagonal elements.
    Let us show that it fulfills the conditions of the third change of basis in Wootters' procedure.

    First, it is a real unitary matrix because it is the product of two real unitary matrices.
    We then use the fact that the asymmetry of matrices is preserved by conjugation with an orthogonal matrix. 
    Consequently, the diagonal of $H_4 Q_R^\top M_I Q_R H_4^\top$ must still be populated by zeros. 
    This results in a matrix in which the diagonal elements are the diagonal elements of $H_4 Q_R^\top M_R Q_R H_4^\top$, which, as shown above, are all zero, i.e.,
    \begin{equation}
        [H_4 Q_R^\top (M_R + i M_I) Q_R H_4^\top]_{ii} = 0 \quad \forall i \, .
    \end{equation}
    With this, we have fulfilled the second condition for Wootters' third step and have $U_3 = H_4 Q_R^\top$.
    \footnote{
        If the state has rank $< 4$ (that is $2$ or $3$, since the rank $1$ case is a pure state for which the entanglement entropy can be directly computed), one can pad the decomposition matrix with zeros, thus effectively making it a decomposition with $4$ elements, although some have probabilities $0$. 
        It is equivalent to truncating the Hadamard matrix to a $4 \times r$ isometry.
        In the case of rank $2$, one can also use the $2 \times 2$ Hadamard instead, thus preserving the cardinality of the decomposition. 
        For rank $3$, however, we are unaware of $3 \times 3$ Hadamard matrices.
    }
    
    Consequently, starting with the sub-normalized states $\{\ket{v_i}\}$, we obtain an explicit, entanglement-minimizing decomposition via
    \begin{equation}
        z = y U_3^\dagger = v U_1^\dagger U_2^\dagger U_3^\dagger.
    \end{equation}

\subsubsection*{The separable case}

    \begin{figure}
        \centering
        \includegraphics[scale=1]{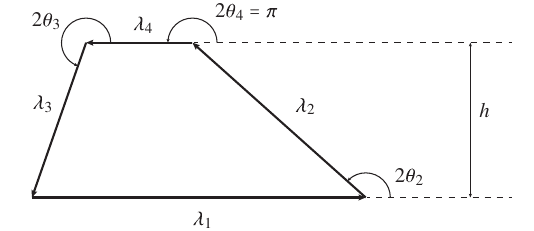}
        \caption{Geometric interpretation of the condition in Eq.~\eqref{eq:separating_angles_condition}.}
        \label{fig:trapezoid-angles}
    \end{figure}
    In the separable case, the condition Eq.~\eqref{eq:separating_angles_condition} must be satisfied to find an optimal decomposition.
    It can be solved by constructing a trapezoid in the complex plane for vectors $\lambda_j^{(R)} \ee^{2\ii\theta_j}$.
    To this end, we start by choosing
    \begin{equation}
        \theta_1=0, \qquad \theta_4=\pi/2
    \end{equation}
    to construct the base and opposite parallel of the trapezoid out of the longest and shortest sides. 
    Note that this construction implicitly contains the case of $\lambda_4 = 0$ since the resulting shape will be a triangle with identical equations for its height and the remaining angles.
    The case of $\lambda_1 = \lambda_2 = \lambda_3 = \lambda_4$ can be treated independently, as it leads to a square in the complex plane and thus choosing $\theta_1=0$, $\theta_2=\pi/4$, $\theta_3=\pi/2$, $\theta_4=3\pi/4$ fulfills Eq.~\eqref{eq:separating_angles_condition}.
    In all other cases, the height of the resulting shape is given by
    \begin{align}
        h & =\frac{2}{\lambda_1-\lambda_4}\sqrt{s(s-\lambda_1+\lambda_4)(s-\lambda_2)(s-\lambda_3)} \\
        s & =\frac{\lambda_1-\lambda_4+\lambda_2+\lambda_3}{2} \, ,
    \end{align}
    resulting in
    \begin{equation}
        2\theta_2 = \pi-\arcsin\left(\frac{h}{\lambda_2}\right) \, .
    \end{equation}
    Now, it might be that the final angle $\theta_3$ of the trapezoid is obtuse, resulting in an erroneous numerical outcome from the $\arcsin$.
    To this end, we first check whether $\cos(2\theta_2)\lambda_2+\lambda_4-\lambda_1>0$.
    If 
    so, we choose
    \begin{equation}
        2\theta_3 = \pi+\arcsin\left(\frac{h}{\lambda_3}\right)
    \end{equation}
    and otherwise
    \begin{equation}
        2\theta_3 = -\arcsin\left(\frac{h}{\lambda_3}\right)
    \end{equation}
    to obtain an explicit, separable decomposition.

\section{Some optimal decompositions}

\subsection{Optimal Kraus decompositions of some common noise channels for their Choi state}
\label{sec:choi_optimal_kraus_examples}

    In this section, we will obtain the optimal Kraus decomposition of some common noise models with respect to the Bell state.
    This is also the Kraus decomposition, which induces an optimal state decomposition of the Choi state of these channels. 
    That is, the resulting decomposition following Eq.~\eqref{eq:Kraus_to_ensemble} minimizes the ensemble-averaged entropy of the Choi state.
    
    We consider noise channels $\mathcal{N}$ described by a valid set of Kraus operators $\{ K_i \}_i$ as in Eq.~\eqref{eq:def_kraus}
    and the unitary freedom $U$ on that representation as in Eq.~\eqref{eq:unitarity_kraus}.
    The Choi state of such a channel is then
    \begin{equation}
        {\choi}(\mathcal{N}) 
        = (\mathcal{N} \otimes \mathcal{I}) [\ketbra{\omega}{\omega}]
        = \sum_{i=1} \left( K^{(U)}_{i} \otimes \II \right) \ketbra{\omega}{\omega} \left( K^{(U)\dagger}_{i} \otimes \II \right) , 
        \quad
        \ket{\omega} = \frac{1}{\sqrt{d}} \sum_{j=0}^{d-1} \ket{jj}
        = \frac{1}{\sqrt{d}} \  \tikzineq{
            \draw [rounded corners, line width=1pt] (.4, .4) -- (0,.4) -- (0,0) -- (.4,0);
        },
    \end{equation}
    where $\{ p_i, \frac{1}{\sqrt{pi}} K_i^{(U)} \ket{\omega} \}_i$ is the induced ensemble decomposition of ${\choi}(\mathcal{N})$.
    We say that the Kraus decomposition is optimal if the ensemble-averaged entanglement entropy of $\{ p_i, \frac{1}{\sqrt{pi}} K_i^{(U)} \ket{\omega} \}_i$ is the entanglement of formation of ${\choi}(\mathcal{N})$ and is thus minimal.
    
    \begin{proposition}
        For the single-qubit dephasing channel, the rotation 
        \begin{equation}
            U = \frac{1}{\sqrt{2}}\begin{pmatrix}
               1 & 1 \\
               1 & -1
            \end{pmatrix},
        \end{equation}
        on the Kraus decomposition 
        \begin{equation}
            \{ K_i \}_{i=1}^2 = \left \{\sqrt{1-\frac{p}{2}} \II, \sqrt{\frac{p}{2}} Z \right \}
        \end{equation}
        is optimal for the Choi state.
    \end{proposition}

    \begin{proof}
        Choi states of single-qubit noise channels are two-qubit states, so the results from Ref.~\cite{woottersEntanglementFormationArbitrary1998} presented in Section~\ref{sec:recap_Wootters} apply.
        Let us now look at the noise model of interest.
        The Choi state (based on the {mixed-unitary} decomposition) is 
        \begin{equation}
            {\choi}(\mathcal{N}) = (1-p/2) \ketbra{\Phi^+}{\Phi^+} + p/2 \ketbra{\Phi^-}{\Phi^-} \, .
        \end{equation}
        with $ \ket{\tilde{\Phi}^+} = -\ket{\Phi^+}$ and $ \ket{\tilde{\Phi}^-} = \ket{\Phi^-}$, so the density matrices of both elements of the decomposition are invariant under the spin-flip.
        As a result, $R = \rho$. 
        The eigenvalues are then $\{ \lambda_i \}_i = \{ 1-p/2, p/2, 0, 0 \}$
        and $C({\choi}(\mathcal{N})) = 1-p$. 
        Consequently, the entanglement of formation is 
        \begin{equation}
        E_{\mathrm{oF}} = h\left( \frac{1+\sqrt{(2-p)p}}{2} \right).
        \end{equation}
        
        Let us find out if our candidate for an optimal decomposition reaches this value of average entanglement, meaning that it would be optimal.
        We have the set of unitary Kraus operators $\{ K_i \}_i = \{\sqrt{1-p/2} \II, \sqrt{p/2}Z\}$
        and the candidate unitary $U = \frac{1}{\sqrt{2}} \left( \begin{smallmatrix} 1 & 1 \\ 1 & -1 \end{smallmatrix} \right)$.
        Then the normalized ensemble states are 
        \begin{equation}
            \ket{\phi_i} = \sqrt{2} K_i^{(U)} \ket{\Phi^+} = \sqrt{1-p/2}\ket{\Phi^+} \pm \sqrt{p/2}\ket{\Phi^-} \, .
        \end{equation}
        These states both have the same probability $p=\frac{1}{2}$, and Schmidt coefficients $\frac{1}{2}(\sqrt{2-p} \pm \sqrt{p})$. 
        From this, we have that $E_{av} = E(\ket{\phi_1}) = E(\ket{\phi_2})$.
        Since there are only two Schmidt coefficients and their square sums to one, we have
        \begin{equation}
            E_{av} = h\left(\left(\frac{1}{2}(\sqrt{2-p} + \sqrt{p}\right)^2\right)
            = h\left(\frac{1+\sqrt{(2-p)p}}{2}\right)
            = E_{\mathrm{oF}}
        \end{equation}
        which is the result from above.
    \end{proof}

    \begin{proposition}
        For the single-qubit amplitude-damping channel, the Kraus decomposition 
        \begin{equation}
            K^{(U)}_{i'} = \sum_{i=1}^{2} U_{i'i} K_{i} \quad \forall i' \in \{ 1, 2 \}
            \quad \text{with} \quad
            \left\{ K_i\right\}_{i=1}^2 = \left \{ \left( \begin{smallmatrix} 1 & 0 \\ 0 & \sqrt{1-\gamma} \end{smallmatrix} \right), \left( \begin{smallmatrix} 0 & \sqrt{\gamma} \\ 0 & 0 \end{smallmatrix} \right) \right \}
            \quad \text{and} \quad
            U = \frac{1}{\sqrt{2}}
            \left( \begin{smallmatrix} 
               1 & 1 \\
               1 & -1
            \end{smallmatrix} \right)
        \end{equation}
        is optimal for the Choi state.
    \end{proposition}
    
    \begin{proof}
        The proof follows a similar argument to that of the dephasing channel.
        In the case of amplitude damping, there are a few more steps, as the effect of the spin flip is not trivial, but the idea is the same.
        From the {Orthogonal} Kraus decomposition 
        $\left \{ \left( \begin{smallmatrix} 1 & 0 \\ 0 & \sqrt{1-\gamma} \end{smallmatrix} \right), \left( \begin{smallmatrix} 0 & \sqrt{\gamma} \\ 0 & 0 \end{smallmatrix} \right) \right \}$
        we get the ensemble decomposition of the Choi state
        \begin{equation}
            \begin{cases}
                \ket{\phi_1} = \frac{1}{\sqrt{p_1}} K_1 \ket{\Phi^+} = \frac{1}{\sqrt{2-\gamma}}(\ket{0,0} + \sqrt{1-\gamma} \ket{1,1}), & 
                p_1 = \frac{2-\gamma}{2} \\
                \ket{\phi_2} = \frac{1}{\sqrt{p_2}} K_2 \ket{\Phi^+} = \ket{01}, & 
                p_2 = \frac{\gamma}{2} \, .
            \end{cases}
        \end{equation}
        Then,
        \begin{equation}
            \widetilde{{\choi}(\mathcal{N})} = \frac{2-\gamma}{2} \ketbra{\tilde{\phi_1}}{\tilde{\phi_1}} + \frac{\gamma}{2} \ketbra{\tilde{\phi_2}}{\tilde{\phi_2}}
            \quad \text{with} \quad 
            \begin{cases}
                \ket{\tilde{\phi_1}} = \frac{-1}{\sqrt{2-\gamma}}(\sqrt{1-\gamma} \ket{0,0} + \ket{1,1}) \\
                \ket{\tilde{\phi_2}} = \ket{1,0}
            \end{cases} \, .
        \end{equation}
        Using that $\inner{\phi_1}{\tilde{\phi}_1} = -\frac{2\sqrt{1-\gamma}}{2-\gamma}$ while $\inner{\phi_1}{\tilde{\phi}_2} = \inner{\phi_2}{\tilde{\phi}_1} = \inner{\phi_2}{\tilde{\phi}_2} = 0$, we get
        \begin{equation}
            R = \sqrt{ \sqrt{ {\choi}(\mathcal{N}) } \widetilde{{\choi}(\mathcal{N})} \sqrt{ {\choi}(\mathcal{N}) } }
            = \sqrt{1-\gamma} \ketbra{\phi_1}{\phi_1}
        \end{equation}
        and thus the only non-zero eigenvalue is $\sqrt{1-\gamma}$ and we obtain $E_{\mathrm{oF}} = h \left( \frac{1+\sqrt{\gamma}}{2} \right)$. 
        
        Now, to check that our Kraus decomposition
        \begin{equation}
            \left\{ K^{(U)}_i\right\}_{i=1}^2 
            = \left \{ \frac{1}{\sqrt{2}} \left( 
            \begin{smallmatrix} 
                1 & \sqrt{\gamma} \\ 
                0 & \sqrt{1-\gamma} 
            \end{smallmatrix} \right), \frac{1}{\sqrt{2}} \left( 
            \begin{smallmatrix} 
                1 & -\sqrt{\gamma} \\ 
                0 & \sqrt{1-\gamma} 
            \end{smallmatrix} \right)\right \}
        \end{equation}
        reaches this value.
        Again, we have that the probabilities of the two elements are equal, and they have the same Schmidt coefficients, namely $\sqrt{\frac{1 \pm \sqrt{\gamma}}{2}}$. 
        With this, we have that 
        \begin{equation}
            E_{av} = h\left(\sqrt{\frac{1 \pm \sqrt{\gamma}}{2}}^2\right) = E_{\mathrm{oF}} \, .
        \end{equation}
    \end{proof}
    
    \begin{proposition}
        For the single-qubit depolarizing channel, the Kraus decomposition 
        \begin{equation}
            K^{(U)}_{i'} = \sum_{i=1}^{4} U_{i'i} K_{i} \quad \forall i' \in \{ 1, \dots, 4 \}
        \end{equation}
        with
        \begin{equation}
            \left\{ K_i\right\}_{i=1}^4 = \left\{ \sqrt{1-\frac{3p}{4}} \II, \frac{\sqrt{p}}{2} X, \frac{\sqrt{p}}{2} Y, \frac{\sqrt{p}}{2} Z \right\}
            \quad \text{and} \quad
            U = \frac{1}{2}
            \left( \begin{smallmatrix} 
                1 & 1 & 1 & 1 \\
                1 & 1 & -1 & -1 \\
                1 & -1 & 1 & -1 \\
                1 & -1 & -1 & 1
            \end{smallmatrix} \right)
        \end{equation}
        is optimal for the average entanglement of the Choi state below the separability threshold of $p=\frac{2}{3}$.
    \end{proposition}
    
    \begin{proof}
        A similar argument holds for depolarizing noise.
        Under our candidate unitary, all four elements of the decomposition have the same probability ($1/4$) and Schmidt coefficients. Each state then has an entanglement entropy equal to the entanglement of formation.

        Again, we start from the Choi state based on the unitary decomposition
        \begin{equation}
            {\choi}(\mathcal{N}) 
            = (1-p) \ketbra{\Phi^+}{\Phi^+} + p \frac{\II}{4} 
            = \left(1-\frac{3p}{4}\right) \ketbra{\Phi^+}{\Phi^+} + \frac{p}{4} \left( \ketbra{\Phi^-}{\Phi^-} + \ketbra{\Psi^+}{\Psi^+} + \ketbra{\Psi^-}{\Psi^-} \right) \, .
        \end{equation}
        As $ \ket{\tilde{\Phi}^\pm} = \mp \ket{\Phi^\pm}$ and $ \ket{\tilde{\Psi}^\pm} = \pm \ket{\Psi^\pm}$, 
        it holds again that $\tilde{\rho} = (\sigma_y \otimes \sigma_y) \rho^* (\sigma_y \otimes \sigma_y)=\rho$ and $R = \sqrt{ \sqrt{ \rho } \tilde{\rho} \sqrt{ \rho } } = \rho$.
        The eigenvalues can then be read out directly from the 
        decomposition of the Choi state in the Bell basis
        $\{ \lambda_i \}_i = \{ 1-3p/4, p/4, p/4, p/4 \}$
        and 
        \begin{equation}
            1-\frac{3p}{4} \geq \frac{p}{4} \ \forall p \leq 1
            \quad \Rightarrow \quad
            C({\choi}(\mathcal{N})) = 1-\frac{3p}{2} \, .
        \end{equation}
        For $p \leq \frac{2}{3}$, $C \geq 0$. 
        Then, the entanglement of formation is $E_{\mathrm{oF}} = h \left( \frac{1}{4} \left( 2 + \sqrt{3} \sqrt{(4-3p)p}\right) \right)$.

        Now to show that, for $p \leq \frac{2}{3}$, each element of the optimal decomposition has the same value of entanglement entropy. 
        From the statement above, each (sub-normalized) element of the decomposition can be written as 
        \begin{equation}
            K_i^{(U)} \ket{\Phi^+} = \frac{1}{2} \left( 
            \sqrt{1-\frac{3p}{4}} \ket{\Phi^+} 
            \pm \frac{\sqrt{p}}{2} \ket{\Psi^+} 
            \mp i \frac{\sqrt{p}}{2} \ket{\Psi^-} 
            \pm \frac{\sqrt{p}}{2} \ket{\Phi^-} 
            \right) \, .
        \end{equation}
        Then the probability of each element is 
        \begin{equation}
            p_i = \| K_i^{(U)} \ket{\Phi^+} \|_2^2 
            = \frac{1}{4} \left( 1-\frac{3p}{4} + \frac{p}{4} + \frac{p}{4} + \frac{p}{4} \right)
            = \frac{1}{4}
            \quad \forall i \, .
        \end{equation}
        The Schmidt coefficients are given by the singular values of the coefficient matrix 
        \begin{equation}
            \frac{1}{\sqrt{2}} \begin{pmatrix}
                \sqrt{1-\frac{3p}{4}} \pm \frac{\sqrt{p}}{2} & 
                \frac{\sqrt{p}}{2} (\pm 1 \mp 1) \\ 
                \frac{\sqrt{p}}{2} (\pm 1 \pm 1) & 
                \sqrt{1-\frac{3p}{4}} \mp \frac{\sqrt{p}}{2}
            \end{pmatrix}
            = \frac{K_i^{(U)}}{\left\|K_i^{(U)}\right\|_F} \, .
        \end{equation}
        All four of them share the same pair of singular values, namely 
        $\left\{ \frac{1}{2} \sqrt{ 2 + \sqrt{3} \sqrt{(4-3p)p}}, \frac{1}{2} \sqrt{ 2 - \sqrt{3} \sqrt{(4-3p)p}} \right\}$
        and with this, we obtain that
        \begin{equation}
            E_{av} = \sum_{i=1}^4 \frac{1}{4} h\left(\frac{1}{4} \left( 2 + \sqrt{3} \sqrt{(4-3p)p} \right) \right) = E_{\mathrm{oF}} \, ,
        \end{equation}
        concluding the proof.
    \end{proof}

\subsection{Optimal decomposition of amplitude damping for states with computational Schmidt basis}
\label{sec:opt_AD_computational}

    In the case of amplitude damping, one may even go further and consider the application of the channel of one of the two qubits of a state with a Schmidt decomposition in the computational basis, not necessarily a maximally entangled state.
    
    \begin{remark}
        For the single-qubit amplitude damping channel, and any two-qubit state with a Schmidt decomposition in the computational basis
        \begin{equation}
            \ket{\psi} = s \ket{0} \otimes \ket{0} + \sqrt{1-s^2} \ket{1} \otimes \ket{1}
            \quad \forall s \in [0,1]
        \end{equation}
        the Kraus decomposition 
        \begin{equation}
            K^{(U)}_{i'} = \sum_{i=1}^{2} U_{i'i} K_{i} \quad \forall i' \in \{ 1, 2 \}
            \quad \text{with} \quad
            \left\{ K_i\right\}_{i=1}^2 = \left \{ \left( \begin{smallmatrix} 1 & 0 \\ 0 & \sqrt{1-\gamma} \end{smallmatrix} \right), \left( \begin{smallmatrix} 0 & \sqrt{\gamma} \\ 0 & 0 \end{smallmatrix} \right) \right \}
            \quad \text{and} \quad
            U = \frac{1}{\sqrt{2}}
            \left( \begin{smallmatrix} 
                1 & -1 \\
               1 & 1
            \end{smallmatrix} \right) 
        \end{equation}
        induces an optimal decomposition of the mixed state after applying the channel to one of the two qubits.
    \end{remark}

    \begin{proof}
        To prove this statement, we will show that the induced decomposition (Eq.~\eqref{eq:Kraus_to_ensemble}) is the eigendecomposition of the noisy state, with this obtaining the concurrence.
        Then we show that the first two unitaries from Wootters' procedure can be chosen as the identity.
        Finally, we show that the third unitary can be chosen as the Hadamard matrix.

        We start by computing the induced decomposition
        \begin{align}
            \rho = \mathcal{N}[\ketbra{\psi}{\psi}] & 
            = (s \ket{0,0} + \sqrt{1-s^2}\sqrt{1-\gamma} \ket{1,1})(s \bra{0,0} + \sqrt{1-s^2}\sqrt{1-\gamma} \bra{1,1}) + (1-s^2) \gamma \ketbra{0,1}{0,1} \\ & 
            = (1-(1-s^2) \gamma) \ketbra{\phi_1}{\phi_1} + (1-s^2) \gamma \ketbra{\phi_2}{\phi_2}
            \nonumber
        \end{align}
        with 
        \begin{align}
            \ket{\phi_1} & 
            = \frac{1}{\sqrt{1-(1-s^2) \gamma}} (s \ket{0,0} + \sqrt{1-s^2}\sqrt{1-\gamma} \ket{1,1}), \\
            \ket{\phi_2} & 
            = \ket{0,1} \, .
        \end{align}
        The states vectors $\ket{\phi_1}$ and $\ket{\phi_2}$ are normalized and fulfill the orthonormality property $\inner{\phi_i}{\phi_j} = \delta_{i,j}$, so they are the eigenstates of the mixed state.
        Now let us apply the spin flip (Eq.~\eqref{eq:spin-flip-vector}).
        We have
        \begin{equation}
            Y \otimes Y = \left(\begin{smallmatrix}
                &&&-1 \\ &&1& \\ &1&& \\ -1&&&
            \end{smallmatrix}\right)
            \Rightarrow
            \begin{cases}
                \ket{\tilde{\phi}_1} 
                = (Y \otimes Y) \ket{\phi_1}
                = \frac{-1}{\sqrt{p_1}} (\sqrt{1-s^2}\sqrt{1-\gamma} \ket{0,0} + s \ket{1,1}) \\
                \ket{\tilde{\phi}_2}
                = \ket{1,0}
            \end{cases}
            \Rightarrow
            \inner{\tilde{\phi}_i}{\tilde{\phi}_j} = \delta_{i,j} \, .
        \end{equation}
        With this $\ket{\tilde{\phi}_i}$ and $\ket{\tilde{\phi}_j}$ are the eigenbasis of $\tilde{\rho}$.
        To compute the $R$ matrix, let us first observe some inner products
        \begin{align}
            \inner{\phi_1}{\tilde{\phi}_1} & = \frac{-2}{p_1}s\sqrt{1-s^2}\sqrt{1-\gamma} \, , \\
            \inner{\phi_2}{\tilde{\phi}_2} & = 0 \, ,\\
            \inner{\phi_2}{\tilde{\phi}_1} & = \inner{\phi_1}{\tilde{\phi}_2} = 0 \, .
        \end{align}
        Now the computation of $R^2$ becomes relatively simple, in that
        \begin{align}
            R^2 = \sqrt{\rho} \tilde{\rho} \sqrt{\rho} &
            = \left( \sqrt{1-(1-s^2) \gamma} \ketbra{\phi_1}{\phi_1} + \sqrt{1-s^2} \gamma \ketbra{\phi_2}{\phi_2} \right)
            \left( (1-(1-s^2) \gamma) \ketbra{\tilde{\phi}_1}{\tilde{\phi}_1} + (1-s^2) \gamma \ketbra{\tilde{\phi}_2}{\tilde{\phi}_2} \right) \\
            \nonumber & \qquad \qquad \hfill \times
            \left( \sqrt{1-(1-s^2) \gamma} \ketbra{\phi_1}{\phi_1} + \sqrt{1-s^2} \gamma \ketbra{\phi_2}{\phi_2} \right) 
            \\
            \nonumber& 
            = \sqrt{1-(1-s^2) \gamma} (1-(1-s^2) \gamma) \sqrt{1-(1-s^2) \gamma} \ket{\phi_1} \! \inner{\phi_1}{\tilde{\phi}_1} \inner{\tilde{\phi}_1}{\phi_1} \! \bra{\phi_1} \\
            \nonumber& 
            = (1-(1-s^2) \gamma)^2 |\ketbra{\phi_1}{\tilde{\phi}_1}|^2 \inner{\phi_1}{\phi_1} \\ & 
            = p_1^2 \frac{4}{p_1^2} s^2 (1-s^2) (1-\gamma) \inner{\phi_1}{\phi_1}
            = \left( \lambda_1^{(R)} \right)^2 \inner{\phi_1}{\phi_1} \, .
            \nonumber
        \end{align}
        Then the concurrence is simply 
        given by
        \begin{equation}
            C(\rho) = \lambda_1^{(R)} = 2 s \sqrt{1-s^2} \sqrt{1-\gamma} \, .
        \end{equation}

        Now that we know the value of the concurrence and that there is only one eigenvalue of $R$ contributing, we can turn towards the three unitaries of Wootters' procedure.
        For the first step we observe that $\inner{\phi_1}{\tilde{\phi}_1} = 
        - \frac{1}{p_1} C(\rho)$. 
        Then choosing 
        \begin{align}
            \ket{x_1} & = i K_i \ket{\psi} = i(s \ket{0,0} + \sqrt{1-s^2}\sqrt{1-\gamma} \ket{1,1}) \, , \\
            \ket{x_2} & = K_2 \ket{\psi} = \sqrt{1-s^2} \sqrt{\gamma} \ket{0,1} \, ,
        \end{align}
        we have a valid decomposition $\rho = \ketbra{x_1}{x_1} + \ketbra{x_2}{x_2}$ and that
        \begin{align}
            \inner{x_1}{\tilde{x}_1} & = (-i)^2 (-2) s \sqrt{1-s^2} \sqrt{1-\gamma} = \lambda_1^{(R)} = C(\rho) \, , \\
            \inner{x_2}{\tilde{x}_2} & = \inner{x_1}{\tilde{x}_2} = \inner{x_2}{\tilde{x}_1} = 0 \, .
        \end{align}
        The second unitary applies a phase $i$ to all elements of the first decomposition except the first.
        Observe that this would imply that the product of the first and second unitary applies a phase to all elements of the eigendecomposition
        \begin{equation}
            U_1 U_2 = i \II \, .
        \end{equation}
        Since the same phase to all Kraus operators does not change the entanglement properties, we can ignore this global phase and use $U_1 U_2 = \II$.

        For the last step, remember that we first want to compute the eigendecomposition of the real part of $M = y^\dagger \tilde{y} - C(\rho) y^\dagger y$.
        For this, we first compute 
        \begin{align}
            y^\dagger \tilde{y} & 
            = \inner{y_1}{\tilde{y}_1} \ketbra{1}{1} + \inner{y_2}{\tilde{y}_2} \ketbra{2}{2} 
            = C(\rho) \ketbra{1}{1} \\
            y^\dagger y &
            = \inner{y_1}{y_1} \ketbra{1}{1} + \inner{y_2}{y_2} \ketbra{2}{2}
            = (s^2 + (1-s^2)(1-\gamma))  \ketbra{1}{1} + (1-s^2)\gamma \ketbra{2}{2} \, .
        \end{align}
        Putting the elements of $M$ together, we can see that it is already diagonal.
        Then the orthogonal matrix $O$ diagonalizing the real part of $M$ is simply the identity. 
        The final step of our derivation was applying the Hadamard gate, which leads to our final result
        \begin{equation}
            U_{\mathrm{opt}} = \underbrace{U_\lambda}_{= \II} \underbrace{U_1 U_2}_{= \II} \underbrace{U_3}_{= \II H} \, ,
        \end{equation}
        concluding the proof.
    \end{proof}

\section{Vectorized Kraus operators and the least unitary unraveling}
\label{sec:optimizing_Kraus_projectiveness}

\subsection{Conceptual interpretation of the least unitary unraveling}

    One defining property of unitary matrices is that all their eigenvalues lie on the complex unit circle.
    Then, all their singular values are $1$.
    After renormalizing according to their Frobenius norm $\norm{U}_F = \sqrt{d}$, we obtain a uniform distribution of singular values $d^{-1/2}$.
    On the other hand, projectors have eigenvalues (and thus also singular values) of $1$ or $0$, where the number of non-zero singular values is the rank. 
    Given a rank-$r$ projector $P$, we have $\norm{P}_F = \sqrt{r}$
    and the normalized singular values $\{r^{-1/2}\}_{i=1}^r \cup \{0\}_{i=r+1}^d$.
    In particular, rank-$1$ projectors have only one non-zero singular value.

    Using the fact that the set of squared normalized singular values forms a valid probability distribution, we can try to quantify ``how unitary'' an operator is based on it. 
    As observed above, unitary operators are identified by a uniform distribution, projectors by a Heaviside function, and rank-$1$ operators by a delta function.
    We can then define a measure of the unitarity of an operator as the entropy of its squared normalized singular values, resulting in Eq.~\eqref{def:operator_unitarity}.
    Given an $d$-dimensional square operator $K$, 
    the set of its singular values $\{ \sigma_i \}_{i=1}^d$
    and its Frobenius norm $\norm{K}_F^2 = {\sum_{i=1}^d \sigma_i^2}$, we define its unitarity as
    \begin{equation}
        \operatorname{unitarity}(K) = - \sum_{i=1}^d \left( \frac{\sigma_i}{\norm{K}_F} \right)^2 \log_2 \left( \frac{\sigma_i}{\norm{K}_F} \right)^2 \, .
    \end{equation}
    Note that this definition is agnostic to the basis of the singular vectors, thus not bound to the hermicity of projectors and applicable to any operator.
    It can be seen as a unitary-invariant measure of projectiveness (or rather ``rank-$1$-ness'') versus unitarity.
    With this definition, we can also define the ensemble-averaged unitarity of a set of operators. 
    Given a set of operators $\{ K_i \}$, we define
    \begin{equation}
        \operatorname{unitarity}_{\mathrm{av}}(\{ K_i \}) = \frac{1}{d} \sum_i \norm{K_i}^2_F \operatorname{unitarity}(K_i) \, .
    \end{equation}

\subsection{From Wootters' optimal decomposition to least unitary Kraus decompositions}

    Having defined our measure of unitarity, we can see how to minimize the ensemble-averaged unitarity of a set of Kraus operators of a noise channel.
    The key is to observe that the singular values of the normalized Kraus operators are the Schmidt coefficients of the states in the decomposition obtained from applying the operator to the normalized maximally entangled state.
    Then, computing the optimal decomposition of the Choi state (with respect to the entanglement of formation) is equivalent to computing the optimal Kraus decomposition (with respect to the untarity).
    
    Let us first recall the definition of the entanglement of formation, Eq.\eqref{def:entanglement_of_formation}, in particular when phrased in terms of Schmidt coefficients of the state vectors of a decomposition.
    \begin{equation}
        E_{\mathrm{oF}}(\rho_{A:B})
        = \inf_{\substack{\{p_i, \ket{\phi_i}\} : \\ \sum_i p_i \ketbra{\phi_i}{\phi_i}=\rho}} \sum_{i} p_i E(\ket{\phi_i}_{A:B}) \, ,
    \end{equation}
    with the entanglement entropy, Eq.~\eqref{def:entanglement_entropy},
    \begin{equation}
        E(\ket{\psi}_{A:B}) 
        = -\Tr{}{\rho_A \log \rho_A}
        = - \sum_i s_i^2 \log_2 s_i^2
    \end{equation}
    for a state with Schmidt decomposition $\ket{\psi}_{A:B} = \sum_i s_i \ket{u_i}_A \otimes \ket{v_i}_B$.
    Next, the Choi state written in terms of Kraus operators is
    \begin{equation}
    \label{eq:def-Choi}
        {\choi}(\mathcal{N}) 
        = (\mathcal{N} \otimes \mathcal{I}) [\ketbra{\omega}{\omega}]
        = \sum_{i} ( K_{i} \otimes \II ) \ketbra{\omega}{\omega} ( K^{\dagger}_{i} \otimes \II ) 
        = \sum_{i} p_i \ketbra{\phi_i}{\phi_i}, 
        \quad
        \ket{\omega} = \frac{1}{\sqrt{d}} \sum_{j=0}^{d-1} \ket{jj}
        = \frac{1}{\sqrt{d}} \  \tikzineq{
            \draw [rounded corners, line width=1pt] (.4, .4) -- (0,.4) -- (0,0) -- (.4,0);
        } \, ,
    \end{equation}
    If we consider the non-normalized maximally entangled state instead, $\ket{\Omega} = \sum_{j=0}^{d-1} \ket{jj}$,
    then the states $K_i \ket{\Omega}$ are, by definition, the vectorizations of the Kraus operators.
    Informally, these states are the result of stacking the columns of $K_i$ into a large vector $\superket{K_i}$,
    \begin{alignat}{2}
        && \superket{K_i} & = \sum_{j, j' = 0}^{d-1} [K_i]_{j,j'} \ket{jj'} 
        = \tikzineq{
            \draw [black, line width=1pt, fill=none, color=black] (-2,.2) rectangle (-1.6,-.6);
            \Edge (-1.6,0)(-1.2,0)
            \Edge (-1.6,-.4)(-1.2,-.4)
            \Text[x=-.8, y=-.2]{$=$};
            \Vertex[x=0, size=.4, shape=rectangle, color=white, label=$K_i$]{Ki}
            \Edge (Ki)(.6,0)
            \draw[rounded corners, line width=1.5pt] (Ki) -- (-.5,0) -- (-.5,-.4) -- (.6,-.4);
        } \\ 
        \text{for} &&
         K_i & = \sum_{j, j' = 0}^{d-1} [K_i]_{j,j'} \ketbra{j}{j'}
         = \tikzineq{
            \Vertex[x=0, size=.4, shape=rectangle, color=white, label=$K_i$]{Ki}
            \Edge[label=$j$,position=above](Ki)(.6,0)
            \Edge[label=$j'$,position=above](Ki)(-.6,0)
        } \, .
    \end{alignat}
    With this, the Choi state in Eq.~\eqref{eq:def-Choi} can be written as 
    \begin{equation}
        {\choi}(\mathcal{N}) 
        = \frac{1}{d} \sum_{i} ( K_{i} \otimes \II ) \ketbra{\Omega}{\Omega} ( K^{\dagger}_{i} \otimes \II ) 
        = \frac{1}{d} \sum_{i} \superketbra{K_i}{K_i} \, .
    \end{equation}
    Consequently, the study of the unravelings of the Choi state (in order to find an optimal decomposition) is the study of the entanglement of the (normalized) vectorized Kraus operators of the channel.
    Let us look at this normalization. 
    First, the 2-norm of the state $\superket{K_i}$ is $ \| \superket{K_i} \|_2 = \sqrt{\sum_{j, j' = 0}^{d-1} |[K_i]_{j,j'}|^2}$.
    Also, we know that each state of the decomposition is normalized by its probability 
    \begin{equation}
        p_i = \| (K_i \otimes \II) \ket{\omega} \|^2_2 
        = \frac{1}{d} \bra{\Omega} (K_i^\dagger K_i \otimes \II)\ket{\Omega}
        = \frac{1}{d} \Tr{}{K_i^\dagger K_i} 
        = \frac{1}{d} \| K_i \|_F^2 \, .
    \end{equation}
    Through properties of matrix norms, we have that $ \| \superket{K_i} \|_2 = \sqrt{d \cdot p_i} = \| K_i \|_F$, which is the Frobenius norm of the Kraus operator.
    The elements of the decomposition of the Choi state are then
    \begin{equation}
        \ket{\phi_i} 
        = \frac{(K_i \otimes \II) \ket{\omega}}{\sqrt{p_i}} 
        = \frac{\sqrt{d}}{\norm{K_i}_F} (K_i \otimes \II) \frac{1}{\sqrt{d}} \ket{\Omega}
        = \frac{\superket{K_i}}{\| K_i \|_F}
    \end{equation}
    and their Schmidt coefficients are obtained by computing the singular values of the coefficients matrix, which is exactly $K_i/\| K_i \|_F$ after undoing the vectorization.
    This observation is even more obvious when looking at the diagrammatic representation of these quantities.
    \begin{alignat*}{2}
        && \text{Schmidt}(\superket{K_i}) &
        = \tikzineq{
            \draw [black, line width=1pt, fill=none, color=black] (-2,.2) rectangle (-1.6,-.6);
            \Edge (-1.6,0)(-1.2,0)
            \Edge (-1.6,-.4)(-1.2,-.4)
            \draw[dotted, line width=1.5pt] (-2.2, -.2) -- (-1.4, -.2);
        } 
        = \tikzineq{
            \Vertex[x=0, size=.4, shape=rectangle, color=white]{Ki}
            \Edge (Ki)(.6,0)
            \draw[rounded corners, line width=1.5pt] (Ki) -- (-.5,0) -- (-.5,-.4) -- (.6,-.4);
            \draw[dotted, line width=1.5pt] (0, .3) -- (0, -.3);
        } 
        = \tikzineq{
            \Vertex[x=0, y=.3, size=.3, shape=rectangle, color=white, label=$U$, position=90]{U}
            \Vertex[x=0, y=-.3, size=.3, shape=rectangle, color=white, label=$V^*$, position=270 ]{V}
            \Vertex[x=-.5, size=.1, shape=rectangle, color=white, label=$S$, position=-45, 
                    style={rotate=45}]{S}
            \Edge[bend=-60](U)(S)
            \Edge[bend=-60](S)(V)
            \Edge (U)(.4,.3)
            \Edge (V)(.4,-.3)
        } \\ 
        \sim &&
         \text{SVD}(K_i) &  
        = \tikzineq{
            \Vertex[x=0, size=.4, shape=rectangle, color=white]{Ki}
            \Edge (Ki)(.6,0)
            \Edge (Ki)(-.6,0)
            \draw[dotted, line width=1.5pt] (0, .3) -- (0, -.3);
        }
        = \tikzineq{
            \Vertex[x=-.5, size=.3, shape=rectangle, color=white, label=$V^\dagger$, position=90 ]{V}
            \Vertex[x=.5, size=.3, shape=rectangle, color=white, label=$U$, position=90]{U}
            \Vertex[x=0, size=.1, shape=rectangle, color=white, label=$S$, position=-135, 
                    style={rotate=45}]{lambda}
            \Edge (V)(U)
            \Edge (U)(.9,0)
            \Edge (V)(-.9,0)
        } \, .
    \end{alignat*}
    The vectorization is then more of a tool to gain physical intuition (and be allowed to use 
    the vocabulary established for quantum states) and does not 
    need to be performed in practice.

    We now have that the Schmidt coefficients of the state vectors of the decomposition of the Choi state are the normalized singular values of the corresponding Kraus operator
    \begin{equation}
        s_j (\ket{\phi_i}) 
        = \sigma_j\left(\frac{K_i}{\norm{K_i}_F}\right)
        = \frac{\sigma_j(K_i)}{\norm{K_i}_F} \, .
    \end{equation}
    Note now that the entanglement entropy of each one of the state vectors is the unitarity of the respective Kraus operator
    \begin{equation}
        E(\ket{\phi_i}) 
        = - \sum_j (\ket{\phi_i})^2 \log_2 (\ket{\phi_i})^2
        = - \sum_j \left(\frac{\sigma_j(K_i)}{\norm{K_i}_F}\right)^2 \log_2 \left(\frac{\sigma_j(K_i)}{\norm{K_i}_F}\right)^2
        = \operatorname{unitarity}(K_i)
    \end{equation}
    and the average quantities also match
    \begin{equation}
        E_{\mathrm{av}}(\{ p_i, \ket{\phi_i} \}) 
        = \sum_{i} p_i E(\ket{\phi_i}) 
        = \sum_i \frac{\norm{K_i}^2_F}{d} \operatorname{unitarity}(K_i)
        = \operatorname{unitarity}_{\mathrm{av}}(\{ K_i \}) \, .
    \end{equation}
    With this, finding the 
    least unitary decomposition of the channel is the same as finding the optimal decomposition of the Choi state of the channel. 
    For single-qubit channels, this can be done by applying the procedure described in Section~\ref{sec:optimal_unitary} to the noisy Bell state.
    Thus, the ensemble-averaged entanglement entropy of the Choi state is maximized when the Kraus operators are (proportional to) unitary matrices, and is minimized when they are as ``non-unitary'' as possible, which in our definition means being as close to rank-$1$ operators as possible.

\section{Exact decomposition of Trotterized Lindbladian evolutions}

    In this section, we will show how to construct a (noisy) quantum circuit simulating open system dynamics described by Lindbladian equations. 
    We first Trotterize the vectorized superoperator 
    to obtain local terms and then show how to obtain the unitaries of the circuit and an exact Kraus decomposition of the noise channels induced by the jump processes (Section~\ref{sec:Lindblad_to_circuits}).
    For two typical choices of jump operators, we also show that they exactly result in the well-known noise channels of dephasing and amplitude damping (Section~\ref{sec:noise_jump_equivalence}).
    
\subsection{From Trotterized Lindbladian evolutions to noisy circuits}
\label{sec:Lindblad_to_circuits}
    
    The Lindbladian description of the evolution of open quantum systems is given by
    \begin{equation}
        \label{eq:def-Lindbladian}
        \frac{\partial \rho}{\partial t} 
        = -\ii[H, \rho] + \sum_{j} \gamma_j \left( c_j \rho c_j^\dagger -\frac{1}{2} \{ c_j^\dagger c_j, \rho \}\right) 
        = \mathcal{L}[\rho] \, .
    \end{equation}
    Here, $\mathcal{L}[\cdot]$ is the superoperator that returns its time derivative when applied to the density matrix, giving rise to a quantum dynamical semi-group.
    One can vectorize Eq.~\eqref{eq:def-Lindbladian} such that a vector represents the density matrix, and the superoperator can be written as a matrix. 
    Then, the evolution is given by the matrix-vector product $\partial_t \superket{\rho} = L \superket{\rho}$, where
    \begin{alignat}{3}
        && \superket{\rho} & 
        = (\rho \otimes \II) \ket{\Omega}, 
        \qquad 
        \ket{\Omega} 
        = \sum_i \ket{i,i} \\
        \text{and} &&
        L & 
        = -\ii(H \otimes \II - \II \otimes H^\top)
        + \sum_{j} \gamma_j \left( c_j \otimes c_j^* -\frac{1}{2} c_j^\dagger c_j \otimes \II -\frac{1}{2} \II \otimes c_j^\top c_j^* \right) \, .
    \end{alignat}
    Assuming a time-independent evolution, the solution of the differential equation is $\superket{\rho(t)} = \ee^{Lt} \superket{\rho(0)}$.
    The matrix $\ee^{Lt}$ is a global operator 
    (and thus a $4^n \times 4^n$ matrix), but it can be decomposed into local operations by Trotterization. 
    We first rewrite the exponential encoding of the evolution at time $t$ as the product of evolutions of time $dt$, i.e.,
    \begin{equation}
        \ee^{Lt} = \prod_{\tau=1}^{t/dt} \ee^{Ldt} \, .
    \end{equation}
    Then, the exponential of sums of terms can be written as a product of exponentials, which is accurate up to quadratic order in $dt$.
    In particular, this allows us to separate the coherent and the incoherent components of the evolution
    \begin{equation}
        \ee^{Ldt} = \ee^{-\ii(H \otimes \II - \II \otimes H^\top) dt} {\prod_j} \ee^{\gamma_j \left( c_j \otimes c_j^* -\frac{1}{2} c_j^\dagger c_j \otimes \II -\frac{1}{2} \II \otimes c_j^\top c_j^* \right)dt} + O(dt^2) \, .
    \end{equation}
    Assuming that the Hamiltonian is composed of two-local, nearest-neighbor terms 
    \begin{equation}
        H = \sum_{i=1}^n h^{(2)}_{i, i+1} + \sum_{i=1}^n h^{(1)}_{i} \, ,
    \end{equation}
    with $ h^{(2)}_{i, i+1} \in \CC^{4 \times 4}$, $h^{(1)}_{i} \in \CC^{2 \times 2}$ 
    and that the jump operators are acting on single sites ($c_j \in \CC^{2 \times 2}$),
    the time evolution can be Trotterized into a brickwork circuit structure with local noise. 
    Let us define the Hamiltonian terms acting on neighboring pairs of qubits as
    \begin{equation}
        H_{i, i+1} = h^{(2)}_{i, i+1} + \frac{1}{2} h^{(1)}_{i} \otimes \II + \frac{1}{2} \II \otimes h^{(1)}_{i+1}
        \quad \Rightarrow \quad 
        H = \sum_{i \text{ odd}} H_{i, i+1} + \sum_{i \text{ even}} H_{i, i+1}
    \end{equation}
    and the corresponding vectorized superoperators as
    \begin{align}
        L_{i, i+1} & = -\ii \left( H_{i, i+1} \otimes \II - \II \otimes H^{\top}_{i, i+1} \right), \\
        L^{\mathrm{jump}}_j & = \gamma_j \left( c_j \otimes c_j^* -\frac{1}{2} c_j^\dagger c_j \otimes \II -\frac{1}{2} \II \otimes c_j^\top c_j^* \right) \, .
    \end{align}
    Then, we can write
    \begin{equation}
        \ee^{Ldt} 
        = 
        \left( \bigotimes_{i \text{ odd}} \ee^{ L_{i, i+1} dt} \right) 
        \left( \bigotimes_{j=1}^n \ee^{L^{\mathrm{jump}}_j dt/2} \right) 
        \left( \bigotimes_{i \text{ even}} \ee^{ L_{i, i+1} dt} \right) 
        \left( \bigotimes_{j=1}^n \ee^{L^{\mathrm{jump}}_j dt/2} \right) 
        + O(dt^2) \, .
    \end{equation}
    The result is then a layer of coherent evolution given by local Hamiltonians acting on the odd pairs of neighboring qubits, followed by a layer of single-site incoherent processes and then a layer of Hamiltonian evolution on the even pairs of qubits, again followed by a layer of local incoherent processes. 
    This fits nicely in the brickwork layout, with noise typically considered in one-dimensional noisy quantum circuits.
    One could also use higher-order Trotter formulas (in particular the second order, which comes at no additional computational cost), but as it is not the focus of this work, we leave it to the reader.
    Next, we will explain what the respective unitary evolutions and noise processes look like.

    For the coherent part, one directly retrieves the usual unitary evolution generated by a Hamiltonian. 
    Take
    \begin{align}
        \ee^{L^{\mathrm{coh}} dt} & 
        = \ee^{-\ii(H \otimes \II - \II \otimes H^\top) dt}
        = \ee^{-\ii(H \otimes \II)dt} \ee^{\ii(\II \otimes H^\top) dt} \\ & 
        = (\ee^{-\ii H dt} \otimes \II) (\II \otimes \ee^{\ii H^\top dt})
        = \ee^{-\ii H dt} \otimes \ee^{\ii H^\top dt} \, ,
        \nonumber
    \end{align}
    it is the vectorization of the superoperator $\mathcal{U}[\cdot] = \ee^{-\ii H dt} (\cdot) (\ee^{-\ii H dt})^\dagger$. 
    
    The Kraus operators of the incoherent 
    part can be obtained by computing the eigendecomposition of the Choi state of the corresponding map. 
    In tensor network language, this translates to a certain eigendecomposition of the vectorized four-legged jump tensor.
    Reshaping to group the indices of the first and second qubits, respectively
    \begin{equation}
        \operatorname{reshape}(\ee^{L_j dt}) = (\bra{\Omega} \otimes \II \otimes \II) ( \II \otimes \ee^{L_j dt} \otimes \II) (\II \otimes \II \otimes \ket{\Omega}) \, ,
    \end{equation}
    the eigendecomposition in terms of subnormalized eigenvectors is given by
    \begin{equation}
        \operatorname{reshape}(\ee^{L_j dt}) = \sum_i \superketbra{K_i}{K_i^*} \, ,
    \end{equation}
    where $\superket{K_i} = (K_i \otimes \II) \ket{\Omega}$ are the vectorized Kraus operators.
    NThe reshaped tensor is indeed the Choi operator of the same quantum channel; thus, we reformulated the well-established way of computing Kraus decompositions of channels. 

\subsection{Noise model for some typical jump operators}
\label{sec:noise_jump_equivalence}

    For some typical choices of jump operators, the resulting channels can be identified as known noise channels.
    We will see two examples.
    
    \begin{remark}
        The noise channel resulting from a Lindbladian jump evolution with jump operator $c = \left( \begin{smallmatrix}
            1 & 0 \\ 0 & 0
        \end{smallmatrix}\right)$, 
        jump rate $\gamma$ and Trotterization time step $dt$ is a dephasing channel with rate $p^{\mathrm{DF}} = 1 - \ee^{-\gamma dt/ 2}$. 
        Furthermore, the decomposition obtained is the typical decomposition in terms of $\II$ and $Z$.
    \end{remark}

    \begin{proof}
        We start with some properties of the jump operator
        \begin{equation}
            c = \left( \begin{smallmatrix}
                1 & 0 \\ 0 & 0
            \end{smallmatrix}\right) = \ketbra{0}{0}
            \quad \Rightarrow \quad
            c^\dagger = c^\top = c^* = c^\dagger c = c \, .
        \end{equation}
        From this, the vectorized Lindbladian superoperator is
        \begin{align}
            L & 
            = \gamma \left( c \otimes c^* -\frac{1}{2} c^\dagger c \otimes \II -\frac{1}{2} \II \otimes c^\top c^* \right) \\
            \nonumber 
            & 
            = \gamma \left( \ketbra{0}{0} \otimes \ketbra{0}{0} -\frac{1}{2} \ketbra{0}{0} \otimes \II -\frac{1}{2} \II \otimes \ketbra{0}{0} \right) \\
            \nonumber& 
            = \gamma \left( \ketbra{0}{0} \otimes \ketbra{0}{0} -\frac{1}{2} \ketbra{0}{0} \otimes \ketbra{0}{0} -\frac{1}{2} \ketbra{0}{0} \otimes \ketbra{1}{1} - \frac{1}{2} \ketbra{0}{0} \otimes \ketbra{0}{0} - \frac{1}{2} \ketbra{1}{1} \otimes \ketbra{0}{0}\right) \\
            \nonumber& 
            = -\frac{\gamma}{2} \left( \ketbra{0}{0} \otimes \ketbra{1}{1} + \ketbra{1}{1} \otimes \ketbra{0}{0}\right) + 0 \left( \ketbra{0}{0} \otimes \ketbra{0}{0} + \ketbra{1}{1} \otimes \ketbra{1}{1}\right) \, .
            \nonumber
        \end{align}
        The last equality is already the eigendecomposition 
        of $L$.
        Then, the single time step superoperator can be computed by the exponentiation of the eigenvalues (note not to forget the $0$ eigenvalues)
        \begin{equation}
            \ee^{Ldt} = \ee^{-\frac{\gamma dt}{2}} \left( \ketbra{0}{0} \otimes \ketbra{1}{1} + \ketbra{1}{1} \otimes \ketbra{0}{0}\right) + \ee^0 \left( \ketbra{0}{0} \otimes \ketbra{0}{0} + \ketbra{1}{1} \otimes \ketbra{1}{1}\right).
        \end{equation}
        Further reshaping the tensor results in
        \begin{align}
            \operatorname{reshape}(\ee^{L}) & 
            = \ee^{-\frac{\gamma dt}{2}} \left( \ketbra{0,0}{1,1} + \ketbra{1,1}{0,0}\right) + \ee^0 \left( \ketbra{0,0}{0,0} + \ketbra{1,1}{1,1}\right) \\ 
            \nonumber & 
            = (1 + \ee^{-\gamma dt/ 2}) \ketbra{\Phi^+}{\Phi^+} + (1 - \ee^{-\gamma dt / 2}) \ketbra{\Phi^-}{\Phi^-} \, .
        \end{align}
        Since $\ket{\Phi^+} = \frac{1}{\sqrt{2}}\superket{\II}$ and $\ket{\Phi^-} = \frac{1}{\sqrt{2}}\superket{Z}$, the retrieved Kraus operators are given by
        \begin{equation}
            \left\{ K_i \right\}_i = \left\{ \sqrt{\frac{1 + \ee^{-\gamma dt/ 2}}{2}} \II, \sqrt{\frac{1 - \ee^{-\gamma dt/ 2}}{2}} Z \right\} \, ,
        \end{equation}
        that are the same Kraus operators as for the dephasing channel $\left \{ \sqrt{1-\frac{p^{\mathrm{DF}}}{2}} \II, \sqrt{\frac{p^{\mathrm{DF}}}{2}} Z \right \}$, with $p^{\mathrm{DF}} = 1 - \ee^{-\gamma dt/ 2}$.
    \end{proof}

    \begin{remark}
        The noise channel resulting from a Lindbladian jump evolution with jump operator $c = \left( \begin{smallmatrix}
            0 & 1 \\ 0 & 0
        \end{smallmatrix}\right)$ 
        with jump rate $\gamma$ and Trotterization time step $dt$ is an amplitude damping channel with noise-rate $\gamma^{\mathrm{AD}} = 1 - \ee^{- \gamma dt}$. 
        Furthermore, the decomposition obtained is the typical decomposition in terms of 
        $\left( \begin{smallmatrix} 
            1 & 0 \\ 0 & \sqrt{1-\gamma^{\mathrm{AD}}}
        \end{smallmatrix} \right)$ and 
        $\left( \begin{smallmatrix} 
            0 & \sqrt{\gamma^{\mathrm{AD}}} \\ 0 & 0
        \end{smallmatrix} \right)$.
    \end{remark}

    \begin{proof}
        For $c = \left( \begin{smallmatrix} 1 & 0 \\ 0 & 0 \end{smallmatrix} \right)$, we have
        \begin{equation}
            c^* = c = \ketbra{0}{1}
            \quad \text{and} \quad
            c^\dagger c = (c^\dagger c)^\top = \ketbra{1}{1} \, 
        \end{equation}
        and thus
        \begin{align}
            L & 
            = \gamma \left( c \otimes c^* - \frac{1}{2} c^\dagger c \otimes \II - \frac{1}{2} \II \otimes c^\top c^* \right) \\ 
            \nonumber
            & 
            = \gamma \left( \ketbra{0}{1} \otimes \ketbra{0}{1} - \frac{1}{2} \ketbra{1}{1} \otimes \II - \frac{1}{2} \II \otimes \ketbra{1}{1} \right) \\ 
            \nonumber& 
            = \gamma \ketbra{0}{1} \otimes \ketbra{0}{1}
            - \frac{\gamma}{2} \left( \ketbra{1}{1} \otimes \ketbra{0}{0} + \ketbra{1}{1} \otimes \ketbra{1}{1} + \ketbra{0}{0} \otimes \ketbra{1}{1} + \ketbra{1}{1} \otimes \ketbra{1}{1} \right) 
            \nonumber
        \end{align}
        and 
        \begin{equation}
            \ee^{Ldt} = \ketbra{0,0}{0,0} + \ee^{-\gamma dt/2} \ketbra{0,1}{0,1} + \ee^{-\gamma dt/2} \ketbra{1,0}{1,0} + \ee^{-\gamma dt} \ketbra{1,1}{1,1} + (1-\ee^{-\gamma dt}) \ketbra{0,0}{1,1} \, .
        \end{equation}
        This can be obtained through symbolic computing (e.g., with Mathematica) or diagonalization 
        (note that the eigenvectors are not orthogonal in this case).
        With this, the reshaping results in 
        \begin{equation}
            \operatorname{reshape}(\ee^{L}) 
            = \ketbra{0,0}{0,0} + \ee^{-\gamma dt/2} \ketbra{0,0}{1,1} + \ee^{-\gamma dt/2} \ketbra{1,1}{0,0} + \ee^{-\gamma dt} \ketbra{1,1}{1,1} + (1-\ee^{-\gamma dt}) \ketbra{0,1}{0,1} ,
        \end{equation}
        which has eigenvalues 
        $\{ 1-\ee^{-\gamma dt}, 1+\ee^{-\gamma dt}, 0, 0\}$ 
        and eigenvectors
        \begin{equation}
            \left\{ \ket{0,1}, \frac{1}{\sqrt{1+\ee^{\gamma dt}}}(\ee^{\gamma dt/2} \ket{0,0} + \ket{1,1}), \ket{1,0}, \frac{1}{\sqrt{1+\ee^{-\gamma dt}}}(\ee^{-\gamma dt/2} \ket{0,0} + \ket{1,1}) \right\}.
        \end{equation}
        These result in the Kraus operators of the noise channel being given by
        \begin{equation}
            \left\{ 
            \sqrt{1-\ee^{-\gamma dt}} \left( \begin{smallmatrix}
                0 & 1 \\ 0 & 0
            \end{smallmatrix} \right), 
            \sqrt{\frac{1+\ee^{-\gamma dt}}{1+\ee^{\gamma dt}}}  \left( \begin{smallmatrix}
                \ee^{\gamma dt/2} & 0 \\ 0 & 1
            \end{smallmatrix} \right)
            \right\}
            = \left\{ 
            \left( \begin{smallmatrix}
                0 & \sqrt{\gamma^{\mathrm{AD}}} \\ 0 & 0
            \end{smallmatrix} \right), 
            \left( \begin{smallmatrix}
                1 & 0 \\ 0 & \sqrt{1-\gamma^{\mathrm{AD}}}
            \end{smallmatrix} \right)
            \right\}
        \end{equation}
        with $\gamma^{\mathrm{AD}} = 1-\ee^{-\gamma dt}$, 
        which can be seen after using that $\frac{1+\ee^{-\gamma dt}}{1+\ee^{\gamma dt}} = \ee^{-\gamma dt} $.
    \end{proof}

\section{Additional numerical results}
\label{sec:auxiliary_numerics}

    In Fig.~\ref{fig:bad_random_circuits} we show further results of numerical simulations where the {Haar Optimal} unraveling performs similarly to our locally entanglement-optimal unraveling.
    
    \begin{figure}
        \centering
        \includegraphics[width=\linewidth]{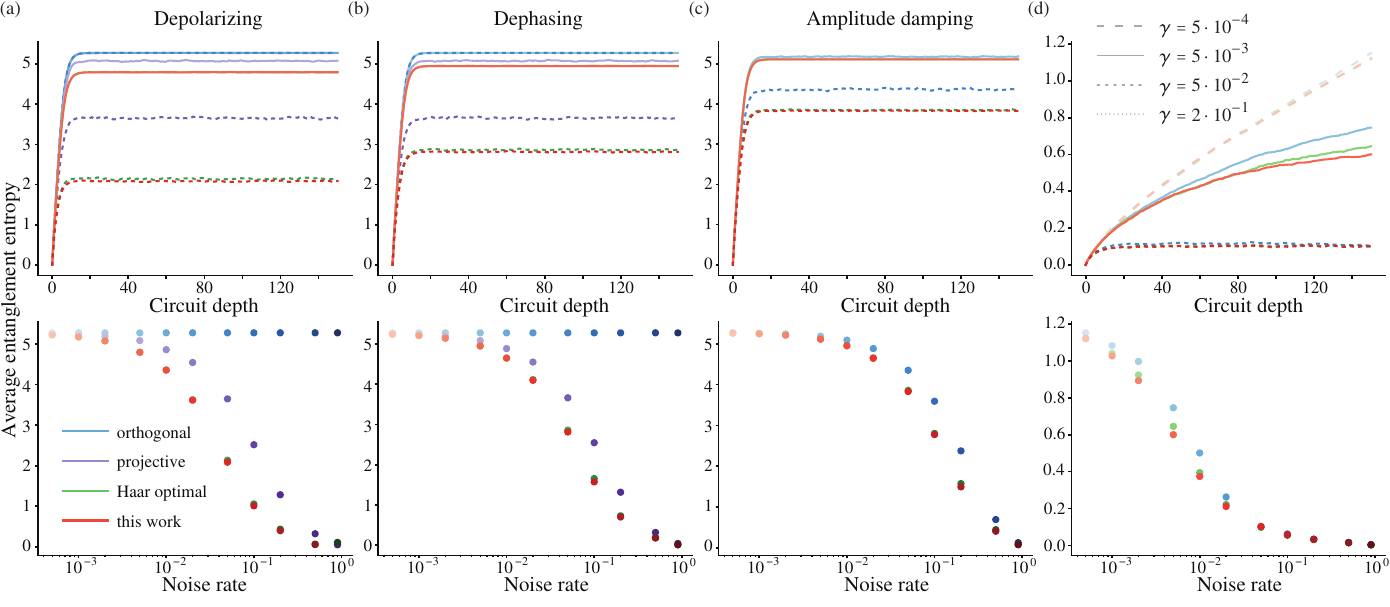}
        \caption{
            Simulation of the average entanglement entropy of states evolving under noisy random circuits. 
            Average entanglement entropy with respect to circuit depth (top row) and average entanglement entropy at depth ($150$) with respect to noise rate (bottom row).
            Brickwork circuits composed of independently sampled Haar random two-qubit gates with depolarizing noise (a), dephasing noise (b) or amplitude damping noise (c) 
            or brickwork circuit of low-entangling random gates without local rotation under amplitude damping noise (d).
        }
        \label{fig:bad_random_circuits}
    \end{figure}

\section{Greediness and effect of local bases}
\label{sec:greediness_discussion}

    \begin{figure}[h!]
        \centering
        \includegraphics[width=\linewidth]{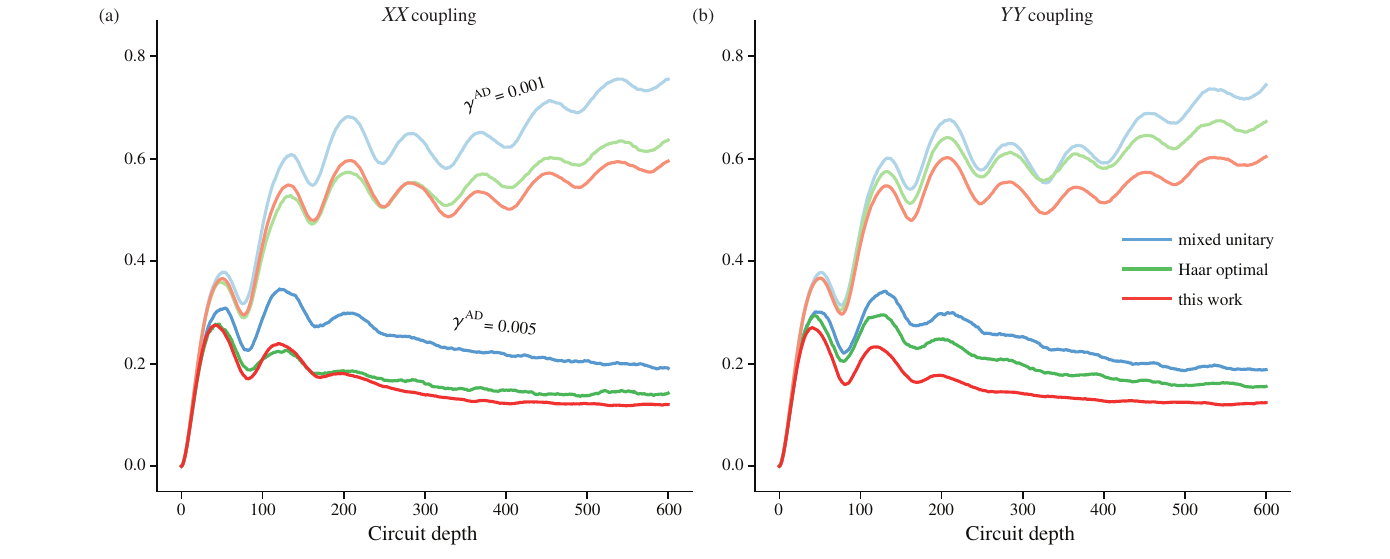}
        \caption{
            Comparison of the simulation of two Hamiltonian when switching Pauli $X$ with Pauli $Y$ operators. 
            Simulations for $H = \sum_i X_iX_{i+1} + a X_i + b Y_i + c Z_i$ (a) or $H = \sum_i Y_iY_{i+1} + b X_i + a Y_i + c Z_i$ (b) under amplitude damping noise, initialized in the $0$ state, with a time step of $0.01$, for various noise rates and three unravelings. 
        }
        \label{fig:XvsY}
    \end{figure}
    
    In this section we discuss how the remaining freedom of choice within the set of ``optimal'' unravelings might have an effect on subsequent unitary evolutions and noise unravelings. 
    In particular, we use the example of the evolution under a Hamiltonian under amplitude damping noise, where all Pauli $X$ and $Y$ operators are switched. 
    Initializing the state in the $0$ state and applying amplitude damping noise isolates the $Z$ axis of the Bloch sphere as special, but the symmetry makes the $X$ and $Y$ axes merely a definition choice. 
    Under this observation, the simulations should be identical (up to statistical fluctuations of the sampling) if replacing all $X$ operators in the Hamiltonian with $Y$ and vice-versa. 
    As seen in Fig.~\ref{fig:XvsY} however, this is not the case. 
    The {Orthogonal} and optimized unravelings show similar evolutions in both settings, but the fixed {Haar Optimal} one performs significantly better for the Hamiltonian with $XX$ coupling compared to the $YY$ coupling. 
    This could be due to the fact that the fixed {Haar Optimal} unraveling is one of the possible ``least unitary unravelings'' (see Section~\ref{sec:least_unitary_unraveling}), but not the unique one. 
    Indeed derivations of optimal decompositions do in some cases leave a certain freedom of choice. 
    Take the {Haar Optimal} decomposition of amplitude damping 
    \begin{equation}
        \left\{ K_i\right\}_{i=1}^2 = \left \{ \left( \begin{smallmatrix} 1 & 0 \\ 0 & \sqrt{1-\gamma} \end{smallmatrix} \right), \left( \begin{smallmatrix} 0 & \sqrt{\gamma} \\ 0 & 0 \end{smallmatrix} \right) \right \}
        \quad 
        \xrightarrow{ U = \frac{1}{\sqrt{2}}
        \left( \begin{smallmatrix} 
           1 & 1 \\
           1 & -1
        \end{smallmatrix} \right)}
        \quad 
        \left\{ K_i^{(U)}\right\}_{i=1}^2 = \left \{ \frac{1}{\sqrt{2}} \left( \begin{smallmatrix} 1 & \sqrt{\gamma} \\ 0 & \sqrt{1-\gamma} \end{smallmatrix} \right), \frac{1}{\sqrt{2}} \left( \begin{smallmatrix} 1 & -\sqrt{\gamma} \\ 0 & \sqrt{1-\gamma} \end{smallmatrix} \right) \right \} \, .
    \end{equation}
    One can add a phase to the unitary, resulting in a different optimal decomposition
    \begin{equation}
        \left\{ K_i\right\}_{i=1}^2 = \left \{ \left( \begin{smallmatrix} 1 & 0 \\ 0 & \sqrt{1-\gamma} \end{smallmatrix} \right), \left( \begin{smallmatrix} 0 & \sqrt{\gamma} \\ 0 & 0 \end{smallmatrix} \right) \right \}
        \quad 
        \xrightarrow{ U = \frac{1}{\sqrt{2}}
        \left( \begin{smallmatrix} 
           1 & 1 \\
           \ii & -\ii
        \end{smallmatrix} \right)}
        \quad 
        \left\{ K_i^{(U)}\right\}_{i=1}^2 = \left \{ \frac{1}{\sqrt{2}} \left( \begin{smallmatrix} 1 & \ii\sqrt{\gamma} \\ 0 & \sqrt{1-\gamma} \end{smallmatrix} \right), \frac{1}{\sqrt{2}} \left( \begin{smallmatrix} 1 & -\ii\sqrt{\gamma} \\ 0 & \sqrt{1-\gamma} \end{smallmatrix} \right) \right \} \, .
    \end{equation}
    This freedom could in principle also be optimized over.

    Now, let us relate this to the greediness of the optimization. 
    For the unraveling of a single time step, both unravelings are equivalent as they are ``equally unitary''. 
    However, let us assume one can find two such unravelings, such that the first projects more on the $X$ basis while the other projects more on the $Y$ basis. 
    Then, the sampled state from such unravelings is the input state of the following unitary layer. 
    Under an evolution with $XX$ coupling, a state which has been projected on the eigenbasis of the $X$ operator will not generate entanglement, while a state projected on the eigenbasis of the $Y$ operator will.
    This could explain why the fixed {Haar Optimal} unraveling seems to have a preferred basis.

\stopcontents[app]

\end{document}